\documentclass[12pt, oneside]{article}
\usepackage[margin=0.85in]{geometry}
\usepackage{graphicx}
\usepackage{amssymb}
\usepackage{amsmath}
\usepackage{mathtools}
\usepackage{setspace}
\usepackage{amsthm}
\usepackage{color}
\usepackage[svgnames]{xcolor}
\usepackage{booktabs}
\usepackage{array} 
\usepackage{caption}
\usepackage[authoryear,round,numbers]{natbib}
\usepackage{diagbox}
\usepackage{enumitem}
\usepackage[]{todonotes}
\usepackage{comment}
\usepackage{dsfont}

\usepackage[colorlinks]{hyperref}
\usepackage[capitalize,nameinlink]{cleveref}

\usepackage[svgnames]{xcolor}

\hypersetup{    
	citecolor=DarkBlue,
	bookmarksnumbered=true,
	urlcolor=Indigo,linkcolor=DarkBlue
}

\usepackage{aliascnt} 
\usepackage{color}
\usepackage{natbib}
\usepackage{dsfont}

\usepackage{subfigure}

\usepackage{pgf}
\usepackage{tikz}
\usetikzlibrary{decorations.markings}
\usetikzlibrary{patterns,hobby}
\usetikzlibrary{shapes.misc, positioning}

\usepackage{setspace}
\newtheorem*{thmbw}{Blackwell's Theorem}

\theoremstyle{plain}
\newtheorem{thm}{Theorem}
\newtheorem*{thm*}{Theorem}

\newtheorem{lem}{Lemma}

\newtheorem{prop}{Proposition}

\newtheorem{cor}{Corollary}

\newtheorem{defn}{Definition}

\newtheorem*{obs}{Observation}
\newtheorem{ex}{Example}

\theoremstyle{remark}

\crefname{thm}{Theorem}{Theorems}

\newcommand{\bu}{\mathbf{u}}

\newcommand{\bv}{{\boldsymbol v}}
\newcommand{\btau}{{\boldsymbol \tau}}

\newcommand{\po}{\mathcal{H}}

\newcommand{\epi}{\text{epi}}
\newcommand{\ext}{\text{ext}}
\newcommand{\co}{\text{co}}
\newcommand{\argmax}{\mathop{\mathrm{argmax}}}
\newcommand{\argmin}{\mathop{\mathrm{argmin}}}

\definecolor{Rust}{RGB}{180,30,43}

\newcommand{\term}[1]{\textcolor{Rust}{#1}}
\newcommand{\blue}[1]{\textcolor{blue}{#1}}

\begin{document}
 
\title{Robust Aggregation of Correlated Information\footnote{This paper has previously been circulated under the title ``Robust Merging of Information.'' For helpful comments that improved the paper, we thank the editor, three anonymous referees, Nageeb Ali, Alex Bloedel, Kfir Eliaz, Marcos Fernandes, Alex Frankel, Marc Henry, Nicole Immorlica, Jan Knoepfle, Asen Kochov, Jiangtao Li, Elliot Lipnowski, Ce Liu, George Mailath, Paula Onuchic, Pietro Ortoleva, Collin Raymond, Fedor Sandomirskiy, Ludvig Sinander, Shamim Sinnar, Ran Spiegler, Rui Tang, Rakesh Vohra, Leeat Yariv, and participants at various conferences and seminars. We also thank Tiago Botelho for his excellent research assistance.}}

\author{Henrique de Oliveira\thanks{S\~{a}o Paulo School of Economics - FGV. Email: \href{mailto:henrique.oliveira@fgv.br}{henrique.oliveira@fgv.br} }  \and Yuhta Ishii\thanks{ Pennsylvania State University. Email: \href{mailto:yxi5014@psu.edu}{yxi5014@psu.edu} } \and Xiao Lin\thanks{University of Pennsylvania. Email: \href{mailto:xiaolin7@sas.upenn.edu}{xiaolin7@sas.upenn.edu}}
}
\date{September 2026}

\maketitle

\begin{abstract}

Decisions are typically made by drawing on multiple sources of information, yet the correlation among these sources is often hard to assess. We study robustly optimal strategies\,---\,those that maximize the worst-case payoff guarantee across all correlation structures. With two states and two actions, the robustly optimal strategy is strikingly simple: use only the best information source and ignore the rest. With more actions, a robustly optimal strategy decomposes the problem into local comparisons and assigns each comparison to the source that is best at that comparison. With more states, the same specialization logic survives through a concavification characterization: the decision problem is decomposed into payoff-relevant components, and each component is assigned to its best source. These results offer a new rationale for why ignoring information can be optimal and highlight the benefit of using specialized information sources when correlations are ambiguous.

\end{abstract}

\section{Introduction}

Most decisions in life, from the mundane to the important, rely on multiple correlated information sources. For example, treatment decisions can be made by consulting multiple doctors, who base their recommendations on the same test results; investment plans are often informed by the advice of many financial analysts, who may have incentives to provide contrarian opinions. Yet, assessing how information sources are correlated is hard, due to the {curse of dimensionality}.  Misinterpreting correlations can lead to flawed inferences and suboptimal decisions. Therefore, an agent may look to make decisions that are robust against potential misspecification of correlation.

In this paper, we study agents who understand each information source individually but lack knowledge about their correlations, and look for a strategy that yields the best payoff guarantee against all possible correlation structures. Our main results characterize how agents can restrict attention to a selective set of specialized information sources.

As an example, consider the following scenario: The Centers for Disease Control (CDC) is setting guidelines for administering a new Covid treatment to a patient population with equal prior probabilities of having either Covid or the Flu. The treatment is designed for Covid, so it is beneficial for Covid patients, but causes only side effects for those with the Flu. The payoff matrix is given in \cref{table:Covid-1}, where the payoffs from no treatment are normalized to zero. 

\begin{table}[htp]
\centering
\begin{tabular}{|c|c|c|}
\hline
 & Treatment & No Treatment\\
\hline
 Covid &  \phantom{-}30 & 0 \\
\hline
 Flu & -20 & 0 \\
\hline
\end{tabular}
\caption{Payoffs from the Treatment}
\label{table:Covid-1}
\end{table}

Since patients with different diseases may develop different symptoms with different probabilities, these symptoms can serve as informative signals to guide treatment decisions. Suppose there are two well-understood studies: one describes the relationship between the diseases (Covid/Flu) and the Cough symptom; the other describes the relationship between the diseases and the Fever symptom. These relationships, represented as Blackwell experiments, are shown in \cref{table:Covid-2}, where ``$+$'' and ``$-$'' denote the presence and absence of a symptom, and the numbers specify their respective probabilities conditional on the disease.  

\begin{table}[htp]
\begin{minipage}{0.5\textwidth}
\centering
\begin{tabular}{|c|c|c|}
\hline
 & $+$ & $-$ \\
\hline
Covid & $0.9$ & $0.1$ \\
\hline
Flu & $0.5$ & $0.5$ \\
\hline
\end{tabular}
\captionof*{table}{Cough}
\end{minipage}
\hfill 
\begin{minipage}{0.5\textwidth}
\centering
\begin{tabular}{|c|c|c|}
\hline
 & $+$ & $-$ \\
\hline
Covid & $0.5$ & $0.5$ \\
\hline
Flu & $0.9$ & $0.1$ \\
\hline
\end{tabular}
\captionof*{table}{Fever}
\end{minipage}
\caption{Known Relationships between Diseases and Symptoms}
\label{table:Covid-2}
\end{table}

However, no studies have jointly examined both Cough and Fever symptoms, so the CDC does not know their joint probabilities. Without this knowledge, a strategy that best responds to a misspecified correlation structure can perform poorly, so the goal is to design a treatment guideline that utilizes the available information in a way that is robust to all possible correlations.

A simple strategy that guards against the unknown correlation is to base the treatment decision on only one symptom. If using only the Cough symptom, the treatment should be administered if and only if the patient has a positive Cough symptom. This strategy guarantees a value of $\frac{1}{2}[0.9\times 30 +0.5\times (-20)]=8.5$ regardless of the correlation. Similarly, the CDC could also base the treatment decision solely on the Fever symptom, which guarantees a value of $\frac{1}{2}[0.5\times 30 +0.1\times (-20)]=6.5$. Since the strategy using the Cough symptom guarantees a higher value, we call it a \textit{best-source strategy}, which selects a \textit{single} information source\,---\,the best one when considered individually\,---\,and best responds to it.

While the best-source strategy has the virtue of being simple, it completely forfeits the potential benefits from observing multiple information sources. Could the CDC do better by using a more sophisticated treatment strategy that incorporates both symptoms? \autoref{thm:binarybinary} shows that the answer is {no}. In fact, it shows that {a best-source strategy is always robustly optimal in any decision problem involving two states and two actions}. Moreover, whenever the best information source is unique, e.g. the Cough symptom in this example, the best-source strategy is the unique robustly optimal strategy.

With more than two actions, best-source strategies are no longer always optimal, and robustly optimal strategies will typically use multiple information sources.  Nevertheless, best-source strategies serve as the basic building blocks from which robustly optimal strategies are constructed. To illustrate, let us revise the example and suppose now there are two treatments: one is the previous treatment, designed for Covid, and the other is an additional treatment, designed for the Flu. The payoff from each treatment is given in \cref{table:Covid-3}; if the patient receives both treatments, the total payoff is the sum of the two.

\begin{table}[htp]
\begin{minipage}{0.5\textwidth}
\centering
\begin{tabular}{|c|c|c|}
\hline
 & $T_1$ & $N_1$\\
\hline
 Covid &  \phantom{-}30 & 0 \\
\hline
 Flu & -20 & 0 \\
\hline
\end{tabular}
\captionof*{table}{Treatment 1}
\end{minipage}
\hfill 
\begin{minipage}{0.5\textwidth}
\centering
\begin{tabular}{|c|c|c|}
\hline
 & $T_2$ & $N_2$\\
\hline
 Covid &  -20 & 0 \\
\hline
 Flu & \phantom{-}30 & 0 \\
\hline
\end{tabular}
\captionof*{table}{Treatment 2}
\end{minipage}
\caption{Payoffs from Two Treatments}
\label{table:Covid-3}
\end{table}

The CDC now chooses among four actions, in the form of $\{T_1,N_1\}\times \{T_2,N_2\}$, specifying whether to administer each of the treatments.  Again, a simple strategy that is not vulnerable to misspecified correlations is to base the treatment decision on only one symptom. It can be easily checked that using either the Cough or the Fever symptom alone guarantees a value of $8.5+6.5=15$. However, the CDC can do better by basing the decision of Treatment 1 on the Cough symptom and Treatment 2 on the Fever symptom, as described in \cref{table:Covid-4}. This strategy, which uses the respective best information source for each treatment, guarantees a value of $8.5+8.5=17$ regardless of the actual correlations between the information sources. 

\begin{table}[htp]
\centering
\begin{tabular}{|c|c|c|}
\hline
  & Fever$_+$ & Fever$_-$\\
\hline
Cough$_+$ &  $T_1+T_2$ & $T_1$\\
\hline
Cough$_-$ & $T_2$ & No Treatment\\
\hline
\end{tabular}
\caption{Using Information from Both Symptoms}
\label{table:Covid-4}
\end{table}

A key property of the decision problem above is the additive separability of payoffs across the two treatments. Indeed, for any decision problem consisting of a collection of binary-state binary-action subproblems whose utilities are summed, which we call a \emph{separable problem}, we show that a robustly optimal strategy is to use the best-source strategy for each subproblem separately. 

The separability property may seem rather restrictive, but in fact,  \textit{every} binary-state decision problem can be written as a separable problem via what we call the binary decomposition of a decision problem. Building on this idea, \cref{thm:binary-general} provides a general construction of robustly optimal strategies for all finite-action, binary-state decision problems. The construction decomposes the decision problem into local binary comparisons between consecutive actions in a natural ordering, uses the best-source strategy for each comparison, and then aggregates these local recommendations back into an action in the original problem.

With three or more states, we adopt a different approach, analyzing the interim value function of the decision problem. The space of interim value functions is generally infinite dimensional, but it can be reduced to a finite-dimensional value vector by focusing on the interim values at a finite set of beliefs, which we call extremal beliefs. The robustly optimal value, as a function of these value vectors, satisfies two properties. First, it must be concave. Second, it must be greater than the value derived from the best-source strategy. Theorem 3 establishes that the robustly optimal value is, in fact, the smallest function satisfying these two properties: the concavification of the value derived from the best source strategy.

A few behavioral implications follow immediately from our characterizations. First, we derive a bound on the number of sources needed for robustly optimal strategies, which equals the number of non-degenerate extremal beliefs of the decision problem.\footnote{This bound simplifies to $|A|-1$ with two states.} The bound depends only on the decision problem, not on the number or specifics of the information sources. This means that when the number of information sources grows large, the fraction of information sources a decision maker needs to pay attention to converges to zero. Second, we characterize which information sources might be used and which can be safely ignored. Specifically, a source can be essential only if it is the unique best standalone source for some decision problems. This highlights the advantage of gathering information from specialized information sources. Lastly, when the state is binary, our results also imply an intuitive property of robust aggregation of action recommendations\,---\,the unanimity rule; that is, to take an action that is unanimously recommended by all information sources. Although unanimity may seem natural, a well-known puzzle is that it fails to hold under conditional independence \citep[see, e.g., ][]{sobel2014relationship}.

Our results show that the agent tends to ignore some freely available information when facing ambiguity on correlations. Information neglect is well-documented, with existing explanations often attributing this behavior to hidden costs or psychological distortions (see \citet{handel2018frictions} for a detailed discussion). Our results offer a different rationale: ignoring certain information can lead to more robust decisions when there is ambiguity  about the correlations among various information sources. This explanation has distinct counterfactual implications. For instance, an agent who finds it costly to acquire or process information would become more informed as stakes are raised, but one who is concerned with correlation robustness would not react to such an incentive.

The rest of the paper is organized as follows: \cref{section-model} introduces the formal model. \cref{section-prelim} establishes preliminary results that will be useful throughout the paper. \cref{section:binary state,section:general state} consider the binary-state and general-state environments, respectively. \cref{section:behavioral-implications} summarizes the behavioral implications of robustly optimal strategies. \cref{section:discussion} discusses extensions. \cref{section:conclusion} concludes. The remainder of this introduction situates our contribution within the broader literature.

\paragraph{Related Literature:}

Our paper studies robust decision making under uncertain correlations among information sources. The practice  of finding robust strategies traces back at least to \citet{wald1950statistical}. The worst-case approach we adopt is in line with the literature on ambiguity aversion \citep{gilboa1989maxmin}. 

Our paper is closest to the literature on robust forecast aggregation.\footnote{Unlike our model, this literature typically assumes that only forecasts\,---\,experts' beliefs about the state\,---\,are observable, instead of the raw information informing those beliefs. See \cref{sec:expert-opinion} for a discussion of how our results change under this assumption.} \citet*{ArieliE12135} propose an adversarial framework for combining forecasts when facing various types of ambiguity on the underlying information structure.  By contrast, we focus solely on ambiguity in the correlation structure and consider general decision problems. \citet{levy2020combining} study forecast aggregation under bounded ambiguity about correlation. They adopt an interim perspective, where the adversarial correlation is chosen after observing the signal realization. We instead take an ex-ante approach, where the adversarial correlation is chosen before observing the signal realizations. The ex-ante approach may be more natural when the decision maker designs a guideline that must be established upfront and applied broadly to an organization or a population, while the interim approach may be better suited for a decision maker facing a specific, individualized problem where signals have already been realized. \citet*{arieli2025random} study a variant of our model and show that when the marginal experiments are symmetric, following a single random information source is robustly optimal under both maxmin payoff and minmax regret criteria.

Our analysis is also related to work comparing information structures when the joint distribution is fully specified. \citet*{borgers2013signals} study when signals are complements or substitutes. \citet{cheng2024diversity} relate disagreement among experts to the informativeness of their joint information. \citet*{brooks2024comparisons} explore the comparison of experiments with known correlations and characterize their ranking that is robust to any additional information. Although our focus is on settings in which only the marginal experiments are known, as an intermediate step in our analysis, we characterize the least informative joint experiment with given marginals.

A separate literature studies environments in which the correlation across information sources is known and can be exploited. Examples include complementary information in social learning \citep{liang2020complementary}, data externalities and the value of correlated data \citep{ichihashi2021economics}, and dynamic aggregation of diverse signals \citep*{liang2022dynamically}. Our results complement this literature by showing how optimal behavior changes when the correlation structure itself is not trusted: uncertainty about correlation alone can make it optimal to ignore some freely available information, and reliance on a single best source is robustly optimal in binary-state binary-action problems.

Robustness to correlations has also been studied in other contexts, such as mechanism design. \citet{carroll2017robustness} studies a multi-dimensional screening problem, where the principal knows only the marginals of the agent's type distribution, and designs a mechanism that is robust to all possible correlation structures. \citet{HeLi2020} and \citet{zhang2021correlation} study an auctioneer's robust design problem when selling an indivisible good, concerning the correlation of values among different agents.

Lastly, a recent literature on robust mechanism design raises concerns about the maxmin criterion and proposes refinements to address them. One concern, raised by \cite{ball2026robust}, is that the payoff guaranteed by a maxmin-optimal mechanism may itself fail to be robust to small perturbations of the ambiguity set. The analogous concern does not arise in our setting, because for any fixed strategy, the expected payoff is continuous in the joint experiment. A second concern is that the maxmin objective often admits a large set of solutions, calling for refinements that discipline this multiplicity. \cite{dworczak2022preparing} propose a two-step selection: among the worst-case optimal policies, select the one that maximizes payoff under a conjectured prior. \cite{kambhampati2025proper} refines the maxmin criterion using lexicographic probability systems and shows that this selects the efficient mechanism in canonical screening settings. \citet*{li2024undominated} propose instead to restrict attention to undominated mechanisms. These refinements are not central in our binary-state, binary-action benchmark, where the robustly optimal strategy is unique except in trivial or knife-edge cases, but they may provide a useful selection criterion in more general environments where robustly optimal strategies are not unique.

\section{Model}\label{section-model}

An agent faces a decision problem $\Gamma=(\Theta,\mu_0,A,\rho)$, with a finite state space $\Theta$, a prior $\mu_0\in\Delta\Theta$, a finite action space $A$, and a utility function $\rho:\Theta\times A\rightarrow \mathbb{R}$. To simplify notation, we define a prior-weighted utility function $u(\theta,a)=\mu_0(\theta)\rho(\theta,a)$, and will refer to the decision problem simply as $(A,u)$. With a slight abuse of notation, for a mixed action $\alpha\in \Delta(A)$, we let $u(\theta,\alpha)=\sum_{a\in A} u(\theta,a)\alpha(a)$.

The agent has access to $m$ information sources, denoted by $\{P_j\}_{j=1}^m$. Each source is a \textbf{marginal experiment}, $P_j:\Theta\rightarrow\Delta Y_j$, mapping each state to a distribution over some finite signal set $Y_j$. Let $\mathbf{Y}=Y_1\times \cdots \times Y_m$ denote the set of all possible profiles of signal realizations, with typical element $\mathbf{y}=(y_1,\ldots,y_m)$. The agent can observe the signals from all marginal experiments, $\{P_j\}_{j=1}^m$, but does not know the joint. Thus, the agent conceives of the following set of \textbf{joint experiments}:
	\[
	\mathcal{J}(P_1,...,P_m) =\left\{
	P:\Theta\rightarrow \Delta(\mathbf{Y}):
	\sum_{y_{-j}}P(y_1,\ldots,y_m|\theta)=P_j(y_j|\theta)\text{ for all }\theta,j,y_j
	\right\}.
	\]

A strategy for the agent is a mapping, $\sigma:\mathbf{Y}\rightarrow \Delta (A)$, and the set of all strategies is denoted by $\Sigma$.
The agent's problem is to maximize her expected payoff considering the worst possible joint experiment:
	\[
	V(P_1, \ldots , P_m; (A,u)):= \max_{\sigma\in\Sigma} \min_{P \in \mathcal{J}(P_1, \ldots , P_m)} \sum_{\theta\in\Theta}\sum_{\mathbf{y} \in \mathbf{Y}} P(\mathbf{y}|\theta)u(\theta,\sigma(\mathbf{y})).
	\]
When there is no confusion about the relevant decision problem, we omit $(A,u)$ from the argument of $V$.  We call a solution to the problem a \textbf{robustly optimal} strategy.
	
If $m=1$, the agent observes only a single experiment $P:\Theta\rightarrow \Delta (Y)$ and $V(P)$ is the classical value of a Blackwell experiment. In this case, a robustly optimal strategy is just an optimal strategy for a Bayesian agent.

\section{Preliminaries}\label{section-prelim}
We begin with some groundwork for our main results. \cref{subsection:payoff-polyhedron} represents decision problems as payoff polyhedrons that capture all feasible payoff vectors. We also introduce the dominance and equivalence relationships between decision problems, which will be instrumental in characterizing robustly optimal strategies in \cref{thm:binary-general,thm:concavification}. \cref{sec:Blackwell} reviews the Blackwell order, drawing on a convenient zonotope representation of Blackwell experiments and the existence of the Blackwell supremum in binary-state environments. Finally, \cref{sec:nature-minmax} solves the dual problem\,---\,Nature's minmax problem, which characterizes the value function and plays a central role in our analysis.

\subsection{Payoff Polyhedron}\label{subsection:payoff-polyhedron}

For any decision problem $(A,u)$, the utility for a given action can be seen as a vector in $\mathbb{R}^\Theta$, denoted by $u(\cdot,a)$. Let
\[\po(A,u)=\co\{u(\cdot,a):a\in A\}-\mathbb{R}_+^{\Theta}\] be the associated polyhedron containing all payoff vectors that are achievable or weakly dominated by some mixed action.\footnote{Here and in what follows, whenever $+$ and $-$ are used as operations between sets, they denote the Minkowski sum and difference.}
An example of $\po(A,u)$ when $|\Theta|=2$ is depicted in \autoref{figure-actionorder}.

\begin{figure}[htp]
	\centering

	\begin{tikzpicture}[domain=0:3, scale=4.5, ultra thick]

    \fill[orange!10] (-0.1,-0.1)--(1,-0.1)--(1,0.2)--(0.9,0.55)--(0.7,0.8)--(0.3,1)--(-0.1,1)--(-0.1,-0.1);
    
	\draw[<->] (0,1.1) node[left]{$\theta=2$} -- (0,0)-- (1.2,0) node[below,xshift=-5]{$\theta=1$};
    
    \filldraw[blue] (1,0.2)circle (0.1pt)node[right]{$\bu(\cdot,a_4)$};
    \filldraw[blue] (0.9,0.55)circle (0.1pt)node[right]{$\bu(\cdot,a_3)$};
    \filldraw[blue] (0.7,0.8)circle (0.1pt)node[right]{$\bu(\cdot,a_2)$};
    \filldraw[blue] (0.3,1)circle (0.1pt)node[right,yshift=2]{$\bu(\cdot,a_1)$};

    \draw[orange] (0.45,0.5)node{$\po(A,u)$};
	
	\end{tikzpicture}   
	\caption{The shaded area represents $\po(A,u)$}
   \label{figure-actionorder}
\end{figure}

Whenever $ \po(A,u)\supseteq \po(\tilde{A},\tilde{u})$, it is immediate that 
\[
V(P_1, \ldots ,P_m ; (A,u)) \geq V(P_1, \ldots , P_m; (\tilde{A},\tilde{u}))
\]
for all Blackwell experiments $P_1, \ldots , P_m$.

\begin{defn}\label{def:equivalent_decision_problem}
We say that $(A,u)$ \textbf{dominates} $(\tilde{A},\tilde{u})$ if 
$\po(A,u)\supseteq\po(\tilde{A},\tilde{u}).$ We say that $(A,u)$ is \textbf{equivalent} to $(\tilde{A},\tilde{u})$ if $\po(A,u)=\po(\tilde{A},\tilde{u})$.
\end{defn}

For instance, if an action is weakly dominated by a mixture of other actions, adding or removing this action leaves us with an equivalent decision problem.

We can also think of the dominance relationship in terms of a direct map between actions that lead to a higher utility across all states.

\begin{defn}\label{def:dominating_maps}
    A \textbf{dominating map} from $(\tilde{A},\tilde{u})$ to $(A,u)$ is a function $f:\tilde{A}\rightarrow \Delta({A})$ such that ${u}(\theta,f(\tilde{a}))\geq \tilde{u}(\theta, \tilde{a})$ for every $\tilde{a}\in \tilde{A}$ and every $\theta\in \Theta$.
\end{defn}

Clearly, $(A,u)$ dominates $(\tilde{A},\tilde{u})$ if and only if there exists a dominating map from $(\tilde{A},\tilde{u})$ to $(A,u)$.

\subsection{The Blackwell Order and the Blackwell Supremum}\label{sec:Blackwell}

We will use the Blackwell order of experiments throughout the paper. For the sake of completeness, we briefly review it in this subsection.\footnote{The zonotope approach to the Blackwell supremum presented in this section is due to \cite{bertschinger2014blackwell}. The lattice structure of the binary-state Blackwell order can also be derived from \citet{kertz1992stochastic}, who establish the lattice structure of the univariate convex order.}
\begin{defn}
$P:\Theta\rightarrow\Delta(Y)$ is more informative than $Q:\Theta\rightarrow\Delta(Z)$  if, for every decision problem, we have the inequality $V(P)\geq V(Q)$. We also say that $P$ Blackwell dominates $Q$.
\end{defn}
We say that two experiments are \textbf{Blackwell equivalent} if they Blackwell-dominate each other. We also say that $P$ is \textbf{strictly more informative} than $Q$ (or $P$ strictly Blackwell dominates $Q$) if $P$ is more informative than $Q$ and they are not Blackwell equivalent.

There are two other natural ways of ranking experiments. The first uses the notion of a \emph{garbling}.
\begin{defn}
$Q:\Theta\rightarrow\Delta(Z)$  is a garbling of $P:\Theta\rightarrow\Delta(Y)$ if there exists a function $g:Y\rightarrow \Delta(Z)$ (the ``garbling") such that 
$
Q(z|\theta)=\sum_{y}g(z|y)P(y|\theta).
$
\end{defn}
Thus $Q$ is a garbling of $P$ when one can replicate $Q$ by ``adding noise'' to the signal generated from $P$. The second ranking uses the feasible state-action distributions.
\begin{defn}
Given a set of actions $A$ and an experiment $P:\Theta\rightarrow\Delta(Y)$, the feasible set of $P$ is
\[
\Lambda_P(A)=\left\{\lambda:\Theta\rightarrow \Delta A\:\Big|\:\lambda(a|\theta)=\sum_y \sigma(a|y)P(y|\theta) \text{ for some } \sigma:Y\rightarrow \Delta(A)\right\}.
\]
\end{defn}

The feasible set of an experiment specifies what conditional action distributions can be obtained by some choice of strategy $\sigma$. One might then say that more information allows for a larger feasible set. 

Blackwell's Theorem states that these rankings of informativeness are equivalent.\footnote{For a proof, see e.g.  \citet{blackwell1953equivalent} or \citet{de2018blackwell}.}

\begin{thmbw}
\label{thm:Blackwell}
The following statements are equivalent
\begin{enumerate}
    \item[(1)] $P$ is more informative than $Q$;
    \item[(2)] $Q$ is a garbling of $P$;
    \item[(3)] For all sets $A$, $\Lambda_Q(A)\subseteq \Lambda_P(A)$.
\end{enumerate}
\end{thmbw}
In addition, when $|\Theta|=2$, Theorem 10 in \cite{blackwell1953equivalent} shows that the above statements are also equivalent to 
\begin{enumerate}
    \item[\textit{(4)}.] \emph{For all sets $A$ with $|A|=2$, $\Lambda_Q(A)\subseteq \Lambda_P(A)$.}
\end{enumerate}

Note that all sets $A$ with the same cardinality give essentially the same set $\Lambda_P(A)$, so condition (3) could equivalently be stated as follows: for every $n\in \mathbb{N}$, we have $\Lambda_Q(\{1,\ldots,n\})\subseteq \Lambda_P(\{1,\ldots,n\})$. Similarly, condition (4) can be stated as $\Lambda_Q(\{1,2\})\subseteq \Lambda_P(\{1,2\})$. To simplify notation, when $|A|=2$, we will omit $A$ in the notation, simply writing $\Lambda_P$.

Condition (4) is particularly useful as it offers a simple graphical representation of Blackwell experiments when $|\Theta|=2$. \autoref{zonotope}(a) illustrates this using the cough symptom from the introduction (see \autoref{table:Covid-2}). To characterize $\Lambda_P$, it suffices to specify the probability of taking one of the two actions, as the probability of taking the other action is the complementary probability. The $x$-axis denotes the probability of taking this action in state $1$, and the $y$-axis denotes the probability in state $2$. This way, $\Lambda_P$ is depicted as a subset of $[0,1]^2$. Clearly $(0,0),(1,1)\in \Lambda_P$ for all $P$, because these two points represent taking a constant action regardless of the signal realization. With the information obtained from the Blackwell experiment, additional points can be obtained. For example, the point $(0.1,0.5)$ in \autoref{zonotope}(a) can be achieved if the decision-maker chooses this action precisely when the patient does not have a cough symptom. Symmetrically, the decision-maker could choose the same action precisely when the patient has a cough symptom, which yields the point $(0.9,0.5)$. Such pure strategies give us the extreme points of the polytope $\Lambda_P$ and the possibility of randomization convexifies the set.  Thus, $\Lambda_P$ is a convex and symmetric\footnote{By symmetric we mean if $\lambda \in \Lambda_P$, $(1,1)-\lambda\in \Lambda_P$.} polytope in $[0,1]^2$, corresponding to the entire shaded area. Conversely, as shown in \cite{bertschinger2014blackwell}, any convex and symmetric polytope in $[0,1]^2$ containing $(0,0)$ and $(1,1)$  corresponds to $\Lambda_P$ for some $P$.\footnote{To see this, take the upper vertices of the convex and symmetric polytope and order them clockwise as $\{\boldsymbol{x_0}\equiv(0,0),\boldsymbol{x_1},...,\boldsymbol{x_{n-1}},\boldsymbol{x_n}\equiv (1,1)\}$. Construct a Blackwell experiment $P:\Theta\rightarrow \Delta\{1,...,n\}$ such that $P(i|\cdot)=\boldsymbol{x_i}-\boldsymbol{x_{i-1}}$ for $i=1,...,n$. One can then verify that $\Lambda_P$ is the polytope.}

\begin{figure}[ht]
	\centering
\subfigure[An example of $\Lambda_P(A)$ with $|\Theta|=|A|=2$]{
	\begin{tikzpicture}[domain=0:3, scale=4.5, ultra thick]

	\draw[<->] (0,1.1) node[above]{$\lambda(\cdot|\theta=2)$} -- (0,0)node[below]{\footnotesize(0,0)}-- (1.2,0) node[below,xshift=-5]{$\lambda(\cdot|\theta=1)$};
	  
	\draw[orange] (0,0)--(0.1,0.5)--(1,1)--(0.9,0.5)--(0,0);
	\fill[orange!20] (0,0)--(0.1,0.5)--(1,1)--(0.9,0.5)--(0,0);
	\draw (1,1)node[right]{\footnotesize(1,1)};

	\draw (0.1,0.5)node[right,yshift=-3]{\footnotesize$(0.1,0.5)$};
    \draw (0.9,0.5)node[below,xshift=17]{\footnotesize(0.9,0.5)};
	\end{tikzpicture}   
}
\hspace{0.2in}
\subfigure[$\Lambda_{\overline{P}}$ as the convex hull of $\Lambda_{P_1}\cup \Lambda_{P_2}$]
{
	\begin{tikzpicture}[domain=0:3, scale=4.5, ultra thick]    

	\fill[orange!20] (0,0)node[below,white]{(0,0)}--(0.1,0.5)--(1,1)--(0.9,0.5)--(0,0);
	\draw[orange] (0.95,0.7)node[right]{$\Lambda_{P_1}$};
	
	\draw[cyan,thin] (0,0)--(0.5,0.9)--(1,1)--(0.5,0.1)--(0,0);
	\fill[pattern=vertical lines, pattern color=cyan] (0,0)--(0.5,0.9)--(1,1)--(0.5,0.1)--(0,0);
	\draw[cyan] (0.7,0.95)node[above]{$\Lambda_{P_2}$};
	
	\draw[gray,dotted,line width=3] (0,0)--(0.1,0.5)--(0.5,0.9)--(1,1)--(0.9,0.5)--(0.5,0.1)--(0,0);
	\draw[gray] (0.3,0.7)node[above,xshift=-6]{$\Lambda_{\overline{P}}$};
	
	\end{tikzpicture}  
}
	
	\caption{}
	\label{zonotope}
\end{figure}

Having reviewed the Blackwell order, we now turn our attention to a concept that will be used extensively in our analysis\,---\,the Blackwell supremum.

\begin{defn}
Let $P_1, P_2, \ldots , P_m$ be arbitrary Blackwell experiments. We say that $\overline{P}$ is the \textbf{Blackwell supremum} of $P_1, P_2, \ldots , P_m$ if
\begin{enumerate}
    \item $\overline{P}$ is more informative than $P_1, P_2, \ldots , P_m$;
    \item If $Q$ is more informative than $P_1, P_2, \ldots , P_m$, then $Q$ is also more informative than $\overline{P}$.
\end{enumerate}
\end{defn}
By definition, if there are two Blackwell suprema, they must Blackwell dominate each other.  This means that any two Blackwell suprema must be Blackwell equivalent and so Blackwell suprema, if they exist, are unique up to Blackwell equivalence.

Furthermore, when the state space is binary, the Blackwell supremum always exists and can be characterized using the feasible set, as illustrated in \autoref{zonotope}(b). From Blackwell's theorem, for any $Q$ that is more informative than $P_1,...,P_m$, the corresponding feasible set $\Lambda_{Q}$ must contain $\Lambda_{P_1}, \ldots , \Lambda_{P_m}$. Since the feasible set is always convex, $\Lambda_{Q}$ must also contain $\co(\Lambda_{P_1}\cup\Lambda_{P_2}\cup \cdots \cup \Lambda_{P_m})$. Moreover, the set $\co(\Lambda_{P_1}\cup\Lambda_{P_2}\cup \cdots \cup \Lambda_{P_m})$ is convex and symmetric, and so it corresponds to some Blackwell experiment $\overline{P}$, which is thus the least informative Blackwell experiment that dominates $P_1,...,P_m$\,---\,the Blackwell supremum. This observation yields the following lemma:

\begin{lem}\label{lem:blackwellsup}
   When $|\Theta|=2$, the Blackwell supremum always exists. An experiment $\overline{P}$ is the Blackwell supremum of $P_1,P_2, \ldots , P_m$ if and only if $\Lambda_{\overline{P}}=\co(\Lambda_{P_1}\cup\Lambda_{P_2}\cup  \cdots \cup \Lambda_{P_m})$.
\end{lem}
\begin{proof}
See Proposition 16 in \cite{bertschinger2014blackwell}.
\end{proof}

It is useful to note that the above lemma holds specifically for $|\Theta| = 2$.  When $|\Theta|\geq 3$, a Blackwell supremum may not exist, as illustrated in Example 18 of \cite{bertschinger2014blackwell}. 

\subsection{Nature's MinMax Problem}\label{sec:nature-minmax}

Most of our focus will be on the robustly optimal strategies for the agent, but it will be helpful to first understand Nature's MinMax problem. Since the objective function is linear in both $\sigma$ and $P$, and the choice sets of $\sigma$ and $P$ are both convex and compact, the minimax theorem \citep{sion1958general} implies that
\begin{align}\label{eq:nature-minmax}
    V(P_1, \ldots , P_m)&= \min_{P \in \mathcal{J}(P_1, \ldots , P_m)} \max_{\sigma\in\Sigma}  \sum_{\theta\in\Theta}\sum_{(y_1, \ldots , y_m) \in \mathbf{Y}} P(y_1, \ldots , y_m|\theta)  u(\theta,\sigma(y_1, \ldots , y_m)) \nonumber\\
                    &= \min_{P \in \mathcal{J}(P_1, \ldots , P_m)} V(P).
\end{align}

That is, the value of the agent's maxmin problem equals the value of a minmax problem where Nature chooses an experiment in the set $\mathcal{J}(P_1, \ldots , P_m)$ to minimize a Bayesian agent's value in the decision problem. 

Let $\mathcal{D}(P_1, \ldots , P_m)$ denote the set of Blackwell experiments that dominate $P_1,...,P_m$.\footnote{Technically, if we allow any finite set to be a signal space, $\mathcal{D}$ is not a set in the strict set-theoretical sense. We can resolve this issue by fixing a large enough universe $U$ of signals. For our purposes, it is enough if every finite subset of $\mathbb{N}$ is in $U$ and all Cartesian products of sets in $U$ are also in $U$.} Every experiment in $\mathcal{J}(P_1, \ldots , P_m)$ must be more informative than each $P_j$, since the projection onto the $j$th coordinate is a garbling, so $\mathcal{J}(P_1, \ldots , P_m)\subseteq \mathcal{D}(P_1, \ldots , P_m)$. The set $\mathcal{D}(P_1, \ldots , P_m)$ is in general a larger set, because not every experiment that dominates $P_1,...,P_m$ can be represented as a joint experiment with marginals $P_1,...,P_m$.\footnote{For a simple example, consider two experiments $P_1$ and $P_2$ whose signal spaces $Y_1$ and $Y_2$ are singletons. Then $\mathcal{J}(P_1,P_2)$ contains only the completely uninformative experiment while $\mathcal{D}(P_1,P_2)$ contains all Blackwell experiments.} However, the next lemma shows that relaxing Nature's problem to choosing an experiment from the set $\mathcal{D}(P_1, \ldots , P_m)$ does not change the value of the problem.

\begin{lem}\label{lem:relaxed-problem}
\begin{equation}\label{eq:relaxed}
       V(P_1, \ldots , P_m)=  \min_{P \in \mathcal{D}(P_1, \ldots , P_m)} V(P) 
\end{equation}
\end{lem}
\begin{proof}
 See   \cref{proof:blackwellsup}.  
\end{proof}

The idea underlying \cref{lem:relaxed-problem} is that in the relaxed problem above, Nature could restrict attention to the experiments that are Blackwell minimal\,---\,those that do not strictly Blackwell dominate any other experiment in $\mathcal{D}(P_1, \ldots , P_m)$. Additionally, any Blackwell minimal element in this set can be represented as a joint experiment in $\mathcal{J}(P_1, \ldots , P_m)$, as shown in \cref{proof:blackwellsup}. 

\cref{lem:relaxed-problem} is particularly useful when the state is binary. Under binary states, the Blackwell supremum $\overline{P}$ of $P_1,...,P_m$ exists, and it is the unique (up to Blackwell equivalence) Blackwell minimal element in $\mathcal{D}(P_1, \ldots , P_m)$. Therefore, $\overline{P}$ solves \eqref{eq:relaxed} regardless of the decision problem, yielding the following corollary.

\begin{cor} \label{cor:value-supremum}
\label{value}
When $|\Theta|=2$,
    $$V(P_1,...,P_m)=V(\overline{P}(P_1,...,P_m))$$
    where $\overline{P}(P_1,...,P_m)$ is the Blackwell supremum of experiments $\{P_1,...,P_m\}$.
\end{cor}
Thus, in binary-state decision problems, the agent's value from using a robust strategy is the same as the value she would obtain if she faced a single experiment\,---\,the Blackwell supremum of all marginal experiments. Moreover, the Blackwell supremum depends only on the marginal experiments, and not on the particular decision problem.

\section{Binary State Environment}\label{section:binary state}

\subsection{Binary-State Binary-Action Problems}

As seen in the introductory example, one simple strategy that generates a robust value independent of the correlations among the marginal information sources is to choose exactly one marginal experiment from $\{P_1, \ldots ,P_m\}$ and play the optimal strategy that uses that information alone, ignoring the signal realizations of all other experiments. By choosing the marginal experiment optimally, the agent achieves an expected payoff of $\max_{j = 1, \ldots , m} V(P_j)$, regardless of the actual joint experiment $P \in \mathcal{J}(P_1, \ldots , P_m)$. We call such a strategy a \textit{best-source strategy}.\footnote{Formally, a strategy $\sigma$ is a best-source strategy if, for some $j^*\in\argmax_{j=1,\ldots,m}V(P_j)$, it is constant in all arguments other than $y_{j^*}$ and responds optimally to the signal from $P_{j^*}$.}

In some cases, it is clear that a best-source strategy is robustly optimal. Suppose, for example, that we have two information sources, $P_1$ and $P_2$, and that $P_1$ is more informative than $P_2$. We can then consider a correlation structure where the signal of $P_2$ is generated by garbling the signal of $P_1$. Consequently, after observing the signal from $P_1$, observing signals from $P_2$ provides no additional information. Therefore, the agent loses nothing by ignoring $P_2$, and the best-source strategy that uses only $P_1$ is optimal. 

The interesting cases are when information sources are not Blackwell-ranked (which will often be the case, given the demanding nature of the Blackwell order). In such cases, the Blackwell supremum is strictly more informative than any single information source, so one may hope to do better than a best-source strategy by combining different sources. Surprisingly, \cref{thm:binarybinary} shows that, in decision problems with binary states and binary actions, the agent can never do better than a best-source strategy. Moreover, if the information sources satisfy full support and we exclude cases where multiple information sources are optimal, then any strategy that uses more than one source is strictly suboptimal.

\begin{thm}\label{thm:binarybinary}
    For all $(A,u)$ with $|A| = |\Theta| = 2$, any best-source strategy is robustly optimal. In other words,
    \[
   V(P_1, \ldots , P_m) =  \max_{j = 1, \ldots , m}V(P_j).
    \]
In addition, if each marginal experiment has full support, i.e., $P_j(y_j|\theta)>0$ for all $j,y_j,\theta$, and $\argmax_{j = 1, \ldots , m} V(P_j)$ is a singleton, then all robustly optimal strategies are best-source strategies.
\end{thm}
\begin{proof}

To simplify notation, we write $\overline{P}$ to refer to the Blackwell supremum, $\overline{P}(P_1, \ldots , P_m)$.
By \autoref{value}, it suffices to show that $V(\overline{P})=\max_{j=1,...,m}V(P_j)$.
By \autoref{lem:blackwellsup}, 
\begin{equation}\label{convex_hull}
\Lambda_{\overline{P}} =\co\left(\Lambda_{P_1}\cup \cdots \cup \Lambda_{P_m}\right).    
\end{equation}
Now, the maximum utility achievable given $\overline{P}$ is $V(\overline{P})=\max_{\lambda\in \Lambda_{\overline{P}}}\sum_{a,\theta} u(\theta,a)\lambda(a|\theta).$
Since the maximand is linear in $\lambda$, the fundamental theorem of linear programming states that the maximum is achieved at an extreme point of $\Lambda_{\overline{P}}$. By (\ref{convex_hull}), an extreme point of $\Lambda_{\overline{P}}$ must belong to some $\Lambda_{P_j}$. Hence, we have
$$V(\overline{P})= \max_{\lambda\in \Lambda_{P_j}}\sum_{a,\theta} u(\theta,a)\lambda(a|\theta)=V(P_j)\leq \max_{j' = 1,\ldots , m} V(P_{j'}).$$
Since $\overline{P}$ is more informative than every $P_j$, we also have $V(\overline{P}) \geq \max_{j'= 1, \ldots , m} V(P_{j'})$, which concludes the proof of the theorem’s first statement. The proof of the second part requires different arguments that use polytope decomposition and optimal transport duality, so we defer it to \cref{sec:proof-uniqueness}.
\end{proof}

\begin{figure}
\centering
	\begin{tikzpicture}[domain=0:3, scale=3.8, ultra thick]

	\fill[orange!20] (0,0)node[below,white]{(0,0)}--(0.1,0.5)--(1,1)--(0.9,0.5)--(0,0);
	\draw[orange] (0.9,0.5)node[right]{$\Lambda_{P_1}$};
	
	\draw[cyan,thin] (0,0)--(0.5,0.9)--(1,1)--(0.5,0.1)--(0,0);
	\fill[pattern=vertical lines, pattern color=cyan] (0,0)--(0.5,0.9)--(1,1)--(0.5,0.1)--(0,0);
	\draw[cyan] (0.5,0.1)node[below,xshift=2]{$\Lambda_{P_2}$};
	
	\draw[gray,dotted,line width=3] (0,0)--(0.1,0.5)--(0.5,0.9)--(1,1)--(0.9,0.5)--(0.5,0.1)--(0,0);

    \draw[gray] (0.7,0.2)node[right]{$\Lambda_{\overline{P}}$};
	
	\draw[red] (0,0.6)--(0.5,0.9)--(1,1.2);

    \draw[red,->] (0.5,0.9)--(0.41,1.05)node[right,xshift=2]{$u$};
	
	\end{tikzpicture}  
	\caption{The maximum is achieved at an extreme point that belongs to $\Lambda_{P_2}$}
	\label{extreme}
\end{figure}

The idea of \autoref{thm:binarybinary} can be visualized in \autoref{extreme} for two marginal experiments. Each marginal Blackwell experiment $P_1,P_2$ can be represented by $\Lambda_{P_1},\Lambda_{P_2}$, the set of feasible state-action distributions generated by the experiment. The corresponding $\Lambda_{\overline{P}}$ for the Blackwell supremum $\overline{P}$ is the convex hull of $\Lambda_{P_1}\cup \Lambda_{P_2}$. Since the utility is linear with respect to $\lambda \in \Lambda_{\overline{P}}$, the maximum is achieved at an extreme point, which belongs to either $\Lambda_{P_1}$ or $\Lambda_{P_2}$, and thus can be achieved by using a single marginal experiment.

\subsection{Separable Problems}\label{sec:separable}

While best-source strategies are sufficient in binary-state, binary-action decision problems, more complicated problems often require the agent to use more sophisticated strategies to robustly aggregate information from multiple sources. However, as we will see, best-source strategies form the building blocks from which the robustly optimal strategies are constructed.  For example, in the Covid example with two treatments presented in the introduction, we saw that a simple yet robust strategy that uses multiple information sources is to consider the two treatments separately, using the corresponding best-source strategy for each treatment decision.  

As a first step toward the analysis of robustly optimal strategies in general binary-state decision problems, we generalize this idea in the example to a class of decision problems, which we call \emph{separable}.

\begin{defn}
  A binary-state decision problem $(A,u)$ is a \textbf{separable problem} if $A$ can be written as a product $A_1\times\cdots\times A_k$ where $|A_\ell|=2$ for all $\ell=1,...,k$, and 
\[u(\theta,a)=u_1(\theta,a_1)+\cdots+u_k(\theta,a_k)\] 
for some $\{u_\ell:\Theta\times A_\ell\rightarrow \mathbb{R}\}_{\ell=1}^k$.
\end{defn}

We will use $\bigoplus_{\ell=1}^k(A_\ell,u_\ell)$ to refer to a separable problem and we refer to each of the binary-action decision problems, $(A_\ell, u_\ell)$, as a \emph{subproblem}.

As highlighted in the Covid example, a simple strategy in separable problems is to use the best-source strategy in each of the subproblems, which guarantees a payoff of $\sum_{\ell=1}^k \max_{j=1,...,m} V(P_j;(A_\ell,u_\ell))$ regardless of the correlations between the information sources.  Hence, this value provides a lower bound on the robustly optimal value:
\begin{equation*}
V\left(P_1,\dots,P_m;\bigoplus_{\ell=1}^k(A_\ell,u_\ell)\right) \geq \sum_{\ell=1}^k \max_{j=1,...,m} V(P_j;(A_\ell,u_\ell)).\label{eq:value-separable}
\end{equation*} 
The following lemma further shows that the above inequality is indeed an equality.
This follows from a special property highlighted in \cref{cor:value-supremum}\,---\,that in binary-state environments, there exists a single $\overline{P}(P_1, \ldots , P_m)$ that uniformly minimizes the agent's value across all decision problems.\footnote{In contrast, with three or more states, Nature's worst-case joint experiment in \cref{eq:nature-minmax} typically depends on the decision problem. Therefore, $\min_{P\in \mathcal{J}}V(P;\bigoplus_{\ell=1}^k(A_\ell,u_\ell))\geq \sum_{\ell=1}^k \min_{P\in\mathcal{J}} V(P;(A_\ell,u_\ell))$, which in general is not an equality.\label{ft:three-states}}

\begin{lem}\label{prop:separable}
    For any separable problem $\bigoplus_{\ell=1}^k(A_\ell,u_\ell)$,
    \[V\left(P_1,\dots,P_m;\bigoplus_{\ell=1}^k(A_\ell,u_\ell)\right) = \sum_{\ell=1}^k \max_{j=1,...,m} V(P_j;(A_\ell,u_\ell)).\]
Moreover, let $\sigma_\ell:\mathbf{Y}\rightarrow A_\ell$ be a pure best-source strategy for subproblem $(A_\ell,u_\ell)$. Then $\sigma:\mathbf{Y}\rightarrow  A_1 \times \cdots\times A_k$ defined by
\begin{equation}\label{eq:separable-optimal-strateg}
\sigma(y_1,...,y_m)=\bigg(\sigma_\ell(y_1,...,y_m)\bigg)_{\ell=1}^k\quad\text{ for all }y_1,...,y_m
\end{equation}
is a robustly optimal strategy for decision problem $\bigoplus_{\ell=1}^k(A_\ell,u_\ell)$.
\end{lem}

\begin{proof}
    See \cref{sec:separableproof}.
\end{proof}

\subsection{General Decision Problems and Decompositions}\label{sec:decomposition}

The special structure of separable problems yields simple robustly optimal strategies. To what extent can this structure be applied in tackling more general decision problems? We use the idea that appears in Theorem 10 in \cite{blackwell1953equivalent} to show that \textit{every} binary-state decision problem is equivalent to a separable problem in the sense of \cref{def:equivalent_decision_problem}. The idea is to decompose an $n$-action decision problem into $n-1$  binary-action decision problems and use these subproblems to construct the separable problem that is equivalent to the original problem. We call the resulting separable problem the \textit{binary decomposition}.

Before we define the binary decomposition, we make a normalization to simplify exposition. First we remove all weakly-dominated actions,\footnote{An action $a\in A$ is weakly-dominated if there exists $\alpha\in \Delta A$ such that $u(\cdot,a)\leq u(\cdot,\alpha)$ and $u(\cdot,a) \neq u(\cdot, \alpha)$. If there are duplicated actions, we remove all but keep one copy.} so that actions can be ordered such that a higher action best responds to a higher belief on $\theta_1$:\footnote{Note that since weakly-dominated actions have been removed, the ordering also implies that the slope of difference vectors $u(\cdot,a_{\ell+1})-u(\cdot,a_\ell)$ is weakly decreasing.}
\begin{align*}
 u(\theta_1,a_1) < u(\theta_1,a_2) < \cdots < u(\theta_1,a_n),\\
 u(\theta_2,a_1) > u(\theta_2,a_2) > \cdots > u(\theta_2,a_n).
\end{align*}
Moreover, by adding a constant vector, we can normalize $u(\cdot,a_1)=(0,0)$. 

\begin{defn}

Given a decision problem $(A, u)$, the \textbf{binary decomposition} of $(A, u)$ is a separable problem $\bigoplus_{\ell=1}^{n-1} (A_\ell,u_\ell)$ where
\[
A_{\ell} := \left\{0, 1 \right\}, u_\ell(\cdot,0) = (0,0), u_\ell(\cdot,1) = u(\cdot,a_{\ell + 1}) - u(\cdot,a_{\ell}).
\]
\end{defn}

\begin{figure}[t]
	\centering
\subfigure[Binary decomposition]{
		\begin{tikzpicture}[domain=0:3, scale=4.8, ultra thick]    
	
    \fill[orange!10] (0.2,0.1)--(1,0.1)--(1,0.2)--(0.9,0.55)--(0.7,0.8)--(0.3,1)--(0.2,1)--(0.2,0.1);
    
	\draw[<->] (0.3,0.1)node[left]{$\theta=2$}--(0.3,1)node[above]{(0,0)}--(1.2,1)node[above]{$\theta=1$};

    \draw[red,<-] (0.7,0.8)--(0.3,1) ;
 	\draw[red] (0.45,0.94)node[right]{$u_1(\cdot,1)$};
 	
 	\draw[red,<-] (0.9,0.55)--(0.7,0.8);
 	\draw[red] (0.75,0.75)node[right]{$u_2(\cdot,1)$};
    
	\draw[red,<-] (1,0.2)--(0.9,0.55);
 	\draw[red] (0.95,0.4)node[right]{$u_3(\cdot,1)$};
    
    \filldraw[blue] (1,0.2)circle (0.1pt)node[left]{$u(\cdot,a_4)$};
    \filldraw[blue] (0.9,0.55)circle (0.1pt)node[left]{$u(\cdot,a_3)$};
    \filldraw[blue] (0.7,0.8)circle (0.1pt)node[left,yshift=-5]{$u(\cdot,a_2)$};
    \filldraw[blue] (0.3,1)circle (0.1pt)node[below left]{$u(\cdot,a_1)$};

    \draw[orange] (0.47,0.4)node{$\po(A,u)$};
	
	\end{tikzpicture}   
}
\hspace{0.1in}
\subfigure[A nonconsecutive sum of $u_\ell(\cdot,1)$ lies in the interior of $\po(A,u)$]{
	\begin{tikzpicture}[domain=0:3, scale=4.8, ultra thick]

    \fill[orange!10] (0.2,0.1)--(1,0.1)--(1,0.2)--(0.9,0.55)--(0.7,0.8)--(0.3,1)--(0.2,1)--(0.2,0.1);

	\draw[<->] (0.3,0.1)node[left]{$\theta=2$}--(0.3,1)node[above]{(0,0)}--(1.2,1)node[above]{$\theta=1$};

    \draw[dotted,orange!50] (1,0.2)--(0.9,0.55)--(0.7,0.8);
   
    \draw[red,<-] (0.7,0.8)--(0.3,1) ;
 	\draw[red] (0.45,0.94)node[right]{$u_1(\cdot,1)$};

	\draw[red,<-] (0.8,0.45)--(0.7,0.8);
 	\draw[red] (0.75,0.65)node[left]{$u_3(\cdot,1)$};
   \filldraw[blue] (0.805,0.44)circle (0.1pt);

    \draw[orange] (0.47,0.4)node{$\po(A,u)$};
	
	\end{tikzpicture}   
}
	
	\caption{}
	\label{figure:canonical}
\end{figure}

The key idea underlying the binary decomposition is to decompose the original problem into binary-action decision problems that compare each pair of consecutive actions. This can be visualized in \autoref{figure:canonical}(a). The four-action decision problem is decomposed into three binary-action decision problems, by examining the difference vectors $u(\cdot,a_{\ell + 1}) - u(\cdot,a_{\ell})$. Each subproblem of the decomposition can be interpreted as choosing whether to move forward to the next action in the ordering.

Notice that every feasible payoff vector in the original problem can be replicated in the binary decomposition, due to the fact that $u(\cdot, a_i)=\sum_{\ell=1}^{i-1}u_{\ell}(\cdot,1)+\sum_{\ell=i}^{n-1}u_{\ell}(\cdot,0)$ for all $i=1,...,n$. So $\po(A,u)\subseteq \po\big(\bigoplus_{\ell=1}^{n-1}(A_{\ell},u_\ell)\big)$. Of course, the binary decomposition $\bigoplus_{\ell=1}^{n-1}(A_{\ell},u_\ell)$ could introduce additional feasible payoff vectors. For example, in the example in \autoref{figure:canonical}(b), the strategy $(1,0,1)$ in the binary decomposition yields a payoff vector that is infeasible in the original problem. However, this additional payoff vector lies in the interior of $\po(A,u)$, and thus it is dominated by one of the original (possibly mixed) actions. The next lemma shows that this is generally true, so we have $\po(A,u)= \po\big(\bigoplus_{\ell=1}^{n-1}(A_{\ell},u_\ell)\big)$\,---\,any decision problem is equivalent to its binary decomposition in the sense of \cref{def:equivalent_decision_problem}.

\begin{lem}\label{lem:equivalence}
    $(A,u)$ is equivalent to its binary decomposition.
\end{lem}
\begin{proof}
    See \cref{proof:equivalence}.
\end{proof}

\cref{lem:equivalence} implies that the set of dominating maps (see \cref{def:dominating_maps}) from $\bigoplus_{\ell=1}^{n-1}(A_{\ell},u_\ell)$ to $(A,u)$ is non-empty. Together with \cref{prop:separable}, it allows us to derive a robustly optimal strategy for any decision problem $(A,u)$ through its binary decomposition.

\begin{thm}\label{thm:binary-general}
Suppose $|\Theta|=2$. Let $\bigoplus_{\ell=1}^{n-1}(A_{\ell},u_\ell)$ be the binary decomposition of $(A, u)$, and $\sigma_\ell$ be a pure best-source strategy for $(A_{\ell},u_\ell)$.   Then
\[
V(P_1, \ldots , P_m; (A,u)) = \sum_{\ell = 1}^{n-1} \max_{j = 1, \ldots , m} V(P_j; (A_\ell, u_\ell)).
\]
Moreover, for any dominating map $f$ from $\bigoplus_{\ell=1}^{n-1}(A_{\ell},u_\ell)$ to $(A,u)$,
$\sigma^*_f(\mathbf{y}) :=  f(\sigma_1(\mathbf{y}), \ldots, \sigma_{n - 1}(\mathbf{y}))$ is robustly optimal.

\end{thm}

\begin{proof}
    See \cref{sec:binarygenproof}.
\end{proof}

\autoref{thm:binary-general} allows the construction of a robustly optimal strategy for any decision problem $(A,u)$ according to a two-step procedure: 
\begin{enumerate}
    \item For each subproblem, $(A_\ell,u_\ell)$, find a best-source strategy $\sigma_\ell$ (which we know is robustly optimal by \cref{thm:binarybinary}). 
    \item For each realization $\mathbf{y}$, pick a (mixed) action $\sigma^*(\mathbf{y})\in\Delta (A)$ such that $u(\cdot,\sigma^*(\mathbf{y}))\geq \sum_{\ell=1}^{n-1} u_\ell(\cdot,\sigma_\ell(\mathbf{y}))$.\footnote{Note that in \cref{thm:binarybinary}, the agent does not need to randomize. In contrast, the construction of a robustly optimal strategy in \cref{thm:binary-general} may require randomization when the payoff $\sum_{\ell=1}^{n-1} u_\ell(\cdot,\sigma_\ell(\mathbf{y}))$ is dominated by a mixed action but not by any pure action. We provide an example in \cref{sec:necessity-of-randomization} where all robustly optimal strategies require randomization.}
\end{enumerate} 
Notably, once $\sigma_\ell(\cdot )$ has been determined in Step 1, the marginal experiments, $P_1, \ldots , P_m$, play no role in Step 2.  In other words, the marginal experiments only influence the ultimate choice of action in Step 1, and more specifically through their effect on the choice of $\sigma_\ell(\mathbf{y})$ in each of the subproblems.

We illustrate the construction with a concrete three-action decision problem.
\begin{ex}\label{ex:three-treatments}
Return to the CDC setting in the introduction, with equal prior probabilities of Covid and the Flu. Suppose now that the CDC must choose one of three treatments, ranging from the most conservative, $a_1$, to the most aggressive toward Covid, $a_3$. Their payoffs are given in \cref{table:three-treatments}. 

\begin{table}[htp]
\centering
\begin{tabular}{|c|c|c|c|}
\hline
 & Conservative $a_1$ & Intermediate $a_2$ & Aggressive $a_3$ \\
\hline
Covid & $0$ & $30$ & $50$ \\
\hline
Flu & $0$ & $-20$ & $-50$ \\
\hline\end{tabular}
\caption{Payoffs from Three Treatments}
\label{table:three-treatments}
\end{table}

Let $(A,u)$ denote this decision problem. The action order is consistent with our normalization in this section. We now construct the binary decomposition by comparing each pair of consecutive treatments.

Moving from $a_1$ to $a_2$ changes payoffs by $(30,-20)$, while moving from $a_2$ to $a_3$ changes payoffs by $(20,-30)$. The two binary subproblems therefore have $A_1=A_2=\{0,1\}$ with payoffs
\[
\begin{array}{c|cc}
 & a=0 & a=1 \\ \hline
u_1(\cdot,a) & (0,0) & (30,-20) \\
u_2(\cdot,a) & (0,0) & (20,-30).
\end{array}
\]

Note that the first subproblem $(A_1,u_1)$ is exactly the decision of whether to administer Treatment 1 in the introduction, and the second subproblem $(A_2,u_2)$ is a normalized version (adding a payoff vector of $(20,-30)$ to both actions) of Treatment 2.

The separable problem consisting of these two decomposed binary-action subproblems is therefore a decision problem with four actions that yields payoff vectors $(0,0),(30,-20),(20,-30),(50,-50)$. Note that $(20,-30)$ is dominated by $(30,-20)$, hence $(A_1,u_1)\bigoplus (A_2,u_2)$ is equivalent to the original problem $(A,u)$. Therefore, $(A_1,u_1)\bigoplus (A_2,u_2)$ is the binary decomposition of $(A,u)$.

Now we construct a robustly optimal strategy using \cref{thm:binary-general}. With the same Cough and Fever experiments as in the introduction, the unique best-source strategy for $(A_1,u_1)$ chooses $a=1$ if and only if the patient has a positive Cough symptom, and the unique best-source strategy for $(A_2,u_2)$ chooses $a=1$ if and only if the patient does not have a Fever symptom. The optimal strategy in each subproblem is thus given by
\begin{table}[htp]
\begin{minipage}{0.5\textwidth}
\centering
\begin{tabular}{|c|c|c|}
\hline
$\sigma_1$ & Fever$_+$ & Fever$_-$ \\
\hline
Cough$_+$ & $a=1$ & $a=1$ \\
\hline
Cough$_-$ & $a=0$ & $a=0$ \\
\hline
\end{tabular}
\captionof*{table}{Subproblem $(A_1,u_1)$}
\end{minipage}
\hfill 
\begin{minipage}{0.5\textwidth}
\centering
\begin{tabular}{|c|c|c|}
\hline
$\sigma_2$ & Fever$_+$ & Fever$_-$ \\
\hline
Cough$_+$ & $a=0$ & $a=1$ \\
\hline
Cough$_-$ & $a=0$ & $a=1$ \\
\hline
\end{tabular}\captionof*{table}{Subproblem $(A_2,u_2)$}
\end{minipage}
\end{table}

One dominating map from $(A_1,u_1)\bigoplus (A_2,u_2)$ to $(A,u)$ is $f(0,0)=a_1$, $f(1,0)=f(0,1)=a_2$, and $f(1,1)=a_3$. Applying this dominating map yields a robustly optimal strategy for the original decision problem $(A,u)$:
\begin{table}[htp]
\centering
\begin{tabular}{|c|c|c|}
\hline
$\sigma_f^*$ & Fever$_+$ & Fever$_-$ \\
\hline
Cough$_+$ & $f(1,0)=a_2$ & $f(1,1)=a_3$ \\
\hline
Cough$_-$ & $f(0,0)=a_1$ & $f(0,1)=a_2$ \\
\hline
\end{tabular}
\end{table}

\end{ex}

In contrast to \cref{thm:binarybinary}, \cref{thm:binary-general} also highlights that robustly optimal strategies can generally be non-unique when there are three or more actions. For example, in \cref{ex:three-treatments}, if we use a different dominating map $f:\{0,1\}^2\to \Delta \{a_1,a_2,a_3\}$, with $f(0,0)=a_1$, $f(1,0)=a_2$, $f(0,1)=\frac{4}{5}\circ a_2+\frac{1}{5}\circ a_3$, and $f(1,1)=a_3$, we will construct a different robustly optimal strategy.\footnote{In fact, a robustly optimal strategy could even use a strictly dominated action under some signal realizations. This means that replacing the dominated action with any other action that dominates it also yields a robustly optimal strategy. For example, in the introductory example, the robustly optimal strategy we constructed in \cref{table:Covid-4} recommends no treatment when a patient has no symptoms. Giving neither treatment is dominated by giving both treatments. Why can a robustly optimal strategy play a dominated action? Because given this strategy, it is optimal for the adversarial Nature to put probability 0 on the symptom realization $(Cough_{-},Fever_{-})$.}

\medskip

\autoref{thm:binary-general}  delivers two immediate corollaries.

\begin{cor}\label{cor:number}
For any decision problem $(A,u)$ and any collection of experiments $\{P_j\}_{j=1}^m$, there exists a subset of marginal experiments $\{P_j\}_{j\in J\subseteq \{1,...,m\}}$ with $|J|\leq |A|-1$, such that 
\[V(P_1,\cdots,P_m;(A,u))=V(\{P_j\}_{j\in J};(A,u)).\]
\end{cor}
\autoref{cor:number} establishes a bound, $|A|-1$, on the number of information sources needed for a robustly optimal strategy. Note that this bound is \emph{independent} of the fine details of the decision problem, such as the exact cardinal utilities of the agent, and the details of the marginal information sources available to the agent.

\begin{cor}\label{cor:benefit}
Suppose $\bigoplus_{\ell=1}^{n-1} (A_\ell,u_\ell)$ is the binary decomposition of $(A,u)$.  For any $j$,
\[V(P_1,...,P_m;(A,u))>V(P_1,\ldots ,P_{j-1},P_{j+1},\ldots, P_m;(A,u))\]
if and only if $V(P_j;(A_\ell,u_\ell))> \max_{j'\neq j}V(P_{j'};(A_\ell,u_\ell))$ for some $\ell=1,...,n-1$.
\end{cor}

\autoref{cor:benefit} shows that an additional marginal experiment robustly improves the agent's value if and only if it outperforms all other marginal experiments in at least one of the subproblems of the decomposition. In particular, an experiment that performs reasonably well across all subproblems can be completely ignored if, for each subproblem, there is some other, more specialized experiment that is the best. \cref{ex:third-symptom} in \cref{sec:ignore-best-standalone-information-source} further illustrates that even the best standalone information source may be ignored.

\section{General-State Decision Problems}\label{section:general state}

In this section, we turn to general-state decision problems and characterize the robustly optimal value using concavification. This characterization offers an explicit formula for the robustly optimal value and implies several corollaries that speak to the structure of robustly optimal strategies.

It is worth noting why our previous approach for binary-state decision problems would not extend to general-state decision problems: First, with more states, it is not immediately clear how to decompose a general decision problem into more ``elementary'' ones. Second, the non-existence of the Blackwell supremum implies that in Nature's minmax problem \cref{eq:nature-minmax}, there may no longer be a single experiment that uniformly minimizes the agent's value across all decision problems, exacerbating the complexity of the analysis  (see \cref{ft:three-states}). Lastly, an agent may want to use multiple information sources even in simple binary-action decision problems, as illustrated in \cref{ex:threestaets} in \cref{sec:necessity-of-combining-information-three-state}.

The above challenges lead us to tackle the problem with a geometric approach that offers a characterization via concavification. An important simplification is to focus on the interim value function $v:\Delta (\Theta)\rightarrow \mathbb{R}$ induced by the decision problem. By \cref{eq:nature-minmax}, the robustly optimal value given an interim value function $v$ can be written as 
\begin{equation}\label{eq:Nature-value-fuction}
  V(P_1,...,P_m;v)=\min_{P\in \mathcal{J}(P_1,...,P_m)}\int v(\mu)d\tau^P(\mu), 
\end{equation}
where $\tau^P$ is the posterior distribution induced by $P$.\footnote{More precisely, the induced posterior distribution $\tau^P\in \Delta(\Delta(\Theta))$ associated with a Blackwell experiment $P:\Theta\rightarrow\Delta Y$ is defined as 
  \[\tau^P(\mu)=\sum_{y\in Y_\mu} \sum_\theta \mu_0(\theta)P(y|\theta),\]
    where $Y_\mu$ is the set of signal realizations that induce the posterior belief $\mu$:
    \[Y_\mu=\left\{y\in Y\Big| \frac{\mu_0(\theta)P(y|\theta)}{\sum_{\theta'\in \Theta} \mu_0(\theta') P(y|\theta')}=\mu(\theta)\text{ for all }\theta\in \Theta\right\}.\] 
} 

This representation immediately reveals two properties of  $V(P_1,...,P_m;\cdot)$ as a function of the value function $v$. First, for any fixed $P$, the map $v\mapsto \int v(\mu)d\tau^P(\mu)$ is linear in $v$. Therefore, $V(P_1,...,P_m;\cdot)$, as a pointwise minimum of linear functions in $v$, is concave in $v$. Second, the robustly optimal value is no less than the best-source value, that is, $V(P_1,...,P_m;\cdot)\geq \max_{j=1,...,m} V(P_j;\cdot)$.

The essence of Theorem 3 is to show that $V(P_1,...,P_m;\cdot)$ is in fact the smallest function with these two properties. That is, 
\begin{equation}\label{eq:concavification}
    V(P_1,...,P_m;v)=conc\left(\max_{j=1,...,m} V(P_j;\cdot )\right)(v),
\end{equation}
where $conc$ denotes the concave envelope (see \cref{sec:convexification-concavification} for the formal definition of concavification). This characterization offers useful insight because the solutions to the concavification problem directly translate into robustly optimal strategies: each point in the decomposition corresponds to a subproblem in which the agent uses only the single best information source for that subproblem, generalizing the construction in \cref{thm:binary-general}.

However, \cref{eq:concavification} is an infinite-dimensional concavification problem, posing two difficulties. First, it is analytically cumbersome to work with directly. Second, and more importantly, the infinite-dimensional nature offers no bound on the complexity of the optimal solution: in a $k$-dimensional concavification problem, Carath\'eodory's theorem guarantees that the concave envelope at any point can be achieved by a convex combination of at most $k$ points. This upper bound on support size has no counterpart when the space of value functions is infinite-dimensional, and without it we cannot say how many information sources or decomposed subproblems are needed to achieve the optimum.

Both difficulties are resolved by dimension reduction: since $v$ is convex and piecewise linear, it is entirely determined by its values at a finite set of extremal beliefs\,---\,the extreme points of the epigraph of $v$ projected onto $\Delta(\Theta)$. We can therefore replace the infinite-dimensional value function  with finite-dimensional value vectors, recording only the values of $v$ at its extremal beliefs. Reformulating \cref{eq:concavification} in terms of value vectors yields a finite-dimensional concavification problem, to which Carath\'eodory's theorem applies and restores the bound on the number of subproblems needed. 

The remainder of this section proceeds as follows. Section 5.1 introduces value functions and extremal beliefs, which are the key to dimension reduction.  Section 5.2 introduces value vectors\,---\,finite-dimensional summaries of the value function that contain all the information needed to characterize a robustly optimal strategy\,---\,and defines the operator $W(P;\cdot)$, which evaluates each information source against a value vector. Section 5.3 establishes the formal concavification formula and its corollaries.

\subsection{Value Functions and Extremal Beliefs}

The interim value function $v:\Delta(\Theta)\rightarrow \mathbb{R}$ of a decision problem $(A,u)$ is defined by
\[v(\mu)\coloneqq \max_{a\in A}\sum_{\theta\in\Theta} \mu(\theta)\rho(\theta,a).
\]
If the value functions of two decision problems differ by an affine function, $g(\mu)$, then these two decision problems are equivalent, in the sense that their robustly optimal values will differ by $g(\mu_0)$ for any marginal experiments.  Hence, without loss of generality, we normalize $v(\delta_\theta)=0$ for all $\theta\in \Theta$.\footnote{The normalization is obtained by adding the linear function $g(\mu)=-\sum_{\theta\in\Theta}\mu(\theta)v(\delta_\theta)$ to the original value function.}

The value function of a finite decision problem is convex and piecewise linear.  We denote the epigraph of $v$ by $\epi(v)=\{(\mu,w):w\geq v(\mu),\mu\in \Delta(\Theta)\}$. The set of extreme points of the epigraph, denoted by $\ext(\epi(v))$, is finite and contains $\{(\delta_\theta,v(\delta_\theta))\}_{\theta\in\Theta}$, where $\delta_\theta$ denotes the degenerate belief on $\theta$. Let $E_v$ denote the projection of $\ext(\epi(v))$ on $\Delta(\Theta)$. We call the elements of $E_v$ the \textbf{extremal beliefs}.\footnote{Similar approaches have been used in \citet*{bergemann2015limits} and \citet{lipnowski2017simplifying}, where these objects are called ``extremal markets'' or  ``outer points.''} Let $K_v$ denote the set of \textbf{non-degenerate extremal beliefs}, that is, those extremal beliefs that do not belong to $\{\delta_\theta\}_{\theta\in\Theta}$. See \autoref{figure:interim-value} for an illustration when $|\Theta|=2$ and $|A|=3$.

\begin{figure}[htp]
\centering
	\begin{tikzpicture}[domain=0:3, scale=5, thick]

            \filldraw[orange!30] (0,0.95)--(0,0.81)--(1/3,0.61)--(2/3,0.61)--(1,0.81)--(1,0.95);
            
			\draw[<->] (0,1)node[left]{$v(\mu)$}--(0,0)node[below]{$0$}--(1.1,0)node[below]{$\mu$};
			\draw (1,0.03)--(1,0)node[below]{1};
				
			\draw[dashed] (0,0.8)--(1,0.2);
		    \draw[dashed] (0,0.6)--(1,0.6);
			\draw[dashed] (0,0.2)--(1,0.8);
				
			\draw[red!60,ultra thick] (0,0.81)--(1/3,0.61)--(2/3,0.61)--(1,0.81);

			\filldraw[blue] (1/3,0.61)circle (0.4pt) (2/3,0.61)circle (0.4pt);  
		\end{tikzpicture}   
    \caption{Interim value function and extremal beliefs when $|\Theta|=2$ and $|A|=3$. Each dashed line denotes the agent's interim payoff from an action, and their upper envelope (in red) is the interim value function. The shaded area represents the epigraph, and the blue dots are its non-degenerate extreme points, whose projections are non-degenerate extremal beliefs.}
    \label{figure:interim-value}
\end{figure}

The importance of the extremal beliefs is that they allow us to reduce \cref{eq:Nature-value-fuction} into a $|K_\nu|$-dimensional problem, as will be explored in the next sections.

\subsection{Value Vectors}
To reduce the dimension, we summarize a value function by its values on extremal beliefs $\mu\in E_v$. Formally, for each (normalized) value function $v$, evaluating it on $E_v$ yields $\boldsymbol{v}_v\in \mathbb{R}_-^{E_v}$ where $\boldsymbol{v}_v(\mu)=v(\mu)$ for all $\mu\in E_v$. Following our normalization of decision problems, we call a vector $\boldsymbol{v}\in \mathbb{R}_-^{E_v}$ a \textbf{value vector} if $\bv(\delta_\theta)=0$ for all $\theta\in\Theta$. Thus, the set of all value vectors is equivalent to $ \mathbb{R}_-^{K_v}$. Note that some value vector $\bv\in \mathbb{R}_-^{K_v}$ may not correspond to a value function, due to the convexity constraint of the value functions. However, the convexification $conv(\bv)$ is always a well-defined value function (see \cref{sec:convexification-concavification} for the formal definition of convexification). In particular, convexifying the value vector obtained from a value function yields the value function itself, i.e., $conv(\bv_v)=v$.

Recall that the value of an experiment given an interim value function $v$ is $V(P;v)=\int v(\mu) d\tau^P(\mu)$. With value vectors in hand, we can reformulate this value as a function of the finite-dimensional value vector. For any experiment $P$, define
\begin{align*}
    W(P;\cdot):\mathbb{R}_-^{K_v}&\rightarrow \mathbb{R}\\
   \bv&\mapsto V(P;conv(\bv))
\end{align*}
as the value of the experiment $P$ given a value vector $\bv$.

\subsection{Robustly Optimal Value}

The robustly optimal value can now be characterized as a concavification problem over value vectors.

\begin{thm}\label{thm:concavification}
    For any value function $v$ with a corresponding value vector $\bv_v$,
    \[V(P_1,...,P_m;v)=conc(\max\{W(P_1;\cdot),...,W(P_m;\cdot)\})(\bv_v).\]
\end{thm}
\begin{proof}
    See \cref{proof:theorem:general-state}.
\end{proof}

To see how \cref{thm:concavification} recovers the binary-state results, suppose that $|\Theta|=2$. Binary decomposition gives us that $v=\sum_{\ell=1}^{n-1} v_\ell$, where the $v_{\ell}$ denote the (normalized) interim value functions of the subproblems.  Since each $v_\ell$ has a unique kink, its value vector, defined using its own extremal beliefs, is one-dimensional, and the best-source value is linear.\footnote{To see the linearity, note that, for any $t>0$, we have $conv (t\bv)=t \:conv(\bv)$ and $V(P_j; tv)=tV(P_j;v)$. This implies that $W(P_j;\bv)$ is positively homogeneous of degree one in $\bv$. When $\bv$ is one dimensional, we can write $W(P_j;\bv)=-\bv W(P_j;-1)$, so $\max_j W(P_j;\bv)=-\bv \max_j W(P_j;-1)$, which is linear in $\bv$.} Therefore, concavification coincides with the best-source value, and $V(P_1,...,P_m;v_\ell)=\max_{j} V(P_j;v_\ell)$. This essentially deduces Theorem 1. The existence of Blackwell supremum (\cref{cor:value-supremum}) implies that $V(P_1,...,P_m;v)=V(\overline{P};v)$, which is linear in $v$. Therefore, 
\[V(P_1,...,P_m;v)=V\left(P_1,...,P_m;\sum_\ell v_\ell\right)=\sum_\ell V(P_1,...,P_m;v_\ell)=\sum_\ell \max_{j} V(P_j;v_\ell),\]
which is the interim-value formulation of Theorem 2.

More generally, the concavification characterization provides a way to construct robustly optimal strategies for general state spaces. Suppose that $(\lambda_\ell, \bv_\ell)_{\ell = 1}^{m} $ constitutes a solution to the concavification problem:
\[
\lambda_\ell \geq 0, \sum_{\ell = 1}^m \lambda_\ell = 1, \sum_{\ell = 1}^{m} \lambda_\ell \bv_\ell = \bv_v, \sum_{\ell = 1}^{m} \lambda_\ell W(P_\ell; \bv_\ell) = conc(\max\{W(P_1;\cdot),...,W(P_m;\cdot)\})(\bv_v).
\]
Now let $(A_\ell, u_\ell)$ be a finite action decision problem whose interim value function is $\lambda_\ell conv(\bv_\ell)$. 
It can be shown that
$\sum_{\ell = 1}^{m} \lambda_\ell conv(\bv_\ell) \leq v$
and so using the language of Section 3, $(A, u)$ is a decision problem that dominates $\bigoplus_{\ell = 1}^{m} (A_\ell, u_\ell)$.  The above equality essentially shows that 
\[
V(P_1, \ldots , P_m; (A, u)) = V\left(P_1, \ldots , P_m ; \bigoplus_{\ell = 1}^{m} (A_\ell, u_\ell)\right) = \sum_{\ell = 1}^{m} V(P_\ell; (A_\ell, u_\ell)).
\]
Aggregating the optimal strategies for these decomposed problems through a dominating map therefore yields a robustly optimal strategy for the original problem.

\begin{cor}
\label{cor:general-robust-strategy}
  Suppose that $\lambda^*_\ell\in \mathbb{R}_+, \bv_\ell^*\in \mathbb{R}_-^{K_v}$, $\ell=1,...,m$ are such that 
    \[
    \sum_{\ell = 1}^{m} \lambda_\ell^* = 1, \sum_{\ell = 1}^{m} \lambda_\ell^* \bv_\ell^* = \bv_v, \sum_{\ell = 1}^{m} \lambda_\ell^* W(P_\ell; \bv_\ell^*) = V(P_1, \ldots , P_m; \bv_v).
    \]
    Let $(A_\ell, u_\ell)$ be a finite-action decision problem whose induced value function is $conv(\lambda_\ell^* \bv_\ell^*)$ and $\sigma_\ell^*: Y_\ell \rightarrow \Delta(A_\ell)$ be an optimal strategy given Blackwell experiment $P_\ell$.  
    Then $(A, u)$ dominates $\bigoplus_{\ell = 1}^{m} (A_\ell, u_\ell)$ and for any dominating map, $f: \bigoplus_{\ell = 1}^{m} (A_\ell, u_\ell) \rightarrow \Delta(A)$, 
    \[
    \sigma^*_f(y_1, \ldots , y_m) = f(\sigma_1^*(y_1), \ldots , \sigma_m^*(y_m))
    \]
    is a robustly optimal strategy.
\end{cor}
\begin{proof}
    See \cref{app:proof-general-state-strategy}.
\end{proof}

\section{Behavioral Implications of Robustly Optimal Strategies}\label{section:behavioral-implications}

\subsection{Ignoring Information}

\cref{thm:concavification} offers two immediate corollaries that speak to which information sources can be ignored and how many of them are needed.

\begin{cor}\label{cor:bound}
   There exists a subset of information sources $\{P_j\}_{j\in J\subseteq \{1,...,m\}}$ with $|J|\leq |K_v|$, such that 
\[V(P_1,...,P_m;v)=V(\{P_j\}_{j\in J};v).\]
\end{cor}
\begin{proof}
    See \cref{app:proof:cor:bound}.
\end{proof}

\cref{cor:bound} establishes that a decision maker would never need to use more information sources than the number of nondegenerate extremal beliefs. The proof follows from the familiar idea that, in a $k$-dimensional concavification problem, the concavification value can be achieved by a convex combination of at most $k+1$ points. One subtlety here is that such an argument only gives us a bound of $|K_v|+1$. To derive the tighter bound $|K_v|$, we utilize the positive homogeneous property of the $W(P;\cdot)$ function to further reduce dimensions.

We show in \cref{appendix:bounding-K_v} that $|K_v|$ can be further bounded by a function of $|\Theta|$ and $|A|$:
\[|K_v|\leq \begin{pmatrix}
        |\Theta|+|A|+1- \left\lfloor \frac{|\Theta|+1}{2} \right\rfloor\\
        |A|+1
    \end{pmatrix}
   +\begin{pmatrix}
        |\Theta|+|A|+1- \left\lfloor \frac{|\Theta|+2}{2} \right\rfloor\\
        |A|+1
    \end{pmatrix}-2|\Theta|,
 \]
where $\begin{pmatrix}
    n\\k
\end{pmatrix}:=\frac{n!}{k!(n-k)!}$ denotes the binomial coefficient. 
For fixed $\Theta$, this bound grows polynomially in $|A|$ with degree $\lfloor \frac{|\Theta|}{2} \rfloor$. For fixed $A$, it grows on the order of $|\Theta|^{|A|+1}$. When $|\Theta|=2$, this bound reduces to $|A|-1$, the bound given in \cref{cor:number}. 

\begin{cor}\label{cor:never-best-experiment}
    Suppose that $V(P_\ell;v)\leq \max_{j\neq \ell} V(P_{j};v)$ for all $v$. Then $V(P_1,...,P_m;v)=V(P_1,...,P_{\ell-1},P_{\ell+1},...,P_m;v)$ for all $v$.
\end{cor}

\begin{proof}
    See \cref{app:proof:cor:safe-to-ignore}.
\end{proof}

\cref{cor:never-best-experiment} tells us that if an information source $P_{\ell}$ is never the unique best information source among $\{P_j\}_{j=1}^m$ for any decision problem, then it can be safely ignored. Note that being never the best information source is weaker than $P_{\ell}$ being Blackwell dominated by one of the other experiments, because the experiment that outperforms $P_\ell$ may depend on the particular decision problem.\footnote{\cite{cheng2023dominance} show that this condition is equivalent to $P_\ell$ being dominated by a convex combination of the rest of the experiments.}

\subsection{Aggregating Action Recommendations and Unanimity}

\cite{sobel2014relationship} considers a group who makes a decision by observing each individual's action recommendation and aggregating those into a single action. He notes that when the group aggregates recommendations as a Bayesian agent, group decisions are unlikely to be moderate, in the sense of being bounded by the lowest and highest recommendations. Moreover, he points out that even when all members recommend the same action, it may be optimal for the group to choose a different action. Here, we show that a group that behaves as a robustly optimal agent is moderate.

Consider a decision problem $(A, u)$ with a binary state space, $\{\theta_1, \theta_2\}$, normalized as in \cref{sec:decomposition}. Recall that, after normalization, actions can be ordered according to their optimality with respect to beliefs, as the clockwise order in \cref{figure-actionorder}. Each source, in isolation, results in a belief which elicits an action\,---\,that source's ``recommendation".

Now consider a Bayesian agent who sees a number of signals, believing that they are independent. If all signals have a higher relative likelihood of occurring in state $\theta_1$, the agent will become more and more sure that the state is indeed $\theta_1$, leading to an extreme action as a response. In contrast, we will show in this section that the robustly optimal strategies constructed in \cref{thm:binary-general} are moderate: they will never take an action that is more extreme than the ones elicited by each observed signal in isolation.

To formalize this claim, we extend the order of actions to mixed strategies. Let $B(\alpha)\subset [0,1]$ denote the set of beliefs that are consistent with interim optimality of an undominated mixed strategy $\alpha\in \Delta(A)$:
\[
B(\alpha) := \left\{ \mu(\theta_1) : \sum_\theta \mu(\theta) \rho(\theta,\alpha) = \max_{a \in A} \sum_\theta \mu(\theta) \rho( \theta,a) \right\}.
\]
We can interpret $B(\alpha)$ as the revealed posterior beliefs that rationalize an agent's choice of $\alpha$. Now, define the \textbf{belief order} as $\alpha \succeq \beta$ if $B(\alpha)$ dominates $B(\beta)$ in the strong set order\footnote{Given $S, T \subseteq [0,1]$, $S$ dominates $T$ in the strong set order if, for any $s\in S$ and $t\in T$, $\max\{s,t\}\in S$ and $\min \{s,t\}\in T$.}, so higher actions correspond to higher revealed beliefs. For $a\in A$, this order corresponds to the labeling in \cref{sec:decomposition}.

Now, suppose that the agent sees only a single signal $y_j$ from a single information source $P_j$. Let $a_j^*(y_j; A,u)$ denote the optimal action given the observation of $y_j$ in isolation, which we assume to be unique for technical convenience. We will interpret these actions as the recommended actions from individual information sources.

The following corollary of \cref{thm:binary-general} shows that there exists some robustly optimal strategy such that the realized actions lie between the recommended actions from the individual information sources.

\begin{cor}\label{cor:aggregate-expert-opinion}
    Let $(A, u)$ be a binary-state decision problem.  There exists some robustly optimal strategy $\sigma$, such that for all realizations of signals $\boldsymbol{y}$,
    \[
    \min_{j = 1, \ldots , m} a_j^*(y_j; A,u) \preceq \sigma(y_1, \ldots, y_m) \preceq \max_{j = 1, \ldots, m} a_j^*(y_j; A,u),
    \]
    where the minimum and maximum are with respect to the belief order. In other words, if $\underline{a} = \min_{j = 1, \ldots , m} a_j^*(y_j; A,u) $ and $\overline{a} = \max_{j = 1, \ldots , m} a_j^*(y_j; A,u) $,
    \[
    B(\underline{a}) \leq B(\sigma(y_1, \ldots , y_m)) \leq B(\overline{a}),
    \]
    where the inequalities are in the strong set order.
\end{cor}
\begin{proof}
    See \cref{sec:unanimity-proof}.
\end{proof}

This corollary says that the revealed beliefs of the constructed robustly optimal strategy always lie between the most extreme beliefs from the individual signals. Another implication is that the robustly optimal strategy satisfies \emph{unanimity}: if under a signal realization $\mathbf{y}$ all marginal information sources agree on the optimal action, i.e.
    \[
    a = a_1^*(y_1) = \cdots = a_m^*(y_m),     
    \]
    then the robustly optimal strategy will prescribe that action. While unanimity seems like an appealing and natural property, it does not necessarily hold when the joint information structure is perfectly known to the decision maker.

\section{Discussion}\label{section:discussion}
This section discusses some extensions of our model. \cref{sec:extension:common_source} discusses the implications of additional knowledge about the correlation structure.  
\cref{sec:extension:marginal} shows that \cref{thm:binarybinary} extends to scenarios where the agent has even less knowledge about the information sources\,---\,introducing an additional layer of ambiguity regarding the marginal experiments. \cref{sec:expert-opinion} considers the case where the information sources available to the agent have already been processed by experts.

\subsection{Knowledge of Correlation}\label{sec:extension:common_source}

\subsubsection{Common Origin}

A natural reason for the correlation among multiple information sources is a shared, common origin. For instance, financial consultants may base their recommendations on the same dataset, leading to correlations among their recommendations. If we know that a common origin is the \textit{only} possible channel generating the correlation among information sources, does this additional knowledge help restrict the presumed set of correlations? In other words, what types of correlation structures can be rationalized by sharing a common origin?

Formally, we say a joint experiment $P\in\mathcal{J}(P_1,...,P_m)$ is \textit{rationalizable by a common origin} if there exists $Q:\Theta\rightarrow \Delta X$ and a collection,  $\left\{ \gamma_j : X \rightarrow \Delta(Y_j) \right\}_{j}$, such that 
\[P(y_1,...,y_m|\theta)=\sum_x \prod_{j=1}^m \gamma_j(y_j|x)Q(x|\theta).\]
The interpretation is that $Q$ is the common but unknown origin, and the experiments $P_1,...,P_j$ are generated by \emph{independent} garblings of signals from $Q$.

We have the following straightforward observation.
\begin{obs}
Every $P\in \mathcal{J}(P_1,...,P_m)$ is rationalizable by a common origin.
\end{obs}
To see why, note that we can let the common source $Q$ be $P$ itself, and the garblings $\gamma_j$ be the deterministic functions that project each vector $y_1,...,y_m$ onto $y_j$. Therefore, sharing a common origin does not exclude any possible correlation.

\subsubsection{Partial Knowledge of Correlations}\label{sec:partial-knowledge}

In certain situations, an agent may understand the correlation among some information sources, even if she does not comprehend all of them. For example, in medical diagnoses, older technologies such as X-rays and MRI have well-understood correlations. On the other hand, newer technologies, such as genetic sequencing, may have correlations with these traditional tests that are not yet fully understood.

In the context of our model, such knowledge can be modeled as imposing additional constraints on the set of conceived joint experiments $\mathcal{J}(P_1,...,P_m)$. A simple case in which our results extend in a straightforward manner is the following:  Suppose that there is a partition, $\Pi = \{S_1, \ldots , S_k\}$, of $\{1,2, \ldots , m\}$ such that for all $S \in \Pi$, the agent knows that joint distribution over signals in $S$ is given by:
\[
\sum_{
y_{-S}} P(y_S,y_{-S} | \theta) = P_S(y_S | \theta).
\]
Then the set of conceived information structures is given by:
\[
\left\{ P \in \mathcal{J}(P_1, \ldots , P_m) : \sum_{y_{-S}} P(y_S,y_{-S} | \theta) = P_S(y_S | \theta), \forall \theta,S\in \Pi, y_S\in Y_S  \right\}.
\]
But note that we could treat each joint experiment, $P_{S_1}, \ldots , P_{S_k}$, as a separate marginal experiment and use our previous analysis.

However, our analysis does not immediately extend to other, more complex situations. In particular, when the knowledge of the correlations spans across non-disjoint subsets, the set of possible joint experiments cannot be treated by replacing a subset of experiments with a single experiment, and our existing results no longer apply. To illustrate, suppose there are three information sources, $\{P_1,P_2,P_3\}$, and that the agent knows that
$P_1$ and $P_2$ are correlated according to $P_{12}:\Theta\rightarrow \Delta (Y_1\times Y_2)$, and that $P_2$ and $P_3$ are correlated according to $P_{23}:\Theta\rightarrow \Delta(Y_2\times Y_3)$. The set of feasible joint experiments would be
\[
\left\{P:\Theta\rightarrow \Delta(Y_1\times Y_2\times Y_3)  \;\middle\vert\;
   \begin{array}{@{}l@{}}
 \sum\limits_{y_3}P(y_1,y_2,y_3|\theta)=P_{12}(y_1,y_2|\theta),\forall \theta,y_1,y_2 \\
  \sum\limits_{y_1}P(y_1,y_2,y_3|\theta)=P_{23}(y_2,y_3|\theta),\forall \theta,y_2,y_3  
   \end{array}
\right\}. 
\]
An interesting direction for future research would be to consider general restrictions on the set of correlation structures derived from causality diagrams (see \cite{pearl2009causality} and \cite{spiegler2016bayesian}).

\subsection{Ambiguity about Marginals}\label{sec:extension:marginal}

Our model so far assumes that the agent understands each information source precisely; that is, she knows $P_j$ for $j=1,...,m$. In this section, we extend our model to allow for additional ambiguity about the marginal information sources. 

Let $\mathcal{P}_j$ denote the set of possible marginal experiments for information source $j=1,...,m$. Let all $P_j\in \mathcal{P}_j$ have the same finite signal space $Y_j$. In addition, each $\mathcal{P}_j$ is assumed to be convex. That is, if $P_j:\Theta\rightarrow \Delta(Y_j)$ and $P'_j:\Theta\rightarrow \Delta(Y_j)$ are both in $\mathcal{P}_j$, then for any $\lambda\in(0,1)$, $ Q_\lambda:\Theta \rightarrow \Delta (Y_j)$ defined as $\theta \mapsto \lambda P_j(\cdot|\theta)+ (1-\lambda) P'_j(\cdot|\theta)$ is also in $\mathcal{P}_j$.

The agent conceives of the following set of joint experiments:
	\[
	\mathcal{J}(\mathcal{P}_1,...,\mathcal{P}_m)=\left\{
	P:\Theta\rightarrow \Delta(\mathbf{Y}): \exists P_j\in \mathcal{P}_j, 
	\sum_{-j}P(y_1,\ldots,y_m|\theta)=P_j(y_j|\theta)\text{ for all }\theta,j,y_j
	\right\}.
	\]
The agent's decision problem is similarly defined:
\[V(\mathcal{P}_1, \ldots , \mathcal{P}_m):= \max_{\sigma:\mathbf{Y}\rightarrow \Delta (A)} \min_{P\in \mathcal{J}(\mathcal{P}_1,...,\mathcal{P}_m)} \sum_{\theta\in\Theta} \sum_{\mathbf{y}\in \mathbf{Y}} P(\mathbf{y}|\theta)u(\theta,\sigma(\mathbf{y})).\]

We show that the prediction in \cref{thm:binarybinary} is robust to this additional layer of ambiguity.

\begin{prop}\label{prop:ex:marginals}
        For all $(A,u)$ with $|A| = |\Theta| = 2$,
    \[
    V(\mathcal{P}_1, \ldots , \mathcal{P}_m) =  \max_{j = 1, \ldots , m}V(\mathcal{P}_j ).
    \]
\end{prop}
\begin{proof}
    See \cref{proof:ex:marginals}.
\end{proof}

\subsection{Aggregating Experts' Opinions}\label{sec:expert-opinion}

In certain instances, an agent may not have the expertise to process raw information sources. Instead, she may rely on experts who understand the information sources to offer their opinions, such as in the form of beliefs (e.g., doctors offering beliefs on the likelihood of a successful surgery) or action recommendations (e.g., financial consultants providing investment recommendations).

Reporting beliefs and offering action recommendations can both be viewed as garblings of the original, raw information sources. For any given information source $P_j:\Theta\rightarrow \Delta(Y_j)$, we call the induced \textit{belief information structure}, denoted by $B_{P_j}:\Theta\rightarrow \Delta(\Delta (\Theta))$, as the information structure derived by garbling each signal into the corresponding induced beliefs. In addition, we call the induced \textit{recommendation information structure}, denoted by $R_{P_j}:\Theta\rightarrow \Delta A$, as the information structure derived by a garbling $\sigma_j^*$, given by an optimal strategy:
\[\sigma_j^*\in\argmax_{\sigma_j:  Y_j \rightarrow  A} \sum_{\theta,y_j}P_j(y_j|\theta) u(\theta,\sigma_j(y_j)).\]
Note that, in contrast to the belief information structure, the recommendation information structure depends on the decision problem.

When the agent has access to a single source of information, garbling information through reporting beliefs or action recommendations does not hurt, that is, $V(P_j)=V(B_{P_j})=V(R_{P_j})$ for any $j$. This is because beliefs and action recommendations contain all the information needed to make an optimal decision. 

When multiple information sources are available, garbling information by reporting only beliefs or recommendations could potentially hurt payoffs because some of the lost information, which is not useful on its own, could become valuable when combined with other sources. This begs the question of whether the agent could still achieve the same value as if she had access to the raw information sources. In other words, does 
\[V(P_1,...,P_m)=V(B_{P_1},...,B_{P_m})=V(R_{P_1},...,R_{P_m})
\]
hold when $m>1$?

First, it is indeed the case that $V(P_1,...,P_m)=V(B_{P_1},...,B_{P_m})$: since $P_j$ is Blackwell equivalent to $B_{P_j}$ for all $j$, \cref{lem:relaxed-problem} implies the values $V(P_1,...,P_m)$ and $V(B_{P_1},...,B_{P_m})$ must be equal. The relationship between $V(R_{P_1},...,R_{P_m})$ and $V(P_1,...,P_m)$ is more interesting: when $|\Theta|=2$, these values coincide, but in general, we could have $V(R_{P_1},...,R_{P_m})<V(P_1,...,P_m)$.

\begin{prop}\label{prop: binary-recommendation}
    When $|\Theta|=2$, for any $(A,u)$,
    \[V(P_1,...,P_m)=V(R_{P_1},...,R_{P_m}).\]
\end{prop}
\begin{proof}
See \cref{proof: binary-recommendation}.
\end{proof}

When there are three or more states, the recommendation information structure could generate a strictly lower value than the raw information structure. This can be seen by revisiting \cref{ex:threestaets}. Recall that in the example, under both $P_X$ and $P_Y$, $a=1$ is the unique optimal action to any signal realization. Therefore, both $R_{P_X}$ and $R_{P_Y}$ are uninformative experiments, and so $V(R_{P_X},R_{P_Y})=1-0.9+1=1.1$. By contrast, the agent obtains perfect information when observing the raw information structures, and thus $V(P_X,P_Y)=1+0+1=2>V(R_{P_X},R_{P_Y})$.

\section{Conclusion}\label{section:conclusion}

Our findings have both normative and positive implications.

In a normative sense, there are settings in which decisions have to be made in highly uncertain environments, where correlations between information sources are hard to ascertain, and worst-case performance is of primary concern. Examples include the design of artificial intelligence and robotic agents, as well as the formulation of public guidelines. Our results can help guide such decision-makers on how to construct robustly optimal strategies, emphasizing the usefulness of specialized information sources.

In a positive sense, our work offers an alternative rationale for information neglect, with implications that differ from existing explanations. For example, in models of rational inattention (see \citet*{mackowiak2023rational} for a survey), higher stakes lead agents to acquire and use more information. In contrast, in our model, multiplying the utility function by any positive constant does not alter the set of information sources attended to. This distinction helps account for why information is ignored even in high-stakes decision problems.\footnote{For example, \citet{olver2020second} found that only 16.1\% of patients sought a second opinion about their cancer treatment.}

\newpage
    
\bibliography{Merging}

\begin{thebibliography}{36}
\providecommand{\natexlab}[1]{#1}
\providecommand{\url}[1]{\texttt{#1}}
\expandafter\ifx\csname urlstyle\endcsname\relax
  \providecommand{\doi}[1]{doi: #1}\else
  \providecommand{\doi}{doi: \begingroup \urlstyle{rm}\Url}\fi

\bibitem[Arieli et~al.(2018)Arieli, Babichenko, and Smorodinsky]{ArieliE12135}
Itai Arieli, Yakov Babichenko, and Rann Smorodinsky.
\newblock Robust forecast aggregation.
\newblock \emph{Proceedings of the National Academy of Sciences}, 115\penalty0 (52):\penalty0 E12135--E12143, 2018.
\newblock ISSN 0027-8424.
\newblock \doi{10.1073/pnas.1813934115}.
\newblock URL \url{https://www.pnas.org/content/115/52/E12135}.

\bibitem[Arieli et~al.(2025)Arieli, Babichenko, Talgam-Cohen, and Zabarnyi]{arieli2025random}
Itai Arieli, Yakov Babichenko, Inbal Talgam-Cohen, and Konstantin Zabarnyi.
\newblock A random dictator is all you need.
\newblock \emph{American Economic Journal: Microeconomics}, 17\penalty0 (1):\penalty0 66--96, 2025.

\bibitem[Ball and Kattwinkel(2026)]{ball2026robust}
Ian Ball and Deniz Kattwinkel.
\newblock Robust robustness.
\newblock \emph{arXiv preprint arXiv:2408.16898}, 2026.

\bibitem[Bergemann et~al.(2015)Bergemann, Brooks, and Morris]{bergemann2015limits}
Dirk Bergemann, Benjamin Brooks, and Stephen Morris.
\newblock The limits of price discrimination.
\newblock \emph{American Economic Review}, 105\penalty0 (3):\penalty0 921--957, 2015.

\bibitem[Bertschinger and Rauh(2014)]{bertschinger2014blackwell}
Nils Bertschinger and Johannes Rauh.
\newblock The blackwell relation defines no lattice.
\newblock \emph{2014 IEEE International Symposium on Information Theory}, pages 2479--2483, 2014.

\bibitem[Blackwell(1953)]{blackwell1953equivalent}
David Blackwell.
\newblock Equivalent comparisons of experiments.
\newblock \emph{The annals of mathematical statistics}, pages 265--272, 1953.

\bibitem[B{\"o}rgers et~al.(2013)B{\"o}rgers, Hernando-Veciana, and Kr{\"a}hmer]{borgers2013signals}
Tilman B{\"o}rgers, Angel Hernando-Veciana, and Daniel Kr{\"a}hmer.
\newblock When are signals complements or substitutes?
\newblock \emph{Journal of Economic Theory}, 148\penalty0 (1):\penalty0 165--195, 2013.

\bibitem[B{\"o}rgers et~al.(2025)B{\"o}rgers, Li, and Wang]{li2024undominated}
Tilman B{\"o}rgers, Jiangtao Li, and Kexin Wang.
\newblock Undominated mechanisms.
\newblock \emph{Available at SSRN 5007659}, 2025.

\bibitem[Brooks et~al.(2024)Brooks, Frankel, and Kamenica]{brooks2024comparisons}
Benjamin Brooks, Alexander Frankel, and Emir Kamenica.
\newblock Comparisons of signals.
\newblock \emph{American Economic Review (Forthcoming)}, 2024.

\bibitem[Carroll(2017)]{carroll2017robustness}
Gabriel Carroll.
\newblock Robustness and separation in multidimensional screening.
\newblock \emph{Econometrica}, 85\penalty0 (2):\penalty0 453--488, 2017.

\bibitem[Cheng and B\"orgers(2024)]{cheng2023dominance}
Xienan Cheng and Tilman B\"orgers.
\newblock Dominance and optimality.
\newblock \emph{Working paper}, 2024.

\bibitem[Cheng and B{\"o}rgers(2024)]{cheng2024diversity}
Xienan Cheng and Tilman B{\"o}rgers.
\newblock Diversity, disagreement, and information aggregation.
\newblock \emph{Working Paper}, 2024.

\bibitem[de~Oliveira(2018)]{de2018blackwell}
Henrique de~Oliveira.
\newblock Blackwell's informativeness theorem using diagrams.
\newblock \emph{Games and Economic Behavior}, 109:\penalty0 126--131, 2018.

\bibitem[Dworczak and Pavan(2022)]{dworczak2022preparing}
Piotr Dworczak and Alessandro Pavan.
\newblock Preparing for the worst but hoping for the best: Robust (bayesian) persuasion.
\newblock \emph{Econometrica}, 90\penalty0 (5):\penalty0 2017--2051, 2022.

\bibitem[Gilboa and Schmeidler(1989)]{gilboa1989maxmin}
Itzhak Gilboa and David Schmeidler.
\newblock Maxmin expected utility with a non-unique prior.
\newblock \emph{Journal of Mathematical Economics}, 18:\penalty0 141--153, 1989.

\bibitem[Gr{\"u}nbaum(2003)]{grunbaum2003convex}
Branko Gr{\"u}nbaum.
\newblock \emph{Convex Polytopes}, volume 221.
\newblock Springer Science \& Business Media, 2003.

\bibitem[Handel and Schwartzstein(2018)]{handel2018frictions}
Benjamin Handel and Joshua Schwartzstein.
\newblock Frictions or mental gaps: what's behind the information we (don't) use and when do we care?
\newblock \emph{Journal of Economic Perspectives}, 32\penalty0 (1):\penalty0 155--78, 2018.

\bibitem[He and Li(2020)]{HeLi2020}
Wei He and Jiangtao Li.
\newblock Correlation-robust auction design.
\newblock \emph{Working paper}, 2020.

\bibitem[Hiriart-Urruty and Lemar{\'e}chal(1996)]{hiriart1996convex}
Jean-Baptiste Hiriart-Urruty and Claude Lemar{\'e}chal.
\newblock \emph{Convex analysis and minimization algorithms I: Fundamentals}, volume 305.
\newblock Springer science \& business media, 1996.

\bibitem[Ichihashi(2021)]{ichihashi2021economics}
Shota Ichihashi.
\newblock The economics of data externalities.
\newblock \emph{Journal of Economic Theory}, 196:\penalty0 105316, 2021.

\bibitem[Kambhampati(2025)]{kambhampati2025proper}
Ashwin Kambhampati.
\newblock Proper robustness and the efficiency of monopoly screening.
\newblock \emph{Unpublished manuscript}, 2025.

\bibitem[Kertz and R{\"o}sler(1992)]{kertz1992stochastic}
Robert~P Kertz and Uwe R{\"o}sler.
\newblock Stochastic and convex orders and lattices of probability measures, with a martingale interpretation.
\newblock \emph{Israel Journal of Mathematics}, 77:\penalty0 129--164, 1992.

\bibitem[Levy and Razin(2020)]{levy2020combining}
Gilat Levy and Ronny Razin.
\newblock Combining forecasts in the presence of ambiguity over correlation structures.
\newblock \emph{Journal of Economic Theory}, page 105075, 2020.

\bibitem[Liang and Mu(2020)]{liang2020complementary}
Annie Liang and Xiaosheng Mu.
\newblock Complementary information and learning traps.
\newblock \emph{The Quarterly Journal of Economics}, 135\penalty0 (1):\penalty0 389--448, 2020.

\bibitem[Liang et~al.(2022)Liang, Mu, and Syrgkanis]{liang2022dynamically}
Annie Liang, Xiaosheng Mu, and Vasilis Syrgkanis.
\newblock Dynamically aggregating diverse information.
\newblock \emph{Econometrica}, 90\penalty0 (1):\penalty0 47--80, 2022.

\bibitem[Lipnowski and Mathevet(2017)]{lipnowski2017simplifying}
Elliot Lipnowski and Laurent Mathevet.
\newblock Simplifying bayesian persuasion.
\newblock \emph{Unpublished Paper, Columbia University.[642]}, 2017.

\bibitem[Ma{\'c}kowiak et~al.(2023)Ma{\'c}kowiak, Mat{\v{e}}jka, and Wiederholt]{mackowiak2023rational}
Bartosz Ma{\'c}kowiak, Filip Mat{\v{e}}jka, and Mirko Wiederholt.
\newblock Rational inattention: A review.
\newblock \emph{Journal of Economic Literature}, 61\penalty0 (1):\penalty0 226--273, 2023.

\bibitem[Olver et~al.(2020)Olver, Carey, Bryant, Boyes, Evans, and Sanson-Fisher]{olver2020second}
Ian Olver, Mariko Carey, Jamie Bryant, Allison Boyes, Tiffany Evans, and Rob Sanson-Fisher.
\newblock Second opinions in medical oncology.
\newblock \emph{BMC Palliative Care}, 19:\penalty0 1--6, 2020.

\bibitem[Pearl(2009)]{pearl2009causality}
Judea Pearl.
\newblock \emph{Causality}.
\newblock Cambridge university press, 2009.

\bibitem[Rockafellar(1970)]{rockafellar1970convex}
R~Tyrrell Rockafellar.
\newblock \emph{Convex analysis}, volume~36.
\newblock Princeton university press, 1970.

\bibitem[Sion(1958)]{sion1958general}
Maurice Sion.
\newblock On general minimax theorems.
\newblock \emph{Pacific Journal of Mathematic}, 8:\penalty0 171--176, 1958.

\bibitem[Sobel(2014)]{sobel2014relationship}
Joel Sobel.
\newblock On the relationship between individual and group decisions.
\newblock \emph{Theoretical Economics}, 9\penalty0 (1):\penalty0 163--185, 2014.

\bibitem[Spiegler(2016)]{spiegler2016bayesian}
Ran Spiegler.
\newblock Bayesian networks and boundedly rational expectations.
\newblock \emph{The Quarterly Journal of Economics}, 131\penalty0 (3):\penalty0 1243--1290, 2016.

\bibitem[Wald(1950)]{wald1950statistical}
Abraham Wald.
\newblock \emph{Statistical decision functions.}
\newblock Wiley, 1950.

\bibitem[Zhang(2021)]{zhang2021correlation}
Wanchang Zhang.
\newblock Correlation-robust optimal auctions.
\newblock \emph{arXiv preprint arXiv:2105.04697}, 2021.

\bibitem[Ziegler(2012)]{ziegler2012lectures}
G{\"u}nter~M Ziegler.
\newblock \emph{Lectures on polytopes}, volume 152.
\newblock Springer Science \& Business Media, 2012.

\end{thebibliography}
\bibliographystyle{plainnat}

\newpage
\appendix
\addtocontents{toc}{\protect\setcounter{tocdepth}{1}}

\section{Appendix}

\subsection{Proof of \autoref{lem:relaxed-problem}}\label{proof:blackwellsup}
\begin{proof}

It suffices to show that for any $Q\in \mathcal{D}(P_1,...,P_m)$, there exists $P\in \mathcal{J}(P_1,...,P_m)$ such that $P$ is Blackwell dominated by $Q$.

Take any $Q\in \mathcal{D}(P_1,...,P_m)$ and let $X$ be the signal space of $Q$. By \hyperref[thm:Blackwell]{Blackwell's Theorem}, there exist  $\gamma_j:X\rightarrow \Delta Y_j$ such that for each $j$,
\[P_j(y_j|\theta)=\sum_x\gamma_j(y_j|x)Q(x|\theta).\]
Define the following joint Blackwell experiment $P:\Theta\rightarrow \Delta(Y_1\times ...\times Y_m)$:
\begin{equation*}\label{garbling}
P(y_1,...,y_m|\theta)=\sum_x \prod_{j=1}^m \gamma_j(y_j|x)Q(x|\theta).
\end{equation*}
Clearly, $P\in \mathcal{J}(P_1,...,P_m)$ because $\sum_{y_{-j}}P(y_1,...,y_m|\theta)=\sum_x\gamma_j(y_j|x)Q(x|\theta)=P_j(y_j|\theta)$. Moreover, $\prod_{j=1}^m \gamma_j(y_j|x)$ defines a garbling, so $P$ is Blackwell dominated by $Q$.
\end{proof}

\subsection{Proof of \autoref{prop:separable}}\label{sec:separableproof}

\begin{proof}
To reduce notation, let's write $\mathbb{E}_{P}\left[u_\ell(\theta,\sigma_\ell)\right]=\sum_{\theta,\mathbf{y}} u_\ell(\theta,\sigma_\ell(\mathbf{y}))P(\mathbf{y}|\theta).$
Since $\sigma=\left(\sigma_\ell\right)_{\ell=1}^k$ is a feasible strategy,

\vspace{-0.3in}
\begin{align*}
V\Big(P_1, \ldots , P_m ; \bigoplus_{\ell = 1}^{k} (A_\ell, u_\ell)\Big) &\geq \min_{P \in \mathcal{J}(P_1, \ldots , P_m)} \sum_{\ell=1}^k  \mathbb{E}_{P}\left[u_\ell(\theta,\sigma_\ell)\right] \\
&\geq  \sum_{\ell=1}^k  \min_{P \in \mathcal{J}(P_1, \ldots , P_m)}\mathbb{E}_{P}[ u_\ell(\theta,\sigma_\ell)]\\
&= \sum_{\ell=1}^k \max_{j=1,...,m} V(P_j;(A_\ell,u_\ell)).
\end{align*}
Moreover, by \cref{thm:binarybinary} and \cref{value},
\begin{align*}
\sum_{\ell=1}^k \max_{j=1,...,m} V(P_j;(A_\ell,u_\ell)) &= \sum_{\ell=1}^k V(\overline{P}(P_1,...,P_m);(A_\ell, u_\ell)) \\
&= V\left(\overline{P}(P_1, \ldots , P_m); \bigoplus_{\ell = 1}^{k} (A_\ell, u_\ell)\right) \\
&\geq V\left(P_1, \ldots , P_m; \bigoplus_{\ell = 1}^{k} (A_\ell, u_\ell)\right).
\end{align*}
Together, these inequalities prove our claim that 
\begin{align*}
V\left(P_1, \ldots , P_m ; \bigoplus_{\ell = 1}^{k} (A_\ell, u_\ell)\right) 
= \sum_{\ell=1}^k \max_{j=1,...,m} V(P_j;(A_\ell,u_\ell))
\end{align*}
and that $\sigma$ is a robustly optimal strategy.

\end{proof}

 \subsection{Proof of \autoref{lem:equivalence}}\label{proof:equivalence}

\begin{proof}

Consider the binary decomposition $\bigoplus_{\ell=1}^{n-1}(A_{\ell},u_\ell)$. We prove that for any ${\boldsymbol\delta}\in\{0,1\}^{n-1}$, $\sum_{\ell=1}^{n-1} {\boldsymbol\delta}_\ell u_\ell(\cdot,1)\in \po(A,u)$. 

Suppose, by way of contradiction, that there exists ${\boldsymbol\delta}\in\{0,1\}^{n-1}$ for which 
$u^*:= \sum_{\ell=1}^{n-1} {\boldsymbol\delta}_\ell u_\ell(\cdot,1) \notin \mathcal{H}(A, u).$ Since $\mathcal{H}(A, u)$ is convex and closed, we can strictly separate it from the singleton $u^*$(Corollary 11.4.2 of \citet{rockafellar1970convex}), i.e. there exists $\lambda \in \mathbb{R}^2 \setminus \{(0,0)\}$ such that
\begin{align}\label{eqn:sep}
\lambda \cdot u^* > \sup_{v \in \mathcal{H}(A, u)} \lambda \cdot v.
\end{align}
Note that $\lambda \geq 0$ since otherwise $\sup_{v \in \mathcal{H}(A, u)} \lambda \cdot v = + \infty$.

From the ordering of the actions and the binary decomposition, $u_\ell(\theta_2,1)\slash u_\ell(\theta_1,1)$ is decreasing in $\ell$. Therefore, for any $\ell' > \ell$,
\[
\lambda \cdot u_\ell( \cdot,1)  \leq  0 \Longrightarrow \lambda \cdot u_{\ell'}(\cdot,1) \leq  0. 
\]
So there exists $\ell^*$ such that $\lambda \cdot u_\ell( \cdot,1)>0$ for $\ell<\ell^*$ and $\lambda \cdot u_\ell( \cdot,1)\leq 0$ for $\ell\geq \ell^*$.

\begin{figure}[htp]
	\centering

	\begin{tikzpicture}[domain=0:3, scale=4.8, ultra thick]

    \fill[orange!10] (0.2,0.1)--(1,0.1)--(1,0.2)--(0.9,0.55)--(0.7,0.8)--(0.3,1)--(0.2,1)--(0.2,0.1);

    \filldraw[blue] (1,0.2)circle (0.1pt)node[right]{$\bu(a_4)$};
    \filldraw[blue] (0.9,0.55)circle (0.1pt)node[left]{$\bu(a_3)$};
    \filldraw[blue] (0.7,0.8)circle (0.1pt)node[left]{$\bu(a_2)$};
    \filldraw[blue] (0.3,1)circle (0.1pt)node[right,yshift=2]{$\bu(a_1)$};
    
    \filldraw[red] (1.05,0.6)circle (0.1pt)node[below,xshift=8,yshift=4]{$u^*$};
    
    \draw[red] (1.1,0.3)--(0.75,1) ;
    \draw[red,->] (0.925,0.65)--(1.125,0.75)node[above,xshift=-6]{$\lambda$};
    
	\draw[dotted] (1,0.2)--(0.9,0.55);
	\draw[dotted] (0.9,0.55)--(0.7,0.8);
	\draw[dotted] (0.7,0.8)--(0.3,1);

    \draw[orange] (0.45,0.5)node{$\po(A,u)$};
	
	\end{tikzpicture}   
	\caption{}
   \label{}
\end{figure}
Thus
\[\max_{{\boldsymbol\delta}'\in \{0,1\}^{n-1}}\sum_{\ell=1}^{n-1} \lambda \cdot {\boldsymbol\delta}'_\ell u_\ell(\cdot,1)\]
is solved by choosing ${\boldsymbol\delta}'_\ell=1$ for $\ell<\ell^*$ and ${\boldsymbol\delta}'_\ell=0$ for $\ell\geq\ell^*$. Hence 
\[
\lambda\cdot u(\cdot,a_{\ell^*})=\lambda \cdot \sum_{\ell=1}^{\ell^*-1}u_\ell(\cdot,1)\geq \lambda \cdot \sum_{\ell=1}^{n-1} {\boldsymbol\delta}_\ell u_\ell(\cdot,1)=\lambda \cdot u^*.
\]
But $u(\cdot,a_{\ell^*}) \in \mathcal{H}(A, u)$, contradicting \eqref{eqn:sep}.
\end{proof}

\subsection{Proof of \autoref{thm:binary-general}}\label{sec:binarygenproof}

\begin{proof}
    From \autoref{lem:equivalence}, $(A,u)$ is equivalent to $\bigoplus_{\ell=1}^{n-1}(A_\ell,u_\ell)$, so
    \[V(P_1,...,P_m;(A,u))=V\left(P_1,...,P_m;\bigoplus_{\ell=1}^{n-1}(A_\ell,u_\ell)\right)=\sum_{\ell = 1}^{n-1} \max_{j = 1, \ldots , m} V(P_j; (A_\ell, u_\ell)),\]
    where the second equality follows from   \autoref{prop:separable}. This establishes the first statement of the theorem.

    By the definition of dominating maps, $u(\cdot,\sigma_f^*(\mathbf{y}))\geq \sum_{\ell=1}^{n-1} u_\ell(\cdot,\sigma_\ell(\mathbf{y}))$ for each $\mathbf{y}$.
    For any $P\in\mathcal{J}(P_1,...,P_m)$,
    \begin{align*}
         \mathbb{E}_{P}\left[ u(\theta,\sigma_f^*(\mathbf{y}))\right]&\geq \mathbb{E}_{P}\left[\sum_{\ell=1}^{n-1} u_\ell (\theta,\sigma_\ell(\mathbf{y}))\right]\\
        &=V\left(P_1,...,P_m;\bigoplus_{\ell=1}^{n-1}(A_\ell,u_\ell)\right)\\
        &=V\left(P_1,...,P_m;(A,u)\right)
    \end{align*}
    where the second line follows from \autoref{prop:separable} and the third line follows from \autoref{lem:equivalence}.
  So $\sigma^*_f$ is a robustly optimal strategy.
\end{proof}

\subsection{Proof of \autoref{thm:concavification}}\label{proof:concavification} \label{proof:theorem:general-state}

\subsubsection{Posterior Polytopes}

Let $\succeq_{MPS}$ denote the mean-preserving spread order on posterior distributions. For any experiment $P$, define $\mathcal{E}_v(P)=\{\tau\in \Delta(\Delta\Theta)\ |\ \tau\succeq_{\text{MPS}}\tau^P \text{ and supp}(\tau)\subseteq E_v\}$ as the set of posterior distributions that are mean-preserving spreads of $\tau^P$ and whose supports are contained in the set of extremal beliefs $E_v$.  This set is non-empty because it always contains the posterior distribution induced by a perfectly informative experiment.

Note that $\mathcal{E}_v(P)$ resides in a $|K_v|$-dimensional space, since by the martingale constraint, the probabilities that a posterior distribution assigns to the $|K_v|$ nondegenerate extremal beliefs pin down the probabilities assigned to the degenerate extremal beliefs.\footnote{Formally, for any posterior distribution $\tau$ supported on $E_v$ with a mean $\mu_0$, we must have $\sum_{\mu\in K_v} \tau(\mu)\leq 1$ and $\tau(\delta_\theta)=\mu_0(\theta)-\sum_{\mu\in K_v} \tau(\mu)\mu(\theta)$ for all $\theta$. Conversely, for any $t\in [0,1]^{K_v}$ such that $\sum_{\mu\in K_v} t(\mu)\leq 1$ and $\mu_0(\theta)-\sum_{\mu\in K_v} t(\mu)\mu(\theta)\geq 0$ for all $\theta$, by letting $\tau(\mu)=t(\mu)$ for $\mu\in K_v$ and $\tau(\delta_\theta)=\mu_0(\theta)-\sum_{\mu\in K_v} t(\mu)\mu(\theta)$ for all $\theta$, $\tau$ is a well-defined posterior distribution.} In addition, because the mean-preserving spread constraints are linear, $\mathcal{E}_v(P)$ is a $|K_v|$-dimensional polytope. Hence, we call $\mathcal{E}_v(P)$ the \emph{posterior polytope} induced by experiment $P$. We will sometimes use the vector notation $\boldsymbol{\tau}\in \mathbb{R}^{|K_v|}_+$ to denote a typical element in $\mathcal{E}_v(P)$, whose components are probabilities assigned to the nondegenerate extremal beliefs.

We first present a lemma that allows us to simplify Nature's problem into choosing elements in the intersection of the posterior polytopes induced by the marginal experiments.
\begin{lem}\label{lem:value-extremal-beliefs}
    \[V(P_1,...,P_m;v)=\min_{\tau\succeq_{\text{MPS}}\,\tau_{P_1},...,\tau_{P_m}}\int  v(\mu) \tau( d\mu)=\min_{\tau\in \cap_{j=1}^m \mathcal{E}_v(P_j)} \sum_{\mu\in E_v} \tau(\mu)v(\mu).\]
\end{lem}
\begin{proof}
    From \cref{eq:Nature-value-fuction},  we have $V(P_1,...,P_m;v)=\min_{P\in \mathcal{J}(P_1,...,P_m)}\int v(\mu)\tau^P(d\mu).$ By \cref{lem:relaxed-problem}, Nature's problem can be relaxed to choosing an experiment that Blackwell dominates all marginal experiments; that is, $V(P_1,...,P_m;v)=\min_{P\in D(P_1,...,P_m)}\int v(\mu)\tau^P(d\mu).$ A posterior distribution can be induced by a Blackwell experiment that dominates all marginals if and only if it is a mean-preserving spread of the posterior distributions induced by each marginal experiment \citep{blackwell1953equivalent}. This yields the first equality in the lemma.

To see the second equality, clearly 
\[
\min_{\tau\succeq_{\text{MPS}}\,\tau_{P_1},...,\tau_{P_m}}\int  v(\mu) \tau( d\mu) \leq \min_{\tau\in \cap_{j=1}^m \mathcal{E}_v(P_j)} \sum_{\mu\in E_v} \tau(\mu)v(\mu).
\]
To establish the reverse inequality, consider any $\tau \succsim_{MPS} \tau_{P_1} , \ldots , \tau_{P_m}$.   Lemma~\ref{lem:split} implies that any posterior $\mu \notin E_v$ can be split into nearby extremal beliefs without changing the value.  Hence, there exists $\tilde{\tau} \succsim_{MPS}\tau$ with $\tilde{\tau} \in \Delta(E_v)$ such that 
\[
\sum_{\mu \in E_v} v(\mu) \tilde{\tau}(\mu)  = \intop v(\mu) \tau(d\mu).
\]
Moreover by transitivity, $\tilde{\tau} \succsim_{MPS} \tau_{P_1}, \ldots , \tau_{P_m} $ and so, $\tilde{\tau} \in \bigcap_{j = 1}^{m} \mathcal{E}_v(P_j)$.  Thus, the reverse inequality follows.
\end{proof}

\begin{lem}\label{lem:extreme-point-epi-graph}
Let $V = \left\{ (\mu, w) \in epi(v) : w \leq 0, \mu \in \Delta(\Theta) \right\}$.  Then for every $(\mu, w) \in ext(V)$, $w = v(\mu)$ and $\mu \in E_v$.  
\end{lem}

\begin{proof}
    The first claim is obvious.  For the second claim, if $\mu = \delta_\theta$ for some $\theta \in \Theta$, we are done.  So let us assume that $\mu \neq \delta_\theta$.  

    In this case, note that $w < 0$.  Otherwise, $(\mu,0)$ can be represented as a convex combination of $\{(\delta_\theta,v(\delta_\theta))\}_{\theta\in\Theta}$, which contradicts the assumption that $(\mu, w) \in ext(V)$.  Now consider $(\mu_1, w_1), (\mu_2, w_2) \in epi(v)$ and some $\alpha \in (0,1)$ such that 
\[
\alpha (\mu_1, w_1) + (1 - \alpha) (\mu_2, w_2) = (\mu, w).
\]
Letting $\mu_i(\beta) = \beta \mu_i + (1 - \beta) \mu$ and $w_i(\beta)  = \beta w_i + (1 - \beta) w $, we have for any $\beta \in (0,1)$, $(\mu_i(\beta), w_i(\beta)) \in epi(v)$ and 
\[
\alpha \big(\mu_1(\beta),w_1(\beta)\big)  + (1 -\alpha) \big(\mu_2(\beta),w_2(\beta)\big) = (\mu, w).
\]
Moreover, for $\beta$ sufficiently small, $w_i(\beta)< 0$, and so $(\mu_i(\beta, w_i(\beta)) \in V$ for each $i=1,2$.  Because $(\mu, w) \in ext(V)$, $(\mu_1(\beta), w_1(\beta)) = (\mu_2(\beta), w_2(\beta))$.  This however, implies that $(\mu_1, w_1) = (\mu_2, w_2)$.  Hence $(\mu, w)$ is an extreme point of $epi(v)$ and therefore, $\mu \in E_v$.  
\end{proof}

\begin{lem}[Splitting Lemma]\label{lem:split}
For every $\tau $, there exists some $\tilde{\tau} \succsim_{MPS} \tau$ such that $\tilde{\tau} \in \Delta(E_v)$ with
\[
\sum_{\mu \in E_v} v(\mu) \tilde{\tau}(\mu) = \intop v(\mu) \tau(d\mu).
\]
\end{lem}

\begin{proof}
Consider any $\mu \in \Delta(\Theta)$.  Since $V$ is convex and compact, there exists a finite collection, $(\mu_1, w_1), \ldots , (\mu_k, w_k) \in ext(V)$ and some $\lambda_1, \ldots , \lambda_k\in [0,1]$ with $\sum_{\ell=1}^k \lambda_\ell = 1$ such that 
\[
(\mu, v(\mu)) = \sum_{\ell = 1}^{k} \lambda_\ell (\mu_\ell, w_\ell).
\]
By \cref{lem:extreme-point-epi-graph}, $\mu_\ell \in E_v$ and $w_\ell = v(\mu_\ell)$.  Hence, for every $\mu \in \Delta(\Theta)$, we have shown the existence of $ \lambda ( \cdot \mid \mu) \in \Delta(E_v)$ such that 
$(\mu, v(\mu)) = \sum_{\nu \in E_v} \lambda(\nu \mid \mu) (\nu, v(\nu)).$

Now define $\tilde{\tau} \in \Delta(E_v)$ as $\tilde{\tau}( \nu) = \intop \lambda(\nu \mid \mu) \tau(d\mu)$.  Clearly, $\tilde{\tau} \succsim_{MPS} \tau$ and 
\[
\intop v(\nu) \tilde{\tau}(d\nu) = \intop \left( \sum_{\nu \in E_v} v(\nu) \lambda(\nu \mid \mu) \right) \tau(d\mu) = \intop v(\mu) \tau(d\mu).
\]
\end{proof}

\subsubsection{Support Functions}

For any convex compact set $S\subseteq \mathbb{R}^{K_v}_+$, we define
\begin{equation}\label{eq:defn:support-function}
  \begin{aligned}
    h_S: \mathbb{R}^{K_v}_-&\rightarrow \mathbb{R}\\
         \boldsymbol{v}&\mapsto \min_{ \boldsymbol{\tau}\in S}\boldsymbol{\tau}\cdot \boldsymbol{v}
\end{aligned}  
\end{equation}
as the support function of $S$.\footnote{Recall that we have  normalized the value functions to be non-positive, so the domain here is restricted to the negative quadrant.}

The next lemma shows that for any collection of convex sets $\{S_j\}_{j=1}^m$, the support function of their intersection, $h_{\cap_{j=1}^m S_j}(\lambda)$, can be characterized by the concavification of the pointwise maximum of the support functions of  $S_j$.
\begin{lem}\label{lem:intersection-lemma}
Suppose $\{S_j\}_{j=1}^m$ is a collection of non-empty, compact, convex sets in $\mathbb{R}^{K_v}_+$ such that for any $s_j\in S_j$ (for $j=1,...,m$), their meet $\wedge_{j=1}^m s_j$ lies in $\cap_j S_j$. Then for any $\bv\in \mathbb{R}^{K_v}_-$,
\[h_{\cap_{j=1}^m S_j}(\bv)=conc(\max\{h_{S_1},...,h_{S_m}\})(\bv).\]
\end{lem}
\begin{proof}
To simplify notation, let us denote $H(\cdot) = \max \left\{ h_{S_1}(\cdot), \ldots , h_{S_m}(\cdot)\right\}$.    Clearly, $h_{\cap_{j=1}^m S_j}(\cdot)$ is concave as it is the minimum of linear functions. In addition, $h_{\cap_{j=1}^m S_j}(\cdot)\geq H(\cdot)$, because any solution to $\min_{ \boldsymbol{\tau}\in \cap_{j=1}^m S_j}\boldsymbol{\tau}\cdot \boldsymbol{v}$ is feasible in $\min_{ \boldsymbol{\tau}\in  S_j}\boldsymbol{\tau}\cdot \boldsymbol{v}$. Therefore, $h_{\cap_{j=1}^m S_j}(\cdot)\geq conc (H)(\cdot)$.

For the reverse inequality, recall that $conc(H)(\cdot)$ is the point-wise infimum of all affine functions, $g$, that dominate $H$ pointwise, i.e. $g(\bv) \geq H(\bv)$ for all $\bv \in \mathbb{R}^{K_v}_-$.  Thus, it suffices to show that for any such function, we also have $g(\cdot)\geq h_{\cap_{j=1}^m S_j}(\cdot)$.

Consider any affine function $g(\bv)=\lambda\cdot \bv+c$ that dominates $H$ pointwise.  Note that for any $\beta > 0$, $\bv \in \mathbb{R}^{K_v}_-$ and $j$, $
0 \leq \frac{g(\beta \bv) - h_{S_j}(\beta \bv)}{\beta} = ( \lambda \cdot \bv - h_{S_j}(\bv)) + \frac{c}{\beta}.$ Taking $\beta \rightarrow 0$, we first note that $ c \geq 0$.  Moreover, taking $\beta \rightarrow \infty$, we have $\lambda \cdot  \bv \geq h_{S_j}(\bv)$.  To summarize, we have shown that for all $\bv \in \mathbb{R}^{K_v}_-$:
\begin{equation}\label{eqn:f}
g(\bv) \geq \lambda \cdot \bv \geq h_{S_j}(\bv) = h_{S_j^-}(\bv),
\end{equation}
where $S_j^- \equiv S_j - \mathbb{R}_+^{K_v}$.  Since $\lambda \cdot \bv \geq h_{S_j^-}(\bv)$ for all $\bv \in \mathbb{R}^{K_v}_-$, we must have $\lambda \in S_j^-$. To see why, suppose toward a contradiction that $\lambda\notin S_j^-$. By the separating hyperplane theorem, there must exist $r\in \mathbb{R}^{K_v}_-$ and $\alpha\in \mathbb{R}$ such that $\lambda\cdot r <\alpha<x\cdot r$ for all $x\in S_j^-$, which implies $\lambda \cdot r<h_{S_j^-} (r) $, a contradiction.

Hence, we have 
\[
g(\bv)\geq \lambda\cdot \bv\geq \inf_{s \in \cap_{j=1}^m S_j^-}s\cdot \bv=h_{\cap_{j=1}^m S_j}(\bv) \]
where the last equality holds because, by assumption, for any $s\in \cap_{j=1}^m S_j^-$, there exists $s_j\in S_j$ such that $s\leq s_j$ for $j=1,...,m$; therefore,  $s\leq \wedge_{j} s_j \in \cap_j S_j$, which implies that $\inf_{s \in \cap_{j=1}^m S_j^-}s\cdot \bv=\min_{s\in \cap_{j=1}^m S_j}s\cdot \bv=h_{\cap_{j=1}^m S_j}(\bv)$. This concludes the proof.
\end{proof}

We will apply \cref{lem:intersection-lemma} to $\cap_{j=1}^m \mathcal{E}_v(P_j)$, and the following lemma ensures that $\{\mathcal{E}_v(P_j)\}_{j=1}^m$ satisfies the desired properties.
\begin{lem}\label{lem:pointwise-minimum-property}
    For any $\btau_j\in \mathcal{E}_v(P_j)$ (for $j=1,...,m$), their meet $\wedge_{j=1}^m \btau_j$ lies in $\cap_{j=1}^m \mathcal{E}_v(P_j)$.
\end{lem}
\begin{proof}
    Note that if $\btau_j\in \mathcal{E}_v(P_j)$, and $\btau'$ satisfies $0\leq \btau'\leq \btau_j$, then $\btau'\in \mathcal{E}_v(P_j)$. This is because  any probability mass on nondegenerate extremal beliefs can be split into degenerate extremal beliefs, which yields a mean-preserving spread that retains its support on $E_v$. Since $0\leq \wedge_{j=1}^m \btau_j \leq \btau_j$ for all $j$, it follows that $\wedge_{j=1}^m \btau_j\in  \cap_{j=1}^m \mathcal{E}_v(P_j)$.
\end{proof}

The next lemma characterizes $W(P;\bv)$ using support functions.
\begin{lem}\label{lem:W-as-support-function}
    For any $P$ and  $\bv\in \mathbb{R}^{K_v}_-$,
    \[W(P;\bv)=h_{\mathcal{E}_v(P)}(\bv).\]
\end{lem}

\begin{proof}
Recall that $W(P; \bv) = V(P; conv(\bv))$.  By Lemma \ref{lem:value-extremal-beliefs}, 
\[
W(P; \bv) = V(P; conv(\bv)) = \min_{\tau \in \mathcal{E}_v(P)} \sum_{\mu \in E_v} conv(\bv) (\mu) \btau(\mu).
\]

Now let $\hat{E}_v=\{\mu\in E_v|conv(\bv)(\mu)=\bv(\mu)\}$, which is non-empty as it always contains $\{\delta_\theta\}_{\theta\in\Theta}$. Let $\hat{\mathcal{E}}_v(P)=\{\tau\in \Delta(\Delta\Theta)|\tau\succeq_{\text{MPS}}\tau^P \text{ and supp}(\tau)\subseteq \hat{E}_v\}$.

We will establish the following:
\begin{align*}
    \min_{\btau\in {\mathcal{E}}_v(P)} \sum_{\mu\in E_v} conv(\bv)(\mu)\btau(\mu)&=\min_{\btau\in \hat{\mathcal{E}}_v(P)} \sum_{\mu\in \hat{E}_v} conv(\bv)(\mu)\btau(\mu)\\
    &=\min_{\btau\in \hat{\mathcal{E}}_v(P)} \sum_{\mu\in \hat{E}_v} \bv(\mu)\btau(\mu)\\
    &\geq \min_{\btau\in {\mathcal{E}}_v(P)} \sum_{\mu\in E_v} \bv(\mu)\btau(\mu) \\
    &\geq \min_{\btau\in {\mathcal{E}}_v(P)} \sum_{\mu\in E_v} conv(\bv)(\mu)\btau(\mu);
\end{align*}
hence, all hold with equality. The first equality follows from the fact that any $\mu$ such that $\bv(\mu)>conv(\bv)(\mu)$ can be split into $\mu_i\in \hat{E}_v$ without changing the value. The second equality is due to $conv(\bv)(\mu)=\bv(\mu)$ for all $\mu\in \hat{E}_v$. The second to last inequality holds because $\hat{\mathcal{E}}_v(P) \subseteq \mathcal{E}_v(P)$, and the last inequality holds because $conv(\bv) \leq \bv$.

Therefore, $W(P;\bv)=\min_{\btau\in {\mathcal{E}}_v(P)} \sum_{\mu\in E_v} \bv(\mu)\btau(\mu)=h_{\mathcal{E}_v(P)}(\bv)$, which concludes the proof.

\end{proof}

\subsubsection{Robustly Optimal Value}
By \cref{lem:value-extremal-beliefs}, \cref{lem:intersection-lemma}, \cref{lem:pointwise-minimum-property}, and \cref{lem:W-as-support-function}, we have 
\begin{align*}
    V(P_1,...,P_m;v)&=\min_{\btau\in \cap_{j=1}^m \mathcal{E}_v(P_j)} \btau\cdot \bv_v\\
    &=h_{\cap_{j=1}^m\mathcal{E}_v(P_j)}(\bv_v)\\
    & =conc(\max\{h_{\mathcal{E}_v(P_1)},...,h_{\mathcal{E}_v(P_m)}\})(\bv_v)\\
    & =conc(\max\{W(P_1;\cdot),...,W(P_m;\cdot)\})(\bv_v)
\end{align*}
and \autoref{thm:concavification} follows.

\subsection{Proof of \cref{cor:general-robust-strategy}}\label{app:proof-general-state-strategy}

We first establish that $(A, u)$ dominates $\bigoplus_{j = 1}^{m} (A_j, u_j)$. Suppose, toward a contradiction, that $\po(\bigoplus_{j = 1}^{m} (A_j, u_j))\nsubseteq \po(A,u)$. Then there is some $r\in \po(\bigoplus_{j = 1}^{m} (A_j, u_j))$ with $r\notin \po(A,u)$. By the separating  hyperplane theorem, there exists $h\in \mathbb{R}^{|\Theta|}_+\slash\{\boldsymbol{0}\}$ such that $h\cdot r>h\cdot x$ for all $x\in \po(A,u)$. Normalizing $h$ so that $h\in \Delta (\Theta)$, it follows that there exists $\alpha_j\in A_j$, $j=1,...,m$  such that
	\[\sum_\theta h(\theta)\sum_j u_j(\theta,\alpha_j)> \sum_\theta h(\theta) u(\theta,\alpha)\quad\text{ for all }\alpha\in \Delta A.\]
This means that $\sum_{j=1}^m \lambda_j conv(\bv_j)(h)> v(h)=conv(\bv_v)(h)=conv(\sum_{j = 1}^{m} \lambda_j^* \bv_j^*)(h)\geq \sum_{j=1}^m \lambda_j conv(\bv_j^*)(h)$, a contradiction.
	
Next, we show that for any dominating map, $f: \bigoplus_{j = 1}^{m} (A_j, u_j) \rightarrow \Delta(A)$, the corresponding strategy $\sigma^*_f(y_1, \ldots , y_m) = f(\sigma_1^*(y_1), \ldots , \sigma_m^*(y_m))$ is a robustly optimal strategy.  For any $P\in \mathcal{J}(P_1,...,P_m)$, the agent's expected payoff when playing $\sigma^*_f$ is 
		\begin{align*}
			\sum_\theta \sum_{y_1,...,y_m}P(y_1,...,y_m|\theta)u(\theta,\sigma^*_f(y_1,...,y_m))&\geq \sum_\theta \sum_{y_1,...,y_m}P(y_1,...,y_m|\theta)\left[\sum_{j=1}^m u_j(\theta,\sigma_j^*(y_j))\right] \\
			&=\sum_{j=1}^m \sum_\theta\sum_{y_j} P_j(y_j|\theta)u_j(\theta,\sigma_j^*(y_j))\\
			&=\sum_{j=1}^m  \lambda_j^* W(P_j;\bv_j^*)\\
			&=V(P_1, \ldots , P_m; \bv_v).
		\end{align*}
Therefore, $\sigma^*_f$ is a robustly optimal strategy.

\subsection{Proof of \cref{cor:bound}}\label{app:proof:cor:bound}
Let $f(\cdot ):=\max\{h_{\mathcal{E}_v(P_1)},...,h_{\mathcal{E}_v(P_m)}\}(\cdot )$. By \cref{thm:concavification}, \[V(P_1,...,P_m;v)=conc(f)(\bv_v)=\max\{r|(\bv_v,r)\in co(G(f))\},\] where $G(f)=\{(\bv,f(\bv))|\bv\in \mathbb{R}_-^{K_v}\}$ is the graph of $f$.\footnote{We can take max instead of sup in the definition because $f$ has a closed graph.} Therefore, there exists $\lambda_i\geq 0$, $\sum_i \lambda_i=1$, and $\bv_i\in  \mathbb{R}_-^{K_v}$ such that $\sum_{i}\lambda_i \bv_i=\bv_v$ and $V(P_1,...,P_m;v)=\sum_i \lambda_i f(\bv_i)$. 

Let $||\bv||=|\bv \cdot \boldsymbol{1}|=-\bv \cdot \boldsymbol{1}$ denote the L1 norm of a vector $\bv\in \mathbb{R}_-^{K_v}$, and  $\boldsymbol{V}=\{\bv\in \mathbb{R}_-^{K_v}\mid ||\bv||=||\bv_v||\}$ denote the set of vectors that have the same norm as $\bv_v$. Let $\hat{\lambda}_i:=\lambda_i \frac{||\bv_i||}{||\bv_v||}$ and $\hat{\bv}_i:=\frac{||\bv_v||}{||\bv_i||}\bv_i$. Note that $f$ is a positive homogeneous function; that is, $f(\alpha \bv)=\alpha f(\bv)$ for any $\alpha\geq  0$ and $\bv\in \mathbb{R}_-^{K_v}$. Therefore, $conc(f)(\bv_v)=\sum_i \lambda_i f(\bv_i)=\sum_i \lambda_i \frac{||\bv_i||}{||\bv_v||} f(\frac{||\bv_v||}{||\bv_i||}\bv_i) =\sum_i \hat{\lambda}_i  f(\hat{\bv}_i)$. Since $\hat{\bv}_i\in \boldsymbol{V}$ and $\sum_i \hat{\lambda}_i=1$, it follows that $(\bv_v,conc(f)(\bv_v))$ is in the convex hull of $G_{\boldsymbol{V}}(f):= \{(\bv,f(\bv))|\bv\in \boldsymbol{{V}}\}$. Therefore, the concavification with domain on $\boldsymbol{V}$ obtains the same value; that is, $conc(f)(\bv_v)=\max\{r|(\bv_v,r)\in co(G_{\boldsymbol{V}}(f))\}$. Since $\left(\bv_v,conc(f)(\bv_v)\right)$ lies on the boundary of $co(G_{\boldsymbol{V}}(f))$, there exists a supporting hyperplane crossing $\left(\bv_v,conc(f)(\bv_v)\right)$. In addition, $(\hat{\bv}_i,f(\hat{\bv}_i))$ must all lie on this supporting hyperplane. Since the supporting hyperplane is of dimension $|K_v|-1$, by Carath\'eodory's theorem, there exists an index set $I$ and some $\tilde{\lambda}_i\geq 0$, $\sum_{i\in I} \tilde{\lambda}_i=1$ such that $|I|\leq |{K_v}|$ and $(\hat{\bv}_v,conc(f)(\hat{\bv}_v))=\sum_{i\in I}\tilde{\lambda}_i(\hat{\bv}_i,f(\hat{\bv}_i))$. 

For each $\hat{\bv}_i$, there exists a $j$ such that $f(\hat{\bv}_i)=h_{\mathcal{E}_v(P_j)}(\hat{\bv}_i)$; we select such a $j$ and  denote it by $j(i)$. Let ${J}=\{j|j=j(i)\text{ for some }i\}$. Note that $|{J}|\leq |K_v|$, and that $V(\{P_j\}_{j\in J};v)=conc(\max\{h_{\mathcal{E}_v(P_j)};j\in J\})\geq \sum_{i\in I} \tilde{\lambda}_i h_{\mathcal{E}_v(P_{j(i)})}(\hat{\bv}_i)=\sum_{i\in I} \tilde{\lambda}_i f(\hat{\bv}_i)=conc(f)(\bv_v)=V(P_1,...,P_m;v)$, and the corollary follows.

\subsection{Proof of \cref{cor:never-best-experiment}} \label{app:proof:cor:safe-to-ignore}

If $V(P_\ell;v)\leq \max_{j\neq \ell} V(P_j;v)$ for all $v$, we have 
\[
\max\{W(P_1;\cdot) ,...,W(P_m;\cdot)\}=\max\{W(P_1;\cdot) ,...,W(P_{\ell-1};\cdot ),W(P_{\ell+1};\cdot),...,W(P_m;\cdot)\}.
\]
The corollary then follows immediately from \cref{thm:concavification}.

\subsection{Proof of \cref{cor:aggregate-expert-opinion}} \label{sec:unanimity-proof}
Before proving the result, we prove a lemma connecting the best-response to a signal in the original problem with the best responses in the decomposed problems.

\begin{lem}
Suppose that $a_i=a_j^*(y_j;A,u)$. Then, using only information source $P_j$, a best response to signal $y_j$ for the decomposed decision problem $(A_\ell,u_\ell)$ is $a_\ell=1$ for $\ell=1,\ldots,i-1$ and $a_\ell=0$ for $\ell=i,\ldots,n-1$.
\end{lem}

\begin{proof}
    Let $\mu\in \Delta(\Theta)$ be the belief that comes from observing $y_j$ from $P_j$, so $a_i\in \argmax_{a\in A} \sum_\theta \mu(\theta)\rho(\theta,a)$. By \cref{lemma: single-peaked} in \cref{proof: binary-recommendation}, for $\ell=1,\ldots,i-1$, we have $\sum_\theta \mu(\theta)\rho(\theta,a_{\ell+1})\geq \sum_\theta \mu(\theta)\rho(\theta,a_{\ell})$, whereas for $\ell=i,\ldots,n-1$, we have $\sum_\theta \mu(\theta)\rho(\theta,a_\ell)\geq \sum_\theta \mu(\theta)\rho(\theta,a_{\ell+1})$. This gives the result by the definition of $(A_\ell, u_\ell)$.
\end{proof}

In particular, using the belief order, it follows that if $a_j^*(y_j;A,u)\succeq a_i$, then $a_j^*(y_j;A_\ell,u_\ell)=1$ for $\ell=1,\ldots, i-1$; likewise, if $a_k \succeq a_j^*(y_j;A,u)$, then $a_j^*(y_j;A_\ell,u_\ell)=0$ for $\ell=k,\ldots, n-1$.

Now let $\mathbf{y}$ be a signal realization for all information sources and 
\[
a_{i} = \min_j a_j^*(y_j; A,u) \qquad a_{k} = \max_j a_j^*(y_j; A,u),
\]
where the minimum and maximum are with respect to the belief order. Then it follows that, for every $j$,

\[a_j^*(y_j;A_\ell,u_\ell)=\begin{cases}
    1 & \text{ if } \ell\leq i-1\\
    0\text{ or } 1 & \text{ if }  i\leq \ell \leq k-1\\
    0 &\text{ if } k\leq \ell.
\end{cases}\]
Moreover, if $a_j^*(y_j;A_\ell,u_\ell)=a_\ell\in \{0,1\}$ for all $j$, then there must be a pure best-source strategy $\sigma_\ell$ for $(A_\ell,u_\ell)$ with $\sigma_\ell(\mathbf{y})=a_\ell$. Hence there are best-source strategies $\sigma_\ell$ with $\sigma_\ell(\mathbf{y})=1$ for $\ell\leq i-1$ and $\sigma_\ell(\mathbf{y})=0$ for $k\leq \ell$.

Now let $f : \prod_{\ell = 1}^{n - 1} A_\ell \rightarrow \Delta(A)$ be a dominating map that always assigns actions that are not weakly dominated by any other mixed actions. By \cref{thm:binary-general}, $\sigma^*_f(\mathbf{y}) :=  f(\sigma_1(\mathbf{y}), \ldots, \sigma_{n - 1}(\mathbf{y}))$ is a robustly optimal strategy for $(A, u)$. Letting $a_\ell= \sigma_\ell(\mathbf{y})$, we have 
\begin{align*}
    u\left(\theta,f[\sigma_1(\mathbf{y}), \ldots, \sigma_{n - 1}(\mathbf{y})]\right)&\geq \sum_{\ell=1}^{i-1} u_\ell(\theta,1)+ \sum_{\ell=i}^{k-1} u_\ell(\theta,\sigma_\ell(\mathbf{y}))+ \sum_{\ell=k}^{n-1} u_\ell(\theta,0)\\
    &=\sum_{\ell=1}^{i-1} [u(\theta,a_{\ell+1})-u(\theta,a_{\ell})]+ \sum_{\ell=i}^{k-1} u_\ell(\theta,\sigma_\ell(\mathbf{y}))+ 0\\
    &= u(\theta,a_i)+ \sum_{\ell=i}^{k-1} u_\ell(\theta,\sigma_\ell(\mathbf{y})).
\end{align*}
Notice that, given our normalization, $u_\ell(\theta_1,1)>0$ and $u_\ell(\theta_2,1)<0$ for every $\ell$. Therefore, the equality above implies 
\[
u\left(\theta_1,f\big(\sigma_1(\mathbf{y}), \ldots, \sigma_{n - 1}(\mathbf{y})\big)\right)\geq u(\theta_1,a_i)
\]
and 
\[
u\left(\theta_2,f\big(\sigma_1(\mathbf{y}), \ldots, \sigma_{n - 1}(\mathbf{y})\big)\right)\geq u(\theta_2,a_i)+\sum_{\ell=i}^{k-1}u_\ell(\theta_2,1)=u(\theta_2,a_k).
\]

We now prove another lemma.

\begin{lem}\label{lem:belief_order}
    Let $\alpha \in \Delta(A)$ be not weakly dominated by any other mixed strategy. If $u(\theta_1,\alpha)>u(\theta_1,a_i)$ then $\alpha \succeq a_i$. If $u(\theta_2,\alpha)>u(\theta_2,a_k)$, then $\alpha \preceq a_k$.
\end{lem}
\begin{proof}
    By definition of the belief order, showing that $\alpha \succeq a_i$ is equivalent to showing that, whenever $\alpha$ is optimal for a belief $\eta$, $a_i$ is optimal for a belief $\mu$ and $\mu(\theta_1)>\eta(\theta_1)$ then $\alpha$ is optimal for $\mu$ and $a_i$ is optimal for $\eta$. So let $a_i$ be optimal for $\mu$ and $\alpha$ be optimal for $\eta$ (since both $a_i$ and $\alpha$ are not weakly dominated, they must be a best response to some belief). Then we have
    \vspace{-0.1in}
    \begin{align*}
        \mu(\theta_1)[\rho(\theta_1,a_i)-\rho(\theta_1,\alpha)]+(1-\mu(\theta_1))[\rho(\theta_2,a_i)-\rho(\theta_2,\alpha)]\geq 0\\
        \eta(\theta_1)[\rho(\theta_1,\alpha)-\rho(\theta_1,a_i)]+(1-\eta(\theta_1))[\rho(\theta_2,\alpha)-\rho(\theta_2,a_i)]\geq 0.
    \end{align*}
Combining these inequalities, we get
\[
(\eta(\theta_1)-\mu(\theta_1))[\rho(\theta_1,\alpha)-\rho(\theta_1,a_i)]+(\mu(\theta_1)-\eta(\theta_1))[\rho(\theta_2,\alpha)-\rho(\theta_2,a_i)]\geq 0.
\]
Since $a_i$ is not weakly dominated and $\rho(\theta_1,\alpha)>\rho(\theta_1,a_i)$, we have that $\rho(\theta_2,\alpha)<\rho(\theta_2,a_i)$. For the inequality above to hold, it is necessary that $\mu(\theta_1)\leq \eta(\theta_1)$. This means that the condition $\mu(\theta_1)>\eta(\theta_1)$ never holds when $a_i$ is optimal for $\mu$ and $\alpha$ is optimal for $\eta$, so $\alpha \succeq a_i$.

The proof that ``if $u(\theta_2,\alpha)>u(\theta_2,a_k)$, then $\alpha \preceq a_k$'' is analogous.
\end{proof}

Now fix a signal profile $\mathbf{y}$ and let $\alpha= f\big(\sigma_1(\mathbf{y}), \ldots, \sigma_{n - 1}(\mathbf{y})\big)$. Then $\alpha$ is not weakly dominated and $u\left(\theta_1,\alpha\right)\geq u(\theta_1,a_i)$. If this inequality holds with equality, it must hold with equality for $\theta_2$ as well, otherwise either $a_i$ or  $\alpha$ would be weakly dominated (in that case, $\alpha \succeq a_i$ holds trivially). In the case where the inequality holds strictly, we can use \cref{lem:belief_order} to conclude that $\alpha \succeq a_i$. Similarly, since $u(\theta_2,\alpha)\geq u(\theta_2,a_k)$, we can conclude that $\alpha \preceq a_k$. Since we can do this for every $\mathbf{y}$, \cref{cor:aggregate-expert-opinion} follows.

\newpage
\section{Online Appendix}

\subsection{Ignoring the Best Standalone Information Source}\label{sec:ignore-best-standalone-information-source}

We provide an example illustrating that even the best standalone information source may be ignored. 

\begin{ex}\label{ex:third-symptom}
Consider the CDC example with two treatments in the introduction. Suppose in addition to the Cough and Fever, now we have a third informative symptom, Headache, whose relationship to the diseases is given in \cref{table:headache}.
\begin{table}[htp]
\centering
\begin{tabular}{|c|c|c|}
\hline
 & $+$ & $-$ \\
\hline
Covid & $0.72$ & $0.28$ \\
\hline
Flu & $0.28$ & $0.72$ \\
\hline
\end{tabular}
\caption{Headache symptom}
\label{table:headache}
\end{table}

Note that in either treatment 1 or treatment 2, when using the Headache symptom alone, the agent can achieve a value of $\frac{1}{2}(0.72\times 30+0.28\times (-20))=8$. This means that Headache is the best standalone information source, because it yields a value of $8+8=16$, which is greater than $15$, the value of using either the Cough or Fever symptom alone. However, this symptom is never the best information source for either treatment 1 or treatment 2, as the value it yields is lower than $8.5$, the value achieved by using the Cough symptom for treatment 1 or using the Fever symptom for treatment 2. Thus, despite being the best standalone information source for the overall decision problem, Headache can be ignored because it fails to be the best information source for any single subproblem.

\end{ex}

\subsection{Necessity of Combining Information with Three States}\label{sec:necessity-of-combining-information-three-state}

We provide an example showing that, even in binary-action decision problems, an agent may need to combine multiple information sources when there are three or more states.

\begin{ex}\label{ex:threestaets}
 Suppose that there are three states $\theta_1, \theta_2, \theta_3$.  The marginal experiments are both binary with respective signals $x_1, x_2$ and $y_1, y_2$, as given by \autoref{table-ex-threestate}.
\begin{table}[htp]
\centering
\makebox[0pt][c]{\parbox{1\textwidth}{
    \begin{minipage}[b]{0.6\hsize}\centering

\begin{tabular}{ccc}
\multicolumn{3}{c}{$P_X$}                                                                  \\ \hline
\multicolumn{1}{|c|}{$P_X(x|\theta)$}           & \multicolumn{1}{c|}{$x_1$} & \multicolumn{1}{c|}{$x_2$} \\ \hline
\multicolumn{1}{|c|}{$\theta_1$} & \multicolumn{1}{c|}{1}     & \multicolumn{1}{c|}{0}     \\ \hline
\multicolumn{1}{|c|}{$\theta_2$} & \multicolumn{1}{c|}{1}     & \multicolumn{1}{c|}{0}     \\ \hline
\multicolumn{1}{|c|}{$\theta_3$} & \multicolumn{1}{c|}{0}     & \multicolumn{1}{c|}{1}     \\ \hline
\end{tabular}
    \end{minipage}
    \hspace{0in}
    \begin{minipage}[b]{0.1\hsize}\centering
\begin{tabular}{ccc}
\multicolumn{3}{c}{$P_Y$}                                                                  \\ \hline
\multicolumn{1}{|c|}{$P_Y(y|\theta)$}           & \multicolumn{1}{c|}{$y_1$} & \multicolumn{1}{c|}{$y_2$} \\ \hline
\multicolumn{1}{|c|}{$\theta_1$} & \multicolumn{1}{c|}{1}     & \multicolumn{1}{c|}{0}     \\ \hline
\multicolumn{1}{|c|}{$\theta_2$} & \multicolumn{1}{c|}{0}     & \multicolumn{1}{c|}{1}     \\ \hline
\multicolumn{1}{|c|}{$\theta_3$} & \multicolumn{1}{c|}{0}     & \multicolumn{1}{c|}{1}     \\ \hline
\end{tabular}
    \end{minipage}
}}
\caption{}
\label{table-ex-threestate}
\end{table}

Intuitively, experiment $P_X$ indicates whether the state is $\theta_3$ or not and experiment $P_Y$ indicates whether the state is $\theta_1$ or not. 
	Note that upon observing both experiments, the agent obtains perfect information, and so in any decision problem, the agent achieves the perfect information payoff.

Let $A = \{0, 1\}$ and suppose that the utilities are as follows:\footnote{Recall that the payoffs here have been weighted by the prior: $u(\theta,a)=\mu_0(\theta)\rho(\theta,a)$.}
	\begin{align*}
	u(\theta,a = 1) &= \mathbf{1}\left( \theta \in \{\theta_1, \theta_3\} \right) - 0.9\cdot \mathbf{1}\left( \theta = \theta_2\right), \\
	u(\theta,a = 0) &= 0.
	\end{align*}
By using only one information source (either $P_X$ or $P_Y$), $a=1$ is the unique optimal action for any signal realization. Therefore, the agent's expected payoff is $1-0.9+1=1.1$. By contrast, when using both information sources, the full information payoff is $1+0+1=2$.

\end{ex}

\subsection{Necessity of Randomization}\label{sec:necessity-of-randomization}
We provide an example that highlights the necessity of mixed actions for robustly optimal strategies. There are two states, Covid ($\theta=1$) and Flu ($\theta=2$), and two information sources. These two information sources are the same as the ones in the introductory example, reproduced below for the ease of reference.

\begin{table}[htp]
\begin{minipage}{0.5\textwidth}
\centering
\begin{tabular}{|c|c|c|}
\hline
 & $+$ & $-$ \\
\hline
Covid & $0.9$ & $0.1$ \\
\hline
Flu & $0.5$ & $0.5$ \\
\hline
\end{tabular}
\captionof*{table}{Cough}
\end{minipage}
\hfill 
\begin{minipage}{0.5\textwidth}
\centering
\begin{tabular}{|c|c|c|}
\hline
 & $+$ & $-$ \\
\hline
Covid & $0.5$ & $0.5$ \\
\hline
Flu & $0.9$ & $0.1$ \\
\hline
\end{tabular}
\captionof*{table}{Fever}
\end{minipage}
\label{TBA}
\end{table}

We consider the following three-action decision problem: $u(\cdot,a_1)=(0,0)$, $u(\cdot,a_2)=(3,-2)$, $u(\cdot,a_3)=(7,-8)$. The corresponding binary decomposition consists of two subproblems with payoffs $u_1(\cdot, 1)=(3,-2)$, $u_2(\cdot,1)=(4,-6)$, as illustrated in \cref{figure-mixed}.

\begin{figure}[htp]
\centering
	\begin{tikzpicture}[domain=0:3, scale=6, ultra thick]    
	
    \fill[orange!10] (0.2,0.1)--(1,0.1)--(1,0.2)--(0.6,0.8)--(0.3,1)--(0.2,1)--(0.2,0.1);
    
	\draw[<->] (0.3,0.1)node[left]{$\theta=2$}--(0.3,1)node[above]{(0,0)}--(1.2,1)node[above]{$\theta=1$};

    \draw[red,<-] (0.6,0.8)--(0.3,1);
  	\draw[red] (0.43,0.92)node[right]{$u_1(\cdot,1)$};
    
    \draw[red,<-] (1,0.2)--(0.6,0.8);
    \draw[red] (0.75,0.65)node[right]{$u_2(\cdot,1)$};
    
    \draw[dashed,red,->] (0.3,1)--(0.7,0.4);
    \draw[blue] (0.71,0.39)circle (0.1pt);

    \filldraw[blue] (1,0.2)circle (0.1pt)node[left]{$u(\cdot,a_3)$};
    \filldraw[blue] (0.6,0.8)circle (0.1pt)node[left,yshift=-5]{$u(\cdot,a_2)$};
    \filldraw[blue] (0.3,1)circle (0.1pt)node[below left]{$u(\cdot,a_1)$};
    
    \draw[orange] (0.47,0.4)node{$\po(A,u)$};
	
	\end{tikzpicture}   	
    \caption{Binary decomposition of the three-action decision problem}
    \label{figure-mixed}
\end{figure}

In the binary decomposition, the action $(0,1)$ corresponds to a payoff vector $(4,-6)$, which is not dominated by any pure action, but is dominated by a convex combination of $a_2$ and $a_3$. This means that the construction of robustly optimal strategies in \cref{thm:binary-general} requires the decision maker to mix under some signal realization. We will further establish that \emph{every} robustly optimal strategy must involve randomization in this example.

In light of \cref{cor:value-supremum}, we can derive the robustly optimal value through the Blackwell supremum of $\{P_{Cough},P_{Fever}\}$. \cref{fig:example:Blackwell-sup} geometrically characterizes the induced feasible set of the Blackwell supremum, $\overline{P}(P_{Cough},P_{Fever})$. It can be verified that the joint experiment $P^*$ given in \cref{table-blackwell-sup-example} induces this feasible set, which means $P^*$ is the Blackwell supremum.

\begin{figure}[ht]
	\centering
	\begin{tikzpicture}[domain=0:3, scale=5, ultra thick]

	\draw[<->] (0,1.1) node[above]{$\lambda(\cdot|\theta=2)$} -- (0,0)node[below]{\footnotesize(0,0)}-- (1.2,0) node[below,xshift=-5]{$\lambda(\cdot|\theta=1)$};
	  
	\draw[orange] (0,0)--(0.1,0.5)--(1,1)--(0.9,0.5)--(0,0);
	\fill[orange!20] (0,0)--(0.1,0.5)--(1,1)--(0.9,0.5)--(0,0);
    \draw[orange] (0.95,0.7)node[right]{$\Lambda_{P_{Cough}}$};

    \draw[cyan,thin] (0,0)--(0.5,0.9)--(1,1)--(0.5,0.1)--(0,0);
	\fill[pattern=vertical lines, pattern color=cyan] (0,0)--(0.5,0.9)--(1,1)--(0.5,0.1)--(0,0);
	\draw[cyan] (0.7,0.95)node[above]{$\Lambda_{P_{Fever}}$};

    \draw[gray,dotted,line width=2] (0,0)--(0.1,0.5)--(0.5,0.9)--(1,1)--(0.9,0.5)--(0.5,0.1)--(0,0);
	\draw[gray] (0.3,0.7)node[above,xshift=-6]{$\Lambda_{\overline{P}}$};
    
	\draw (1,1)node[right]{\footnotesize(1,1)};

	\draw (0.1,0.5)node[right,yshift=-3]{\footnotesize$(0.1,0.5)$};
    \draw (0.9,0.5)node[below,xshift=17]{\footnotesize$(0.9,0.5)$};
    \draw (0.5,0.9)node[left,yshift=5]{\footnotesize$(0.5,0.9)$};
	\end{tikzpicture}   
	\caption{Blackwell Supremum of $\{P_{Cough},P_{Fever}\}$}
    \label{fig:example:Blackwell-sup}
\end{figure}

\begin{table}[ht]
\begin{minipage}{0.5\textwidth}
\centering
\begin{tabular}{cccc}
\cline{1-3}
 \multicolumn{1}{|c|}{Covid} &  \multicolumn{1}{|c|}{Fever$_+$} &  \multicolumn{1}{|c|}{Fever$_-$}\\
\cline{1-3}
 \multicolumn{1}{|c|}{Cough$_+$} &  \multicolumn{1}{|c|}{$0.4$} &  \multicolumn{1}{|c|}{$0.5$} &\term{0.9}\\
\cline{1-3}
 \multicolumn{1}{|c|}{Cough$_-$} &  \multicolumn{1}{|c|}{$0.1$} &  \multicolumn{1}{|c|}{$0$} &\term{0.1}\\
\cline{1-3}
&	\blue{0.5}	 & \blue{0.5}	&
\end{tabular}
\end{minipage}
\hfill
\begin{minipage}{0.5\textwidth}
\centering
\begin{tabular}{cccc}
\cline{1-3}
 \multicolumn{1}{|c|}{Flu} &  \multicolumn{1}{|c|}{Fever$_+$} &  \multicolumn{1}{|c|}{Fever$_-$}\\
\cline{1-3}
 \multicolumn{1}{|c|}{Cough$_+$} &  \multicolumn{1}{|c|}{$0.4$} &  \multicolumn{1}{|c|}{$0.1$} &\term{0.5}\\
\cline{1-3}
 \multicolumn{1}{|c|}{Cough$_-$} &  \multicolumn{1}{|c|}{$0.5$} &  \multicolumn{1}{|c|}{$0$} &\term{0.5}\\
\cline{1-3}
&	\blue{0.9}	 & \blue{0.1}	&
\end{tabular}
\end{minipage}
\caption{Blackwell Supremum as Joint Experiment}
\label{table-blackwell-sup-example}
\end{table}

From \cref{eq:nature-minmax}, any robustly optimal strategy $\sigma^*$ must be a best response to $P^*$. Therefore, $\sigma^*(C_+,F_-)=a_3$, $\sigma^*(C_+,F_+)=a_2$, and $\sigma^*(C_-,F_+)=a_1$. Note that $\sigma^*(C_-,F_-)$ is not determined by the best response requirement because it is not in the support of $P^*$. A simple calculation gives us the robustly optimal value $V(P^*)=1.55$.

We will show that $\sigma^*(C_-,F_-)$ must be mixed to guarantee such a value. We do so by showing that if the decision maker plays any of the three pure actions under $(C_-,F_-)$, there exists a joint experiment that leads to a value strictly less than 1.55.

If $\sigma(C_-,F_-)=a_1$, consider the joint experiment given in \cref{table-blackwell-sup-example-2}, which yields a value of $\frac{1}{2}[0.4*7+0.5*3-0.1*8-0.4*2]=1.35<1.55$. If $\sigma(C_-,F_-)=a_2$, the same joint experiment yields a value of $\frac{1}{2}[0.4*7+0.5*3+0.1*3-0.1*8-0.4*2]=1.5<1.55$.

\begin{table}[ht]
\begin{minipage}{0.5\textwidth}
\centering
\begin{tabular}{cccc}
\cline{1-3}
 \multicolumn{1}{|c|}{Covid} &  \multicolumn{1}{|c|}{Fever$_+$} &  \multicolumn{1}{|c|}{Fever$_-$}\\
\cline{1-3}
 \multicolumn{1}{|c|}{Cough$_+$} &  \multicolumn{1}{|c|}{$0.5$} &  \multicolumn{1}{|c|}{$0.4$} &\term{0.9}\\
\cline{1-3}
 \multicolumn{1}{|c|}{Cough$_-$} &  \multicolumn{1}{|c|}{$0$} &  \multicolumn{1}{|c|}{$0.1$} &\term{0.1}\\
\cline{1-3}
&	\blue{0.5}	 & \blue{0.5}	&
\end{tabular}
\end{minipage}
\hfill
\begin{minipage}{0.5\textwidth}
\centering
\begin{tabular}{cccc}
\cline{1-3}
 \multicolumn{1}{|c|}{Flu} &  \multicolumn{1}{|c|}{Fever$_+$} &  \multicolumn{1}{|c|}{Fever$_-$}\\
\cline{1-3}
 \multicolumn{1}{|c|}{Cough$_+$} &  \multicolumn{1}{|c|}{$0.4$} &  \multicolumn{1}{|c|}{$0.1$} &\term{0.5}\\
\cline{1-3}
 \multicolumn{1}{|c|}{Cough$_-$} &  \multicolumn{1}{|c|}{$0.5$} &  \multicolumn{1}{|c|}{$0$} &\term{0.5}\\
\cline{1-3}
&	\blue{0.9}	 & \blue{0.1}	&
\end{tabular}
\end{minipage}
\caption{Nature's alternative choice of joint experiment}
\label{table-blackwell-sup-example-2}
\end{table}

If $\sigma(C_-,F_-)=a_3$, consider the joint experiment given in \cref{table-blackwell-sup-example-3}, which yields a value of $\frac{1}{2}[0.5*7+0.4*3-0.5*2-0.1*8]=1.45<1.55$.

\begin{table}[ht]
\begin{minipage}{0.5\textwidth}
\centering
\begin{tabular}{cccc}
\cline{1-3}
 \multicolumn{1}{|c|}{Covid} &  \multicolumn{1}{|c|}{Fever$_+$} &  \multicolumn{1}{|c|}{Fever$_-$}\\
\cline{1-3}
 \multicolumn{1}{|c|}{Cough$_+$} &  \multicolumn{1}{|c|}{$0.4$} &  \multicolumn{1}{|c|}{$0.5$} &\term{0.9}\\
\cline{1-3}
 \multicolumn{1}{|c|}{Cough$_-$} &  \multicolumn{1}{|c|}{$0.1$} &  \multicolumn{1}{|c|}{$0$} &\term{0.1}\\
\cline{1-3}
&	\blue{0.5}	 & \blue{0.5}	&
\end{tabular}
\end{minipage}
\hfill
\begin{minipage}{0.5\textwidth}
\centering
\begin{tabular}{cccc}
\cline{1-3}
 \multicolumn{1}{|c|}{Flu} &  \multicolumn{1}{|c|}{Fever$_+$} &  \multicolumn{1}{|c|}{Fever$_-$}\\
\cline{1-3}
 \multicolumn{1}{|c|}{Cough$_+$} &  \multicolumn{1}{|c|}{$0.5$} &  \multicolumn{1}{|c|}{$0$} &\term{0.5}\\
\cline{1-3}
 \multicolumn{1}{|c|}{Cough$_-$} &  \multicolumn{1}{|c|}{$0.4$} &  \multicolumn{1}{|c|}{$0.1$} &\term{0.5}\\
\cline{1-3}
&	\blue{0.9}	 & \blue{0.1}	&
\end{tabular}
\end{minipage}
\caption{Nature's alternative choice of joint experiment 2}
\label{table-blackwell-sup-example-3}
\end{table}

\subsection{Convexification and Concavification}\label{sec:convexification-concavification}
This section collects the definitions and properties of convex and concave envelopes used in \cref{section:general state}.

Let $X$ be a subset of a vector space, and $f:X\rightarrow \mathbb{R}$. The \textbf{convexification} of $f$, denoted by $conv(f)$, is the largest convex function defined on $co(X)$ below $f$. That is, for any $z\in co(X)$,
\[conv(f)(z)=\sup\{g(z)|g\in \mathbb{R}^{co(X)}\text{ is convex and }g(x)\leq f(x),\forall x\in X\}.\]
Similarly, we define the \textbf{concavification} of $f$ as
\[conc(f)(z)=\inf\{g(z)|g\in \mathbb{R}^{co(X)}\text{ is concave and }g(x)\geq f(x),\forall x\in X\}.\]
The convexification and concavification are well-defined real functions, as long as the sets in their definitions are non-empty, which will be assumed whenever these concepts are used. Note that $f$ need not be defined on a convex set, but its convexification or concavification will be extended to the convex hull of $f$'s domain.

An equivalent definition of concavification that will be useful in our analysis is $conc(f)(z)=\sup\{y|(z,y)\in co(G(f))\}$, where $G(f)=\{(x,f(x))|x\in X\}$ is the graph of $f$.\footnote{See Chapter IV, Proposition 2.5.1 in \citet {hiriart1996convex}.}

\subsection{Bounding $|K_v|$ by $|\Theta|$ and $|A|$}\label{appendix:bounding-K_v}

\begin{proof}
Recall that $\epi(v)=\{(\mu,w)\in \Delta(\Theta)\times \mathbb{R}\:\vert \:w\geq v(\mu)\}$. Let $n=|\Theta|$. We can represent $\epi(v)$ as a polyhedron in $\mathbb{R}^n$ that is the intersection of $|A|+|\Theta|$ halfspaces, as follows:
\[
\left\{ (\mu_1,\ldots,\mu_{n-1},w )\in \mathbb{R}^{n} \;\middle\vert\;
   \begin{array}{@{}l@{}}
   w\geq \sum_{i=1}^{n-1} \mu_i \rho(\theta_i,a) +(1-\sum_{i=1}^{n-1}\mu_i)\rho(\theta_n,a)\quad \forall a\in A \\
   \mu_i\geq 0 \quad i=1,\ldots,n-1 \\ 
   \mu_1+\cdots+\mu_{n-1}\leqslant 1
   \end{array}
\right\}. 
\]
Here, we simply replaced the set $\Delta(\Theta)$ by its first $n-1$ coordinates; the original element $\mu\in \Delta(\Theta)$ can be recovered by $\mu_n=1-\mu_1-\cdots-\mu_{n-1}$, so this change is inconsequential. In this representation, we have $|A|$ halfspaces corresponding to the constraints $ w\geq \sum_{\theta\in\Theta} \mu(\theta)\rho(\theta,a)$ and $|\Theta|=n$ constraints corresponding to the description of $\Delta(\Theta)$.

 This polyhedron is unbounded. To bound it, we also intersect $\epi(v)$ with an additional halfspace, creating a bounded polytope $B=\epi(v)\cap \{(\mu,w): w\leq \max_{\theta,a} u(\theta,a)+1\}$, which has at most $|A|+|\Theta|+1$ facets.

The Upper Bound Theorem (see Theorem 8.23 in \citet{ziegler2012lectures}) gives an upper bound on the number of facets that a polytope with a given number of vertices can have. Every polytope has a dual polytope (see Section 3.4 in \citet{grunbaum2003convex}), where each vertex corresponds to a facet and each facet corresponds to a vertex. Thus, we can apply the Upper Bound Theorem to the dual of $B$, which implies $B$ can have at most 
 \[\begin{pmatrix}
        |\Theta|+|A|+1- \left\lfloor \frac{|\Theta|+1}{2} \right\rfloor\\
        |A|+1
    \end{pmatrix}
   +\begin{pmatrix}
        |\Theta|+|A|+1- \left\lfloor \frac{|\Theta|+2}{2} \right\rfloor\\
        |A|+1
    \end{pmatrix}\]
    number of vertices. 
    
    These vertices include $\{(\delta_i,v(\delta_i))\}_{i=1}^n$ and  $\{(\delta_i,\max_{\theta,a} u(\theta,a)+1)\}_{i=1}^n$, which means $|K_v|$ can be no more than 
   \[ \begin{pmatrix}
        |\Theta|+|A|+1- \left\lfloor \frac{|\Theta|+1}{2} \right\rfloor\\
        |A|+1
    \end{pmatrix}
   +\begin{pmatrix}
        |\Theta|+|A|+1- \left\lfloor \frac{|\Theta|+2}{2} \right\rfloor\\
        |A|+1
    \end{pmatrix}-2|\Theta|.\] 
\end{proof}

\subsection{Proof of \cref{prop:ex:marginals}} \label{proof:ex:marginals}

\begin{proof}

First observe that the agent's maxmin value is no more than her minmax value:
\begin{align*}
    V(\mathcal{P}_1, \ldots , \mathcal{P}_m)
&\leq  \min_{P\in \mathcal{J}(\mathcal{P}_1,...,\mathcal{P}_m)} \, \max_{\sigma:\mathbf{Y}\rightarrow \Delta (A)} \sum_\theta\sum_{\mathbf{y}}P(\mathbf{y}|\theta)u(\theta,\sigma(\mathbf{y}))
 \intertext{Now in the minmax problem, Nature's choice can be split into first choosing each marginal experiment $P_j\in \mathcal{P}_j$, and then choosing a joint experiment $P\in\mathcal{J}(P_1,...,P_m)$:}
 &=\min_{\substack{P_j\in \mathcal{P}_j\\j =1,...,m}}\min_{P\in \mathcal{J}({P}_1,...,{P}_m)} \, \max_{\sigma:\mathbf{Y}\rightarrow \Delta (A)} \sum_\theta\sum_{\mathbf{y}}P(\mathbf{y}|\theta)u(\theta,\sigma(\mathbf{y}))
  \intertext{And the value of the inner minmax problem is exactly $V(P_1,...,P_m)$, which equals $\max_j V(P_j)$ from \autoref{thm:binarybinary}:}
  &= \min_{\substack{P_j\in \mathcal{P}_j\\j =1,...,m}}\, \max_{j=1,...,m} V(P_j)\\
  &= \max_{j=1,...,m} V(\underline{P_j})
\end{align*}
where $\underline{P}_j\in \argmin_{P_j\in \mathcal{P}_j} V(P_j)$ is a worst experiment among the set $\mathcal{P}_j$ if the agent faces this information source solely. Let $j^*\in \argmax_j V(\underline{P_j})$, and consider the problem where the decision maker faces only a single set of marginal experiments $\mathcal{P}_{j^*}$:
\[
  V(\mathcal{P}_{j^*} )= \max_{\sigma:Y_{j^*}\rightarrow \Delta (A)}\,  \min_{{P_{j^*}}\in \mathcal{P}_{j^*} } \sum_\theta \sum_{y_{j^*}\in Y_{j^*}} P_{j^*}(y_{j^*}|\theta)u(\theta,\sigma(y_j^*)).
\]
Since $\mathcal{P}_{j^*}$ is convex, from the minmax theorem, the value of the problem equals 
\[ V(\mathcal{P}_{j^*} )=\min_{{P_{j^*}}\in \mathcal{P}_{j^*} } \, \max_{\sigma:Y_{j^*}\rightarrow \Delta (A)} \sum_\theta \sum_{y_{j^*}\in Y_{j^*}} P_{j^*}(y_{j^*}|\theta)u(\theta,\sigma(y_j^*))=V(\underline{P_{j^*}}).\]
So there exists a best-source strategy, using only signals from the experiment $P_{j^*}$, that guarantees the robustly optimal value $V(\underline{P_{j^*}})=\max_j V(\underline{P_j})\geq  V(\mathcal{P}_1, \ldots , \mathcal{P}_m)$.
\end{proof}

\subsection{Proof of \cref{prop: binary-recommendation}}\label{proof: binary-recommendation}

\begin{lem}[Single-Peaked Property] \label{lemma: single-peaked}
	Suppose in a binary-state decision problem $(A,u)$, every action is a unique best response to some belief, and actions are ordered as follows
\begin{align*}
 u(\theta_1,a_1) < u(\theta_1,a_2) < \cdots < u(\theta_1,a_n),\\
 u(\theta_2,a_1) > u(\theta_2,a_2) > \cdots > u(\theta_2,a_n).
\end{align*}
Then, for any belief $\mu\in \Delta (\Theta)$,
\[a_i\in \argmax_{a\in A} \sum_\theta \mu(\theta)u(\theta,a)\]
implies that for $k>j\geq i$, 
\[\sum_\theta \mu(\theta)u(\theta,a_j)\geq \sum_\theta \ \mu(\theta)u(\theta,a_{k})\] 
and for $k<j\leq i$,
\[\sum_\theta \mu(\theta)u(\theta,a_j)\geq \sum_\theta \ \mu(\theta)u(\theta,a_{k}).\]

\end{lem}
\begin{proof}

	Suppose by contradiction that there exists $k>j\geq i$, such that 
	\[\mu(\theta_1)u(\theta_1,a_j)+\mu(\theta_2)u(\theta_2,a_j)<\mu(\theta_1)u(\theta_1,a_k)+\mu(\theta_2)u(\theta_2,a_k).\]
	Rearranging, we obtain
	\[\mu(\theta_2)[u(\theta_2,a_j)-u(\theta_2,a_k)]<\mu(\theta_1)[u(\theta_1,a_k)-u(\theta_1,a_j)].\]
Given that $u(\theta_2,a_j)-u(\theta_2,a_k)>0$ and $u(\theta_1,a_k)-u(\theta_1,a_j)>0$, the inequality above still holds if we raise $\mu(\theta_1)$ (and consequently lower $\mu(\theta_2)$). That is, for any $\mu'\in \Delta(\Theta)$ such that $\mu'(\theta_1)\geq \mu(\theta_1)$, we have 
	\begin{equation}\label{eq: lemma-single-peaked-1}
		\mu'(\theta_1)u(\theta_1,a_j)+\mu'(\theta_2)u(\theta_2,a_j)<\mu'(\theta_1)u(\theta_1,a_k)+\mu'(\theta_2)u(\theta_2,a_k).
	\end{equation}
	
	Since $a_i$ is, by definition, 	a best response for $\mu$,
 \[\mu(\theta_1)u(\theta_1,a_j)+\mu(\theta_2)u(\theta_2,a_j)\leq \mu(\theta_1)u(\theta_1,a_i)+\mu(\theta_2)u(\theta_2,a_i).\]
Since $u(\theta_1,a_j)\geq u(\theta_1,a_i)$ and $u(\theta_2,a_j)\leq u(\theta_2,a_i)$, for any $\mu'\in \Delta(\Theta)$ such that $\mu'(\theta_1)\leq \mu(\theta_1)$, we have 
\begin{equation}\label{eq: lemma-single-peaked-2}
	\mu'(\theta_1)u(\theta_1,a_j)+\mu'(\theta_2)u(\theta_2,a_j)\leq \mu'(\theta_1)u(\theta_1,a_i)+\mu'(\theta_2)u(\theta_2,a_i)
\end{equation}
The inequalities \eqref{eq: lemma-single-peaked-1} and \eqref{eq: lemma-single-peaked-2} together imply that $a_j$ is never a unique best response to any belief, contradicting our assumption.

The case where $k<j\leq i$ follows from a similar argument.
\end{proof}

\begin{lem}\label{lemma: binary-recommendation} Let $(A_\ell,u_\ell)$ be a subproblem in a binary decomposition of $(A,u)$ and let $R_{P_j}$ be a recommendation information structure with respect to $(A,u)$. Then 
    \[V(P_j;(A_\ell,u_\ell))=V(R_{P_j};(A_\ell,u_\ell)).\]
\end{lem}
\begin{proof}
Recall that $P_j$ Blackwell dominates $R_{P_j}$, so $V(P_j;(A_\ell,u_\ell))\geq V(R_{P_j};(A_\ell,u_\ell))$. We prove the result by constructing a recommendation information structure $R_{P_j}^\ell$ for $(A_\ell,u_\ell)$ and showing that $V(R_{P_j};(A_\ell,u_\ell))\geq V(R_{P_j}^\ell;(A_\ell,u_\ell))=V(P_j;(A_\ell,u_\ell))$.

Recall that $R_{P_j}$ is defined using a garbling of $P_j$ given by $\sigma^*:Y_j\rightarrow A $ that satisfies, for each $y_j$ in the support,
\[\sigma^* (y_j)\in \argmax_{a\in A} \sum_{\theta} P_j(y_j|\theta) u(\theta,a).\]
From \cref{lemma: single-peaked}, if $a_i\in \argmax_{a\in A} \sum_{\theta} P_j(y_j|\theta) u(\theta,a)$, for all $i \leq \ell\leq n-1$, $\sum_{\theta} P_j(y_j|\theta) u(\theta,a_\ell)\geq \sum_{\theta} P_j(y_j|\theta) u(\theta,a_{\ell+1})$, and for all $2 \leq \ell\leq i$, $\sum_{\theta} P_j(y_j|\theta) u(\theta,a_\ell)\geq \sum_{\theta} P_j(y_j|\theta) u(\theta,a_{\ell-1})$. This means that, if $R_{P_j}$ recommends action $a_i$, then $0\in A_\ell$ is optimal for the subproblems with $i\leqslant \ell$ and $1\in A_\ell$ is optimal for the subproblems with $i>\ell$. Now let $\gamma_\ell:  A \rightarrow \{0,1\}$ be the garbling defined by
\[\gamma_\ell (a_i)=\begin{cases}
0 &\text{ if }i\leq \ell\\
1 &\text{ if }i>\ell. 	
\end{cases}
\]
By construction, for each $y_i$ in the support,
\[\gamma_\ell(\sigma^*(y_j))\in \argmax_{a\in A_\ell } \sum_{\theta} P_j(y_j|\theta) u_\ell (\theta,a),\]
so the experiment $R_{P_j}^\ell$, induced by garbling $P_j$ according to $\gamma_\ell\circ \sigma^*:Y_j\rightarrow A$, is a recommendation information structure for the decision problem $(A_\ell,u_\ell)$, so $V(R_{P_j}^\ell;(A_\ell,u_\ell))=V(P_j;(A_\ell,u_\ell))$. Moreover, by construction, $R_{P_j}$ Blackwell dominates $R_{P_j}^\ell$, so $ V(R_{P_j};(A_\ell,u_\ell))\geq V(R_{P_j}^\ell;(A_\ell,u_\ell))$.

\end{proof}

\begin{proof}[Proof of \cref{prop: binary-recommendation}]
Let $\bigoplus_{\ell=1}^{k} (A_\ell,u_\ell)$ be a binary decomposition of $(A,u)$. From \cref{thm:binary-general} and \cref{lemma: binary-recommendation},
\begin{align*}
    V(P_1,\ldots,P_m;(A,u))&=\sum_{l=1}^k \max_{j=1,...,m} V(P_j;(A_\ell,u_\ell)) \\
                         &=\sum_{l=1}^k \max_{j=1,...,m} V(R_{P_j};(A_\ell,u_\ell))\\
                         &= V(R_{P_1},\ldots,R_{P_m};(A,u)).
\end{align*}
\end{proof}

\subsection{Proof of the Second Part of  \cref{thm:binarybinary}} \label{sec:proof-uniqueness}

Consider any binary-state binary-action decision problem, denoted by $(A^{bi},u^{bi})$. Without loss of generality, suppose $P_1$ is the unique best marginal information source:  $V(P_1;(A^{bi},u^{bi}))>V(P_j;(A^{bi},u^{bi}))$ for $j\neq 1$.

\subsubsection{Payoff Sets}
Recall that as in \cref{sec:decomposition}, any binary-state decision problem $(A,u)$ induces a payoff polyhedron:
\[\po(A,u)=co\{u(\cdot,a):a\in A\}-\mathbb{R}_+^{2},\]
which captures the feasible payoff vectors that can be achieved by the decision maker when allowing for free disposal of utils.  Such a polyhedron is upper bounded, convex, closed, and has a finite number of extreme points. 

\begin{defn}
A non-empty subset $D\subseteq \mathbb{R}^{|\Theta|}$	is a \textbf{payoff set} if $D$ is upper bounded, convex, closed, and has a finite number of extreme points.
\end{defn}

For any payoff set $D$, we define the robustly optimal value in a manner similar to that for decision problems:
\[W(P_1,...,P_m;D)=\max_{t:\mathbf{Y}\rightarrow D} \, \min _{P\in \mathcal{J}(P_1,...,P_m)}\sum_{\mathbf{y}} \mathbf{P}(\mathbf{y})\cdot t(\mathbf{y})\]
where $\mathbf{P}(\mathbf{y})=P(\mathbf{y}|\cdot )\in \mathbb{R}^{|\Theta|}$ denotes the vector corresponding to the probability of $\mathbf{y}$ in each state.

If only a single experiment $P:\Theta\rightarrow \Delta (Y)$ is considered ($m=1$),
\[W(P;D)=\max_{t:Y\rightarrow D} \sum_{y} \mathbf{P}({y})\cdot t(y). \]
Note that the value for a payoff set is tightly connected to the value of the decision problem that induces it. Specifically, we have $V(P_1,...,P_m;(A,u))=W(P_1,...,P_m;\po(A,u))$.

\bigskip
Similar to $V$, $W$ also has the property that having access to more experiments can be no worse than having access to just one experiment.

\begin{lem}\label{lem:one-marginal}
For any payoff set $D$,
\[W(P_1,...,P_m;D)\geq 	W(P_1;D) \]
\end{lem}
\begin{proof}
Suppose $t_1:Y_1\rightarrow D$ is the solution to $W(P_1;D)$. Define $\tilde t:Y_1\times \cdots \times Y_m\rightarrow D$  as $\tilde t(y_1,...,y_m)=t_1(y_1)$, and we have
\[W(P_1,...,P_m;D)\geq \min _{P\in \mathcal{J}(P_1,...,P_m)}\sum_{\mathbf{y}} \mathbf{P}(\mathbf{y})\cdot \tilde t(\mathbf{y})=\sum_{y_1}P_1(y_1)\cdot t_1(y_1)=W(P_1;D).\]
\end{proof}

Another  useful property of $W$ is its separability with respect to payoff sets, analogous to the separability of $V$ with respect to separable decision problems.

\begin{lem}Let $C,D\subseteq \mathbb{R}^2$ be two payoff sets, and $C+D$ denote their Minkowski sum. Then
\[W(P;C+D)=W(P;C)+W(P;D).\]
\end{lem}
\begin{proof}
	Let $t_C^*$ and $t^*_D$ be solutions to $W(P;C)$ and $W(P;D)$, respectively. Define $t:Y\rightarrow C+D$ to be $t(y)=t^*_C(y)+t^*_D(y)$. Then 
	\begin{align*}
	W(P;C+D)&\geq \sum_{{y}} \mathbf{P}({y})\cdot t(y)\\
		 &= \sum_{{y}} \mathbf{P}({y})\cdot (t^*_C(y)+t^*_D(y))\\
		&=\sum_y\mathbf{P}({y})\cdot t^*_C(y)+\sum_y\mathbf{P}({y})\cdot t^*_D(y)\\
		&=W(P;C)+W(P;D).
	\end{align*}
Conversely, let $t^*$ be a solution to $W(P;C+D)$. Then for any $y$, there exists $c_y\in C$ and $d_y\in D$ such that $t^*(y)=c_y+d_y$. Define $t_C(y)=c_y$ and $t_D(y)=d_y$, then 
\begin{align*}
	W(P;C)+W(P;D)&\geq \sum_y\mathbf{P}({y})\cdot t_C(y)+\sum_y\mathbf{P}({y})\cdot t_D(y) \\
	&=\sum_y\mathbf{P}({y})\cdot  t^*(y)\\
	&=W(P;C+D).
\end{align*}
\end{proof}

\subsubsection{Binary-Action Decision Problems}

Now we return to the binary action decision problem $(A^{bi},u^{bi})$. The payoff polyhedron corresponding to $(A^{bi},u^{bi})$ can be represented as intersection of three halfspaces:
\[\po(A^{bi},u^{bi})=\bigcap_{\beta \in \mathcal{B}_{(A^{bi},u^{bi})}} \{v\in \mathbb{R}^2:\beta\cdot v\leq k_\beta\} \]
where $\mathcal{B}_{(A^{bi},u^{bi})}= \{e_1,e_2,\beta^*\}$ with $e_1=(1,0)$, $e_2=(0,1)$, and $\beta^*\in \mathbb{R}^2_{++}$ denotes the set of normal vectors, and $k_{e_1}=\max_{a\in A}u(\theta=1,a)$, $k_{e_2}=\max_{a\in A}u(\theta=2,a)$, and $k_{\beta^*}\in \mathbb{R}$. This is visualized in \cref{fig:polyhedron}.

The set of normal vectors, $\mathcal{B}_{(A^{bi},u^{bi})}$, depends on the binary action decision problem, where $\beta^*$ is proportional to the belief at which the decision maker is indifferent between the two actions. Since the decision problem $(A^{bi},u^{bi})$ is fixed, for notational simplicity, we will henceforth omit the dependence of $\mathcal{B}$ on $(A^{bi},u^{bi})$.

\begin{figure}[htp]
    \centering
    \begin{tikzpicture}[domain=0:3, scale=4.8, ultra thick]    
        \fill[orange!10] (-0.2,-0.1)--(-0.2,0.8)--(0.4,0.8)--(0.8,0.4)--(0.8,-0.1); 
    	\draw[<->] (0,1.1) node[left]{$\theta=2$} --(0,0)--(1.1,0)node[below]{$\theta=1$};

     	\filldraw[blue] (0.4,0.8)circle (0.1pt)node[right,]{$u(\cdot,a_2)$};
        \filldraw[blue] (0.8,0.4)circle (0.1pt)node[right]{$u(\cdot,a_1)$};
    
        \draw[red,->] (0.2,0.81)--(0.2,1); 
        \draw[red,->] (0.81,0.2)--(1,0.2);
        \draw[red] (0.2,0.9)node[left]{$e_2$}  (0.9,0.2)node[below]{$e_1$};
        \draw[red,->] (0.6,0.6)--(0.74,0.74)node[below,yshift=-5,xshift=-2]{$\beta^*$};
        
        \draw[orange] (0.47,0.4)node{$\po(A,u)$};	
	\end{tikzpicture}   
    \caption{Payoff polyhedron for a binary-state binary-action problem}
    \label{fig:polyhedron}
\end{figure}

We next define payoff sets that have the same shape as the $\po(A^{bi},u^{bi})$.

\begin{defn}
    A payoff set $D\subset \mathbb{R}^2$ is a $\mathcal{B}$-shape polyhedron if 
    \[D=\bigcap_{\beta \in \mathcal{B}} \{v\in \mathbb{R}^2:\beta\cdot v\leq k_\beta\} \]
    for some constants $\{k_\beta\}_{\beta\in\mathcal{B}}\in \mathbb{R}$.
\end{defn}

Note that the constraint $\beta^*\cdot v\leq k_{\beta^* }$ may be redundant in a $\mathcal{B}$-shape polyhedron, in which case the polyhedron is an unbounded rectangle. Such a polyhedron can be represented as $\{v:v\leq v^*\}$ for some $v^*\in \mathbb{R}^2$ and corresponds to a single-action decision problem. We call such a $\mathcal{B}$-shape polyhedron \textit{trivial}. 

Clearly, if $D$ is a trivial $\mathcal{B}$-shape polyhedron, $W(P;D)=W(P';D)$ for any $P,P'$. The next lemma shows that for any non-trivial $\mathcal{B}$-shape polyhedron, the relative value  of experiments under $(A^{bi},u^{bi})$ is preserved.

\begin{lem}\label{lem:strict-higher-value}
	If $D$ is a non-trivial $\mathcal{B}-$shape polyhedron, then $W(P_1;D)>\max_{j\neq 1} W(P_j;D)$.
\end{lem}
\begin{proof}
Any non-trivial $\mathcal{B}$-shape polyhedron $D$ has two extreme points, denoted by $ex(D)_1$ and $ex(D)_2$. See \cref{fig:Lambda-polyhedron-extreme} for an illustration.

\begin{figure}[htp]
    \centering
    \begin{tikzpicture}[domain=0:3, scale=4.8, ultra thick]    
        \fill[orange!10] (-0.2,-0.1)--(-0.2,0.8)--(0.4,0.8)--(0.8,0.4)--(0.8,-0.1); 
    	\draw[<->] (0,1.1) node[left]{$\theta=2$} --(0,0)--(1.1,0)node[below]{$\theta=1$};

     	\filldraw[blue] (0.4,0.8)circle (0.1pt)node[right,]{$ex(D)_2$};
        \filldraw[blue] (0.8,0.4)circle (0.1pt)node[right]{$ex(D)_1$};
    
        \draw[red,->] (0.2,0.81)--(0.2,1); 
        \draw[red,->] (0.81,0.2)--(1,0.2);
        \draw[red] (0.2,0.9)node[left]{$e_2$}  (0.9,0.2)node[below]{$e_1$};
        \draw[red,->] (0.6,0.6)--(0.74,0.74)node[below,yshift=-5,xshift=-2]{$\beta^*$};
        
        \draw[orange] (0.47,0.4)node{$D$};	
	\end{tikzpicture}   
    \caption{Extreme points of a non-trivial $\mathcal{B}$-polyhedron}
    \label{fig:Lambda-polyhedron-extreme}
\end{figure}

The two extreme points are defined by two linear equations: \[\begin{pmatrix}
     e_1 \\
     \beta^*
 \end{pmatrix} v=\begin{pmatrix}
     k_{e_1}\\
     k_{\beta^*}
 \end{pmatrix} \qquad \qquad \begin{pmatrix}
     e_2 \\
     \beta^*
 \end{pmatrix} v=\begin{pmatrix}
     k_{e_2}\\
     k_{\beta^*}
 \end{pmatrix},\]
with the closed-form solutions $ex(D)_1=\begin{pmatrix}
	k_{e_1}\\
	\frac{k_{\beta^*}-\beta_1^*k_{e_1}}{\beta_2^*}
\end{pmatrix}$ and $ex(D)_2=\begin{pmatrix}
	\frac{k_{\beta^*}-\beta_2^*k_{e_2}}{\beta_1^*}\\
		k_{e_2}\\
\end{pmatrix}$. A useful observation is that $(ex(D)_2-ex(D)_1)=(k_{e_2}\beta_2^*+k_{e_1}\beta_1^*-k_{\beta^*})\begin{pmatrix}
    -\frac{1}{\beta_1^*}\\
    \frac{1}{\beta_2^*}
\end{pmatrix}.$ That is, $\beta^*$ determines the direction of the vector $(ex(D)_2-ex(D)_1)$, and the constant terms $k_\beta$ only affect the scalar multiplier. Moreover, the multiplier $(k_{e_2}\beta_2^*+k_{e_1}\beta_1^*-k_{\beta^*})>0$, because $(k_{e_1},k_{e_2})\notin D$.

For any $\mathcal{B}$-shape polyhedron $D$, and any $P_j$,
\begin{align*}
    W(P_j;D)&=\max_{t_j:Y_j\rightarrow D} \sum_{y_j\in Y_j} \mathbf{P}_j(y_j)\cdot t_j(y_j)
\end{align*}
Since the objective function is linear and the  extreme points of $D$ are $ex(D)_1$ and $ex(D)_2$, a solution to the problem is 
\[t_j^*(y_j)=\begin{cases}
    ex(D)_1 & \text{ if } \mathbf{P}_j(y_j)\cdot  \begin{pmatrix}
    -\frac{1}{\beta_1^*}\\
    \frac{1}{\beta_2^*}
\end{pmatrix}\leq  0\\
    ex(D)_2 & \text{ if } \mathbf{P}_j(y_j)\cdot  \begin{pmatrix}
    -\frac{1}{\beta_1^*}\\
    \frac{1}{\beta_2^*}
\end{pmatrix}>  0.
\end{cases}\]
For each $P_j$, let $\tilde{Y}_j=\{y\in Y_j:\mathbf{P}_j(y_j)\cdot  \begin{pmatrix}
    -\frac{1}{\beta_1^*}\\
    \frac{1}{\beta_2^*}
\end{pmatrix}\leq 0\}$, and we can rewrite:
\[W(P_j;D)=\sum_{y_j\in \tilde{Y}_j} \mathbf{P}_j(y_j)\cdot ex(D)_1+\sum_{y_j\in Y_j/\tilde{Y}_j} \mathbf{P}_j(y_j)\cdot ex(D)_2.\]
Let $\mathbf{x}_{P_j}=\sum_{y_j\in \tilde{Y}_j} \mathbf{P}_j(y_j)$, then 
\begin{align*}
    W(P_j;D)&=\mathbf{x}_{P_j}\cdot ex(D)_1+(\boldsymbol{1}-\mathbf{x}_{P_j})\cdot ex(D)_2\\
            &=\boldsymbol{1}\cdot ex(D)_2 + \mathbf{x}_{P_j}\cdot (ex(D)_1-ex(D)_2).
\end{align*}
Now consider any $j\neq 1$, we have 
\begin{align*}
      W({P}_1;D)-W(P_j;D)&=(\mathbf{x}_{P_j}-\mathbf{x}_{P_1})\cdot (ex(D)_2-ex(D)_1)\\
      &=(k_{e_2}\beta_2^*+k_{e_1}\beta_1^*-k_{\beta^*})(\mathbf{x}_{P_j}-\mathbf{x}_{P_1})\cdot\begin{pmatrix}
    -\frac{1}{\beta_1^*}\\
    \frac{1}{\beta_2^*}
\end{pmatrix}
\end{align*}
Note that for different non-trivial $\mathcal{B}$-shape polyhedra $D$ (i.e., different parameters $k_{e_1}, k_{e_2},k_{\beta^*})$, the above value differs only by a positive constant factor.  This implies that if $W(P_1;D)-W(P_j;D)>0$ for one non-trivial $\mathcal{B}$-shape polyhedron, the value is also strictly positive for any non-trivial $\mathcal{B}$-shape polyhedron.

Recall that \[W(P_1;\po(A^{bi},u^{bi}))-W(P_j;\po(A^{bi},u^{bi}))=V(P_1;(A^{bi},u^{bi}))-V(P_j;(A^{bi},u^{bi}))>0\]
where $\po(A^{bi},u^{bi})$ is a $\mathcal{B}$-shape polyhedron. Therefore,
\[W(P_1;D)-W(P_j;D)>0,\] 
for any non-trivial $\mathcal{B}$-shape polyhedron.

\end{proof}

\subsubsection{$\mathcal{B}$-cover}

For any payoff set $D$, we define the smallest $\mathcal{B}$-shape polyhedron that covers $D$ as its \textit{$\mathcal{B}$-cover}. See \cref{fig:Lambda-minimal cover} for an illustration.

\begin{defn}
For any payoff set $D$, its \textbf{$\mathcal{B}$-cover} is defined as     \[cov_\mathcal{B} (D) \doteq \bigcap_{\beta\in \mathcal{B}}\{v: \beta\cdot v \leq \rho_{D}(\beta)\},\]
where $\rho_D(\beta)=\sup_{v\in D}\beta\cdot v$ is the support function of $D$.
\end{defn}

\begin{figure}[htp]
	\centering
\subfigure[A payoff set $D$ derived from some three-action decision problem]{
    	\begin{tikzpicture}[domain=0:3, scale=4.8, ultra thick]    
    \fill[orange!10] (-0.2,-0.1)--(-0.2,0.72)--(0.3,0.72)--(0.6,0.6)--(0.72,0.3)--(0.72,-0.1); 
	\draw[<->] (0,1.1) --(0,0)--(1.1,0);

    \draw[orange] (-0.2,0.72)--(0.3,0.72)--(0.6,0.6)--(0.72,0.3)--(0.72,-0.1);

    \filldraw[blue] (0.3,0.72)circle (0.1pt)node[above,xshift=7]{};
    \filldraw[blue] (0.6,0.6)circle (0.1pt)node[right,yshift=-5]{};
    \filldraw[blue] (0.72,0.3)circle (0.1pt)node[right]{};

    \draw[orange] (0.26,0.3)node{$D$};
	
	\end{tikzpicture}   
}
\hspace{0.2in}
\subfigure[The corresponding $\mathcal{B}$-cover $cov_\mathcal{B} (D)$]
{
\begin{tikzpicture}[domain=0:3, scale=4.8, ultra thick]    
    \fill[red!10] (-0.2,-0.1)--(-0.2,0.72)--(0.48,0.72)--(0.72,0.48)--(0.72,-0.1); 
	\draw[<->] (0,1.1) --(0,0)--(1.1,0);

    \draw[red] (-0.2,0.72)--(0.48,0.72)--(0.72,0.48)--(0.72,-0.1);

    \draw[orange] (-0.2,0.72)--(0.3,0.72)--(0.6,0.6)--(0.72,0.3)--(0.72,-0.1);

    \filldraw[blue] (0.3,0.72)circle (0.1pt);
    \filldraw[blue] (0.6,0.6)circle (0.1pt);
    \filldraw[blue] (0.72,0.3)circle (0.1pt);
    
    \draw[red,->] (0.2,0.73)--(0.2,0.9); 
    \draw[red,->] (0.73,0.2)--(0.9,0.2);
    \draw[red] (0.2,0.81)node[left]{$e_2$}  (0.83,0.2)node[below]{$e_1$};
    \draw[red,->] (0.6,0.6)--(0.74,0.74)node[below,yshift=-5,xshift=-2]{$\beta^*$};

    \draw[red] (0.26,0.3)node{$cov_\mathcal{B} (D)$};
	
	\end{tikzpicture}   
}
	
	\caption{}
	\label{fig:Lambda-minimal cover}
\end{figure}

We state a few properties of $\mathcal{B}$-cover that will be useful in our analysis.
\begin{lem}\label{lem:cov-properties}
\begin{enumerate}
    \item (Monotonicity) If $D\subseteq D'$, $cov_\mathcal{B} (D)\subseteq cov_\mathcal{B} (D')$.
    \item (Reflexive) If $D$ is a $\mathcal{B}$-shape polyhedron, $cov_\mathcal{B} (D)=D$.
    \item (Superadditivity) $cov_\mathcal{B} (D+D')\supseteq cov_\mathcal{B} (D)+cov_\mathcal{B} (D')$
    \item (Preserving Triviality) If $cov_{\mathcal{B}}(D)$ is trivial, then there exists a maximum in $D$. That is, $\exists \bar{v}\in D$ such that $v\leq \bar{v}$ for all $v\in D$.
\end{enumerate}
\end{lem}

\begin{proof}
\begin{enumerate}
	\item Since $D\subseteq D'$, $\rho_D(\beta)\leq \rho_{D'}(\beta)$ for all $\beta\in \mathcal{B}$. Therefore,
	\[ \bigcap_{\beta\in \mathcal{B}}\{v: \beta\cdot v \leq \rho_{D}(\beta)\}\subseteq \bigcap_{\beta\in \mathcal{B}}\{v: \beta\cdot v \leq \rho_{D'}(\beta)\}. \]
	\item Clearly $D\subseteq cov_\mathcal{B}(D)$, because for every $v\in D$ and every $\beta\in\mathcal{B}$, $\beta\cdot v \leq \rho_D(\beta)$.

	Now consider any $\mathcal{B}$-shape polyhedron, represented by \[D=\bigcap_{\beta \in \mathcal{B}} \{v\in \mathbb{R}^2:\beta\cdot v\leq k_\beta\} \]
    for some $\{k_\beta\}_{\beta\in\mathcal{B}}\in \mathbb{R}^2$. Note that for all $\beta\in\mathcal{B}$ and $v\in D$, $\beta\cdot v\leq k_\beta$, so we have  $\rho_D(\beta)=\max_{v\in D} \beta\cdot v\leq k_\beta$. Therefore,
    \[cov_\mathcal{B}(D)= \bigcap_{\beta\in \mathcal{B}}\{v: \beta\cdot v \leq \rho_{D}(\beta)\}\subseteq \bigcap_{\beta \in \mathcal{B}} \{v:\beta\cdot v\leq k_\beta\}=D,\]
    which implies $cov_\mathcal{B}(D)=D$.
    
	\item For any $\tilde{v}\in cov_\mathcal{B}(D)+cov_\mathcal{B}(D')$, there exists $v\in cov_\mathcal{B} (D)$ and $v'\in cov_\mathcal{B}(D')$ such that $\tilde{v}=v+v'$. Since $v\in cov_\mathcal{B} (D)$ and $v'\in cov_\mathcal{B}(D')$, we have $\beta\cdot v\leq \rho_D(\beta)$ and $\beta\cdot v'\leq \rho_{D'}(\beta)$ for all $\beta\in\mathcal{B}$. Therefore, for every $\beta\in\mathcal{B}$, $\beta \cdot \tilde{v}=\beta\cdot (v+v')\leq\rho_D(\beta)+\rho_{D'}(\beta)=\rho_{D+D'}(\beta)$, which implies $\tilde{v}\in cov_\mathcal{B}(D+D')$.

    \item If $cov_\mathcal{B}(D)$ is trivial, the constraint $\beta^*\cdot v\leq \rho_D(\beta^*)$ is redundant. That is $\{v:\beta^*\cdot v\leq \rho_D(\beta^*)\}\supseteq \{v:e_1 \cdot v\leq \rho_D(e_1)\}\cap \{v:e_2\cdot v\leq \rho_D(e_2)\}$. 
    
    Let $\bar{v}_1=\max_{v\in D}e_1\cdot v$ and $\bar{v}_2=\max_{v\in D}e_2\cdot v$. We claim that $\bar{v}=(\bar{v}_1,\bar{v}_2)\in D$. Suppose not, then we have $\max_{v\in D} \beta^*\cdot v<\beta^* \cdot \bar{v}$.  However,  $\bar{v}\in \{v:e_1 \cdot v\leq \rho_D(e_1)\}\cap \{v:e_2\cdot v\leq \rho_D(e_2)\}$ but $\bar{v}\notin \{v:\beta^*\cdot v\leq \rho_D(\beta^*)\}$, contradicting to the constraint $\beta^*\cdot v\leq \rho_D(\beta^*)$ being redundant. Thus, $\bar{v}\in D$ and for all $v\in D$, $v\leq \bar{v}$, which concludes the proof.

\end{enumerate}
\end{proof}

\subsubsection{Dominance}
We say a collection of payoff sets $D_1,...,D_k\subseteq \mathbb{R}^{|\Theta|}$ is \textit{dominated} by $D$ if 
\[ D_1+\cdots +D_k\subseteq D.\]

The following observation is immediate:
\begin{lem}\label{lem:dominance}
If $\{D_\ell\}_{\ell=1}^k$ is dominated by $D$,
		\[W(P_1,...,P_m;D)\geq \sum_{\ell=1}^k W(P_1,...,P_m;D_\ell).\] 	
\end{lem}
\begin{proof}
	Let $t_\ell $ be a maxmin strategy to $W(P_1,...,P_m;D_\ell)$. Construct 
	\begin{align*}
	t:\mathbf{Y}&\rightarrow D\\
		y&\mapsto \sum_{\ell=1}^k t_\ell(y).
	\end{align*}
	Then 
	\begin{align*}
		W(P_1,...,P_m;D)&\geq \min_{P\in\mathcal{J}}\sum_{{y}} \mathbf{P}({y})\cdot t(y)\\
		&=\min_{P\in\mathcal{J}}\sum_{{y}} \mathbf{P}({y})\cdot \sum_{\ell=1}^k t_\ell(y)\\
		&=\min_{P\in\mathcal{J}}\sum_{\ell=1}^k \sum_{{y}} \mathbf{P}({y})\cdot  t_\ell(y)\\
		&\geq \sum_{\ell=1}^k   \min_{P\in\mathcal{J}} \sum_{{y}} \mathbf{P}({y})\cdot  t_\ell(y)\\
		&=\sum_{\ell=1}^k W(P_1,...,P_m;D_\ell).
	\end{align*}
\end{proof}

Next, we present the key lemma underlying our uniqueness theorem.
\begin{lem}\label{lem:other-trivial}
Suppose a collection of payoff sets $D_1,...,D_m$ is dominated by a $\mathcal{B}$-shape polyhedron $D$, and satisfies 
\[\sum_{j=1}^m W(P_j;D_j)\geq W(P_1,...,P_m;D).\]
Then $cov_{\mathcal{B}}(D_j)$ must be trivial for all $j\neq 1$.
\end{lem}
\begin{proof}

Since $D_1+\cdots+D_m\subseteq D$, from properties 1 and 2 in \cref{lem:cov-properties}, 
\[cov_\mathcal{B}(D_1+\cdots+D_m)\subseteq cov_\mathcal{B}(D)= D.\]
From property 3 in \cref{lem:cov-properties}, 
\[cov_\mathcal{B}(D_1)+\cdots+cov_\mathcal{B}(D_m)\subseteq cov_\mathcal{B}(D_1+\cdots+D_m),\]
so $cov_\mathcal{B}(D_1)$, $\cdots$, $cov_\mathcal{B}(D_m)$ is also dominated by $D$.

Now suppose by contradiction that $cov_\mathcal{B}(D_j)$ is not trivial for some $j\neq 1$. Then
\begin{align*}
	W(P_1,...,P_m;D)&\geq \sum_{j=1}^m W(P_1,...,P_m;cov_\mathcal{B}(D_j))\\
	&\geq \sum_{j=1}^m W(P_1;cov_\mathcal{B}(D_j)) \\
	&>\sum_{j=1}^m W(P_j;cov_\mathcal{B}(D_j) ) \\
	&\geq \sum_{j=1}^m W(P_j;D_j ) 
\end{align*}
where the first inequality follows from  \cref{lem:dominance}, the second inequality follows from \cref{lem:one-marginal}, the third inequality follows from \cref{lem:strict-higher-value}, and the last inequality follows from $cov_\mathcal{B}(D_j)\supseteq D_j$. Therefore, this contradicts $\sum_{j=1}^m W(P_j;D_j)\geq W(P_1,...,P_m;D)$, and $D_j$ must be trivial for all $j\neq 1$.
\end{proof}

\subsubsection{Common Support of the Blackwell Supremum}
\begin{lem}\label{lem:common-support}
	Suppose $P_j(y_j|\theta)>0$ for all $j,y_j,\theta$, and $P^*\in \mathcal{J}(P_1,...,P_m)$ is a Blackwell supremum of $P_1,...,P_m$. Then, $P^*(\cdot|\theta_1)$ and $P^*(\cdot|\theta_2)$ have common support; that is, for any $y_1,...,y_m$, $P^*(y_1,...,y_m|\theta_1)>0$ if and only if $P^*(y_1,...,y_m|\theta_2)>0$.
\end{lem}
\begin{proof}
If $P^*(\cdot|\theta_1)$ and $P^*(\cdot|\theta_2)$ have different supports, then there exists $\mathbf{y}$ that induces a point-mass belief either on state $\theta_1$ or $\theta_2$. So the corresponding zonotope $\Lambda_{P^*}$ will include either a point $(x,0)$ or $(0,x)$ for some $x>0$.  Since  $P_j(y_j|\theta)>0$ for all $j,y_j,\theta$, none of the zonotopes $\Lambda_{P_j}$ contains such points. From \cref{lem:blackwellsup}, $\Lambda_{P^*}=\co(\Lambda_{P_1}\cup \dots \cup \Lambda_{P_m})$, which also should not contain such points, leading to a contradiction.
\end{proof}

\subsubsection{Proof of the Theorem}

\begin{proof}[Proof of Uniqueness for \cref{thm:binarybinary}]

Let $\sigma^*$ be a robustly optimal strategy in the decision problem $(A^{bi},u^{bi})$. We have\[V(P_1,...,P_m;(A^{bi},u^{bi}))=\min_{P\in \mathcal{J}(P_1,...,P_m)} \sum_\theta \sum_{\mathbf{y}}P(\mathbf{y}|\theta)u^{bi}(\theta,\sigma^*(\mathbf{y})).\]
This is a state-by-state optimal transport problem, and so
the corresponding dual problem is 
\[\max_{\phi_j:\Theta\times Y_j\rightarrow \mathbb{R},\, j=1,...,m} \sum_\theta \sum_j  \sum_{y_j} \phi_j(\theta,y_j)P_j(y_j|\theta) \]
\[s.t. \quad \sum_{j=1}^m \phi_j(\theta,y_j)\leq u^{bi}(\theta,\sigma^*(\mathbf{y}))\quad \forall \theta, \mathbf{y}.\]
Or in vector form:
\[\max_{\phi_j: Y_j\rightarrow \mathbb{R}^{|\Theta|},\, j=1,...,m} \sum_j  \sum_{y_j} \phi_j(y_j)\cdot \mathbf{P}_j(y_j) \]
\[s.t. \quad \sum_{j=1}^m \phi_j(y_j)\leq u^{bi}(\cdot,\sigma^*(\mathbf{y}))\quad \forall \mathbf{y}.\]

Let $\{\phi_j^*\}_{j=1}^m$ be a solution to the dual problem. Define $D_j=co(\{\phi_j^*(y_j)|y_j\in Y_j\})-\mathbb{R}_+^2$ for $j=1,...,m$. Note that $D_1+\cdots +D_m\subseteq \po(A^{bi},u^{bi})$, so $\{D_j\}_{j=1}^m$ is dominated by $\po(A^{bi},u^{bi})$, and satisfies 
\begin{align*}
	\sum_{j=1}^m W(P_j;D_j)&\geq  \sum_{j=1}^m  \sum_{y_j} \phi_j^*(\cdot,y_j)\cdot \mathbf{P_j}(y_j)\\
	&=V(P_1,...,P_m;(A^{bi},u^{bi}))\\
	&=W(P_1,...,P_m;\po(A^{bi},u^{bi})).
\end{align*}
From \cref{lem:other-trivial}, $cov_\mathcal{B}(D_2),...,cov_\mathcal{B}(D_m)$ must be trivial, and property 4 of \cref{lem:cov-properties} implies that  for each $j\neq 1$, there exists $y_j^*$ such that $\phi^*_j( y^*_j)\geq \phi^*_j( y_j)$ for all $y_j$. Now we define $\tilde\phi_j(y_j)=\phi^*_j( y^*_j)$ for all $y_j$ as a constant function. Since $\tilde\phi_j(y_j)\geq \phi_j^*(y_j)$ and $\phi_1^*,\tilde\phi_2,...,\tilde\phi_m$ is feasible in the dual problem, $\phi_1^*,\tilde\phi_2,...,\tilde\phi_m$ is also a solution to the dual problem.

From \cref{lem:relaxed-problem} and \cref{cor:value-supremum}, a Blackwell supremum $P^*\in \mathcal{J}(P_1,...,P_m)$ solves Nature's MinMax Problem. From the minmax theorem, $P^*$ is a solution to 
\[\min_{P\in \mathcal{J}(P_1,...,P_m)} \sum_{\theta,\mathbf{y}} P(\mathbf{y}|\theta)u^{bi}(\theta,\sigma^*(\mathbf{y})).\]
\cref{lem:common-support} implies that $P^*(\cdot|\theta_1)$ and $P^*(\cdot|\theta_2)$ have a common support, which we denote by $\bar{\mathbf{Y}}=\{\mathbf y\in \mathbf Y, \mathbf{\overline{P}}(\mathbf{y})>0\}$.

Now for any $(y_1,\bar{y}_{-1})\in \bar{Y}$, complementary slackness implies 
\[\phi_1^*(\cdot,y_1)+\sum_{j=2}^m \tilde \phi_j(\cdot,\bar{y}_j)= u^{bi}(\cdot,\sigma^*(y_1,\bar{y}_{-1})).\]
For any $(y_1,y_{-1})\in Y$, the dual constraint says
\[\phi_1^*(\cdot,y_1)+\sum_{j=2}^m \tilde \phi_j(\cdot,y_j)\leq  u^{bi}(\cdot,\sigma^*(y_1,{y}_{-1})).\]
Since $\tilde\phi_j$ is constant for $j\geq 2$, the left-hand sides of the two equations above are the same, which implies $u(\cdot,\sigma^*(y_1,\bar y_{-1}))\leq u(\cdot,\sigma^*(y_1, y_{-1}))$. Since $(A^{bi},u^{bi})$ is a non-trivial binary-action decision problem, any two (mixed) actions are either identical or induce payoff vectors that are not ordered. Therefore, $u^{bi}(\cdot,\sigma^*(y_1,\bar y_{-1}))\leq u^{bi}(\cdot,\sigma^*(y_1, y_{-1}))$ implies $\sigma^*(y_1,\bar y_{-1})=\sigma^*(y_1, y_{-1})$. So we have derived that for any $y_1\in Y_1$ and $y_{-1},y_{-1}'\in Y_{-1}$, $\sigma^*(y_1, y_{-1})=\sigma^*(y_1, y_{-1}')$, which concludes the proof.

\end{proof}

\end{document}